%% file: paper.tex
\documentclass[a4paper,accepted=2025-04-03,onecolumn]{quantumarticle}
\pdfoutput=1

\usepackage[]{geometry}
\usepackage[utf8]{inputenc}
\usepackage[T1]{fontenc}
\usepackage{amsmath,amssymb,amsthm}
\usepackage{mathtools}%

\usepackage{graphicx}
\usepackage{url}
\usepackage[section]{placeins}

\usepackage{mathrsfs}
\usepackage{dsfont}
\usepackage{setspace}
\usepackage[export]{adjustbox}
\usepackage{makecell}
\usepackage{enumitem}
\setlist{noitemsep}

\usepackage[rgb]{xcolor}
\definecolor{darkblue}{RGB}{0,0,128}
\definecolor{darkgreen}{RGB}{0,150,0}

\definecolor{tiletwo}{RGB}{184, 184, 184}
\definecolor{tileone}{RGB}{120, 120, 120}
\definecolor{syndtwo}{RGB}{242,100,48}
\definecolor{syndone}{RGB}{0,157,220}
\definecolor{qubit}{RGB}{60,60,60}

\definecolor{tiletop}{RGB}{33,131,128}
\definecolor{tilebot}{RGB}{60,7,72}

\usepackage[pdfusetitle]{hyperref}

\usepackage[
    backend=biber,
    style=alphabetic,
    maxbibnames=99,
    maxalphanames=4,
    backref,
    hyperref=true
]{biblatex}

\DeclareLabelalphaTemplate{
  \labelelement{
    \field[final]{shorthand}
    \field[ifnames=1,strwidth=3,strside=left]{labelname}
    \field[strwidth=1,strside=left]{labelname}
  }
  \labelelement{\field[strwidth=2,strside=right]{year}}
}

\hypersetup{breaklinks, colorlinks, linkcolor=blue, citecolor=darkgreen, filecolor=red, urlcolor=darkblue}
\usepackage{cleveref}

\newtheorem{theorem}{Theorem}[section]
\usepackage{chngcntr}
\usepackage{apptools}
\AtAppendix{\counterwithin{theorem}{section}}
\newtheorem*{theorem*}{Theorem}
\newtheorem{lemma}[theorem]{Lemma}
\newtheorem{corollary}[theorem]{Corollary}
\newtheorem{fact}[theorem]{Fact}
\newtheorem{definition}[theorem]{Definition}

\usepackage{thmtools}
\usepackage{thm-restate}

\usepackage{algorithm}
\usepackage{algpseudocodex}

\usepackage{mystyle}
\graphicspath{{./}{pics/}{tikz/}}

\usepackage{tikz}
\usetikzlibrary{quantikz}
\usetikzlibrary{decorations.markings}
\usetikzlibrary{decorations.pathreplacing}
\usetikzlibrary{calligraphy}
\usetikzlibrary{calc}
\usetikzlibrary{math}

\usetikzlibrary{backgrounds,positioning,shadows,shapes.multipart}
\newsavebox\mmt
\begin{lrbox}{\mmt}
\begin{quantikz}
     \meter{$\ssX$}
\end{quantikz}
\end{lrbox}
\newsavebox\cnottarget
\begin{lrbox}{\cnottarget}
\begin{quantikz}
     \targ{}
\end{quantikz}
\end{lrbox}
\tikzset{on layer/.code={
    \pgfonlayer{#1}\begingroup
    \aftergroup\endpgfonlayer
    \aftergroup\endgroup
  }}

\usepackage[]{authblk}
\usepackage{relsize}

\renewcommand{\epsilon}{\varepsilon}

\renewcommand{\partname}{Direction}

\title{Hierarchical memories: Simulating quantum LDPC codes with local gates}
\author[1]{Christopher A. Pattison}
\author[2,3]{Anirudh Krishna}
\author[1,4]{John Preskill}
\affil[1]{Institute for Quantum Information and Matter, California Institute of Technology, Pasadena, CA 91125}
\affil[2]{Department of Computer Science, Stanford University, Stanford, CA, 94305}
\affil[3]{Stanford Institute for Theoretical Physics, Stanford University, Stanford, CA, 94305}
\affil[4]{AWS Center for Quantum Computing, Pasadena CA 91125}
\date{April 3, 2025}

\begin{document}

\maketitle

\begin{abstract}
Constant-rate low-density parity-check (LDPC) codes are promising candidates for constructing efficient fault-tolerant quantum memories.
However, if physical gates are subject to geometric-locality constraints, it becomes challenging to realize these codes.
In this paper, we construct a new family of $\dsl N,K,D\dsr$ codes, referred to as hierarchical codes, that encode a number of logical qubits $K = \Omega(N/\log(N)^2)$.
The $N$\textsuperscript{th} element $\cH_N$ of this code family is obtained by concatenating a constant-rate quantum LDPC code with a surface code; nearest-neighbor gates in two dimensions are sufficient to implement the syndrome-extraction circuit $\CFT_N$ and achieve a threshold.
Below threshold the logical failure rate vanishes superpolynomially as a function of the distance $D(N)$.
We present a bilayer architecture for implementing $\CFT_N$, and estimate the logical failure rate for this architecture.
Under conservative assumptions, we find that the hierarchical code outperforms the basic encoding where all logical qubits are encoded in the surface code.
\end{abstract}

\section{Introduction}
Quantum error-correcting codes encode quantum information in entangled states over many qubits.
They are defined by a set of operators called stabilizer generators.
Errors can accumulate in the state due to imperfect control and interactions with the environment.
Stabilizer generators can be measured using syndrome-extraction circuits; the outcome of these measurements are called syndromes, classical information used to infer corrections to these errors.
To minimize the probability of corrupting information beyond recovery, it is imperative to minimize the points of failure in the syndrome-extraction circuit.
This can be realized by restricting the number of gates that each qubit interacts with and minimizing the total space-time volume of this circuit.
The extent to which this can be done depends on the choice of error-correcting code and physical constraints.

Syndrome-extraction circuits are the workhorse of quantum memories, devices that can reliably store qubits for some fixed duration.
In this paper, we are concerned with designing memories that can encode a growing number of qubits and simultaneously have a low probability of failure.\footnote{We leave fault-tolerant \emph{computation} for future work.}
We focus on their design when qubits are embedded in a two-dimensional lattice and gates are subject to constraints on geometric locality.

Quantum low-density parity-check (LDPC) codes are natural candidates for constructing quantum memories.
A quantum LDPC code refers to a family $\{\cQ_t\}_t$ of $\{\dsl n(t),k(t),d(t), \Delta_q, \Delta_g \dsr\}$ codes.
This notation means that the $t$\textsuperscript{th} element in the family uses $n(t)$ data qubits to encode $k(t)$ logical qubits and has distance $d(t)$, i.e.\ it is robust to $\floor{(d(t)-1)/2}$ Pauli errors.
We assume that for all $t$, $n(t) > n(t-1)$.
In what follows, we wish to focus on the dependence of $k$ and $d$ as functions of $n$.
In this case, we implicitly parameterize the family using the size $n$ of the code, i.e.\ we use the notation $\{\cQ_n\}_n$, and we say that we are working with a $\dsl n,k(n),d(n), \Delta_q,\Delta_g \dsr$ code family.

A quantum LDPC code is one where, for all codes in the code family, every stabilizer generator only involves at most a constant number $\Delta_g$ of qubits, and each qubit is supported within at most a constant number $\Delta_q$ of stabilizer generators.
Such codes can encode a number of qubits that increases with the code size; simultaneously, the probability of \emph{any} error on the encoded level is suppressed exponentially in the distance $d(n)$.

Furthermore, the syndrome-extraction circuit $C_n$ can be efficient as measured by two figures of merit.
The \emph{depth} of the syndrome-extraction circuit is the number of time steps $\Depth(C_n)$ it takes to implement.
The \emph{width} of the syndrome-extraction circuit is the total number of qubits $\Space(C_n)$ it uses (including ancilla qubits in addition to data qubits).
The size or volume of the circuit is the product of the depth and the width.
Building on a result by Kovalev and Pryadko \cite{kovalev2013fault}, Gottesman \cite{gottesman2014fault} constructed fault-tolerant syndrome-extraction circuits that have volume which is a constant times the volume of the noise-free syndrome-extraction circuit --- there exists a threshold $q$ such that if gates fail with fixed probability $p < q$, the probability of the circuit failing falls exponentially in the distance $d(n)$.

However, realizing this architecture in a 2-dimensional layout is challenging.
It requires high-fidelity gates acting on qubits that may be far apart.
Some architectures might not support such interactions.

It is known that geometric locality severely constrains quantum error-correcting codes in $2$ and $3$ (Euclidean) dimensions.
The most famous codes that are implemented using only geometrically-local gates are surface codes \cite{kitaev2003fault,bravyi1998quantum} and color codes \cite{bombin2006topological,kubica2015universal}.
Seminal results by Bravyi and Terhal \cite{bravyi2009no}, and later by Bravyi, Poulin and Terhal \cite{bravyi2010tradeoffs} showed that these codes are optimal for quantum LDPC codes defined using geometrically-local stabilizers.
Subsequently, it was shown that to implement LDPC codes where the parameters $k$ and $d$ are both strictly better than the surface code, we require a growing amount of long-range connectivity \cite{baspin2021connectivity,baspin2021quantifying}.

When restricted to using only nearest-neighbor gates in $2$ dimensions, Delfosse \emph{et al}.\ \cite{delfosse2021bounds} proved the following tradeoff for syndrome-extraction circuits for constant-rate LDPC codes: \footnote{The bound applies to classes of codes that are called \emph{locally expanding}.
The exact definition of locally-expanding codes is not relevant; the interested reader is pointed to the paper by Delfosse \emph{et al}.\ \cite{delfosse2021bounds}.
For our purposes, it includes some important classes of quantum LDPC codes such as hypergraph product codes \cite{tillich2014quantum} and good quantum LDPC codes \cite{breuckmann2021balanced,panteleev2021asymptotically,leverrier2022quantum,lin2022good}.
}
\begin{equation}
    \label{eq:dbt}
    \Depth(C_n) = \Omega\left(\frac{n}{\sqrt{\Space(C_n)}}\right)~,
\end{equation}
where $\Depth(C_n)$ is the depth of the syndrome-extraction circuit and $\Space(C_n)$ is the total number of qubits, data and ancilla, used in the circuit \footnote{For an explanation of $O(\cdot)$, $\Theta(\cdot)$ and $\Omega(\cdot)$ notation, please refer to Appendix \ref{app:glossary}.}.
In words, this shows that given only nearest-neighbor gates to build a syndrome-extraction circuit for constant-rate LDPC codes, we can choose to minimize either the depth or the width of $C_n$, but cannot do both.

This sets the stage for presenting the main questions we address in this paper: does the family of circuits saturating Equation~\eqref{eq:dbt} still have a threshold?
If not, how do we modify the code and associated circuit to achieve a threshold as efficiently as possible?
How do we construct the most efficient syndrome-extraction circuits given access to gates whose range is more than merely nearest neighbor?
Can we improve on the bound in Equation~\eqref{eq:dbt}?

\subsection*{Our contributions}

This paper is centered around the theme of implementing efficient quantum memories.
Our main result is that our proposal, called a hierarchical code, has a threshold and that it achieves asymptotically better error suppression than the surface code.
As it brings together a few different ideas, we present a short summary of each section and how to navigate the paper.
Although these results build on each other, our presentation is modular --- readers ought to be able to proceed to their section of choice after reading this overview and Section~\ref{sec:background} where we define all the concepts required to formally state our results. (The statements of the main theorems of each section are only presented informally below.)

\paragraph{Section~\ref{sec:design-routing} : Permutation routing on graphs}
Connectivity beyond nearest-neighbor interactions is being explored in many architectures.
There is evidence that some architectures can support gates of range $R$ where $R$ can be large \cite{leung2019deterministic,periwal2021programmable}.
Motivated by these developments, we ask: given nearest-neighbor Clifford gates and $\swapp$ gates of range $R$, can we reduce the depth of the syndrome-extraction circuit for constant-rate LDPC codes?
To this end, we will permute qubits to bring them within range to apply an entangling gate.
This is expressed as a \emph{permutation routing}, a task on a graph $G = (V,E)$ specified by a permutation $\alpha: V \to V$.
In this task, two vertices labeled $u$ and $v$ connected by an edge $(u,v)$ are allowed to exchange labels within each step.
The objective is to ensure that all labels match destinations $\alpha(u)$ while minimizing the total time required.
Permuting vertices in parallel is non-trivial---the paths along which one permutes different pairs can overlap and thereby require more time.
Section~\ref{subsec:routing} reviews a permutation routing algorithm due to Annexstein and Baumslag \cite{annexstein1990unified}.
This algorithm yields a permutation routing on a product of two graphs given permutation routings on each of the input graphs.
In Section~\ref{subsec:permutn-routing-long-range}, we build on this algorithm to permute vertices on an $L \times L$ lattice where two vertices separated by a distance $R$ are connected by an edge using a sparse subgraph.
The main technical result of this section is the following existence result.

\begin{theorem}[Permutation routing]
\label{thm:permutnrouting}
  For $R$ even, there is an efficient construction of a degree-12 graph $G = (V,E)$ whose vertex set $V$ is identified with an $L \times L$ lattice with edges of length at most $R$.
  Given a permutation $\alpha: V \to V$, a permutation routing implementing $\alpha$ can be performed in depth $3L/R + O(\log^2 R)$.
\end{theorem}

While it is itself not the main result of our paper, it will be used in service of proving Theorem~\ref{thm:hierarchical-ppties} which demonstrates the existence of efficient syndrome-extraction circuits given $\swapp$ gates of range up to $R$ (but not all gates need to have length equal to $R$).
This section is entirely technical and only discusses graph properties and permutation routings.

\paragraph{Section~\ref{sec:surface-code} and Section~\ref{sec:ft-asymptotics}: Hierarchical codes \& the bilayer architecture}

Given access to only nearest-neighbor gates, Delfosse \emph{et al}.\ present some evidence \emph{against} the existence of a threshold if one were to permute qubits to bring them within range to perform a $\cnot$ (see Figure~2 of \cite{delfosse2021bounds}).
In particular, in the setting where $\Space(C_n) = \Theta(n)$ and $\Depth(C_n) = \Theta(\sqrt{n})$, it appears too many errors accumulate before we can complete executing the syndrome-extraction circuit.

We circumvent this problem using code concatenation.
We concatenate a constant-rate $\dsl n,k,d, \Delta_q, \Delta_g \dsr$ LDPC code $\{\cQ_n\}$ with a $\dsl d_{\ell}^2, 1, d_{\ell}\dsr$ rotated surface code $\cRS_{\ell}$ to obtain the \emph{hierarchical code} $\{\cH_N\}$ with parameters denoted $\dsl N,K,D\dsr$.
This means that each qubit of the syndrome-extraction circuit for the LDPC code $\cQ_n$, henceforth referred to as the ``outer code'', is itself the logical qubit of a rotated surface code $\cRS_{\ell}$, which we refer to as the ``inner code'' or sometimes as a ``tile.''
As a rotated surface code can suppress errors exponentially in $d_{\ell}$, we can suppress errors long enough to complete syndrome measurements of the outer quantum LDPC code using relatively small inner codes.
The lattice length $d_{\ell}$ of the inner code only scales logarithmically in the size of the outer LDPC code, i.e.\ $d_{\ell} = \Theta(\log(n))$.
Here $\ell$ indexes the qubits in the rotated surface code, $\ell^2 = 2d_{\ell}^2-1$.
Section~\ref{sec:surface-code} is dedicated to the construction of syndrome-extraction circuits $\CFT_N$ corresponding to $\cH_N$.
The hierarchical code family $\{\cH_N\}$ is not LDPC: The stabilizer generators for the outer code act on a number of physical qubits that scales with the size $d_{\ell}$ of the inner code.
However, local operations are sufficient to implement the corresponding syndrome-extraction circuit $\CFT_N$.
The main result of this section is summarized in the following theorem.

\begin{theorem}
    \label{thm:hierarchical-ppties}
    The $\dsl N,K,D \dsr$ hierarchical code $\cH_N$ is constructed by concatenating an outer code, a constant-rate $\dsl n,k,d, \Delta_q, \Delta_g\dsr$ quantum LDPC code $\cQ_n$, and an inner code, a rotated surface code $\cRS_{\ell}$ where $d_{\ell} = \Theta(\log(n))$.
    Let $\rho > 0$ and $\delta \geq 1/2$, such that $k = \rho \cdot n$ and $d = \Theta(n^{\delta})$.
    The code $\cH_N$ has parameters
    \begin{align*}
        K(N) = \Omega\left(\frac{N}{\log^2(N)} \right)~, \qquad
        D(N) = \Omega\left(\frac{N^{\delta}}{\log^{2\delta-1}\left[ N / \log(N) \right]}\right)~.
    \end{align*}
    There exists an explicit and efficient construction of an associated family of syndrome-extraction circuits $\CFT_N$ constructed using only local Clifford operations and $\swapp$ gates of range $R$ such that
    \begin{align*}
        \Space(\CFT_N) = O(N)~, \qquad
        \Depth(\CFT_N) = O\left( \frac{\sqrt{N}}{R} \right)~.
    \end{align*}
\end{theorem}
Our construction works for all values of $\delta > 0$; we choose $\delta \geq 1/2$ to make theorem statements simpler.
Before describing how the circuit $\CFT_N$ is constructed, we motivate why it is interesting---it has a threshold.

We work in a model where errors occur in a stochastic manner.
We declare a logical failure if \emph{any} of the $K$ encoded qubits fail.
More generally, we declare failure if any logical error occurs on the code space.
The main result of Section~\ref{sec:ft-asymptotics} is the following theorem.
\begin{theorem}[Informal]
\label{thm:informal-threshold}
    Consider the $\dsl N, K, D\dsr$ family of hierarchical codes $\cH_N$ and the associated family of syndrome-extraction circuits $\CFT_N$.
    Suppose the outer code $\cQ_n$ has constant rate $k = \rho n$ and distance $d(n) = \Theta(n^{\delta})$.
    If we repeat the syndrome-extraction circuit $\CFT_N$ for $d(n)$ rounds, then there exists a threshold $q \in (0,1]$ corresponding to $\CFT_N$ such that, if each gate fails with fixed probability $0 < p < q$, then the probability of logical failure under minimum-weight decoding, $p_{\cH}(N)$, obeys
    \begin{align*}
        p_{\cH}(N) \leq \exp\left(- c_{\cH} \cdot \frac{N^{\delta}}{\log^{2\delta}(N)} \right)~,
    \end{align*}
    for some positive number $c_{\cH}$ independent of $N$.
\end{theorem}
The theorem is only stated informally here because we have not yet defined the noise model with respect to which this result holds.
We will consider a \emph{locally decaying error model} to account for correlated errors that may occur in a circuit.
This error model is defined in Section~\ref{subsec:noise-and-imperfect-syndromes}.
Section~\ref{sec:ft-asymptotics} is dedicated to a proof of the existence of a threshold.
We build on Gottesman's proof of the existence of a threshold for syndrome-extraction circuits (Theorem 4 of \cite{gottesman2014fault}).
The central idea is the requirement that the probability of failure for a qubit \emph{per round of syndrome extraction}, denoted $\pround$, remains a sufficiently small constant.
This is reviewed in Section~\ref{subsec:ldpc-review}.
Gottesman's result was based on syndrome-extraction circuits for LDPC codes that have constant depth.
As Equation~\eqref{eq:dbt} highlights, this is not possible when subject to locality constraints.
We study the dependence of $\pround$ on the circuit depth in Section~\ref{subsec:evoln-syndrome-ext}.
In Section~\ref{subsec:cft-threshold}, we show that $\ell = \Theta(\log(n))$ is sufficient for $\CFT_N$ to have a threshold.

As we ask to minimize circuit width $\Space(\CFT_N)$ \emph{and} subject the circuit to locality constraints, we pay a price --- in addition to the growing depth, the number of encoded qubits $K(N)$ and distance $D(N)$ are suppressed by polylogarithmic factors in $n$ relative to the outer code $\cQ_n$ which has constant rate and distance $d(n) = \Theta(n^{\delta})$. 
Furthermore, for fixed gate error rates $p \in [0,1]$, the sub-threshold scaling of the logical error rate $p_{\cH}(N)$ of $\{\cH_N\}$ is subexponential, but superpolynomial, in the distance $D(N)$; for any positive constants $\alpha, \beta$, the logical failure probability $p_{\cH}(N)$ vanishes faster than any polynomial function $N^{-\beta}$ but slower than any exponential function $\exp(-\alpha \cdot N)$:
\begin{equation*}
    \frac{p_{\cH}(N)}{N^{-\beta}} \xrightarrow{N \to \infty} 0~, \qquad \frac{p_{\cH}(N)}{\exp(-\alpha \cdot N)} \xrightarrow{N\to \infty} \infty~.
\end{equation*}

Having motivated why we are interested in $\cH_N$, we return to the construction of $\CFT_N$.
In Section~\ref{sec:surface-code}, we propose a novel bilayer architecture to implement it.
We begin the section by presenting the syndrome-extraction circuit $C_n$ for the constant-rate $\dsl n,k,d, \Delta_q, \Delta_g\dsr$ LDPC code.
Physical qubits are arranged in two parallel layers, each a lattice of side length approximately $L = \Theta(\sqrt{n})$.
To obtain the syndrome-extraction circuit $\CFT_N$ for the concatenated code, each of the $\Space$ qubits in $C^{\cQ}_n$ is replaced by a rotated surface code.

In Section~\ref{subsec:inner-code-construction}, we describe how to arrange $\Space = \Space(C_n^{\cQ})$ surface codes $\cRS_{\ell}$ in a bilayer architecture.
Each layer now has side length approximately $2L \ell$ qubits to accommodate the tiles.
An instance of a single layer is shown in Figure~\ref{fig:bilayer-intro} (a).
We assume access to nearest-neighbor physical Clifford operations and $\swapp$ gates of range $R$ within a layer and Clifford operations between adjacent qubits in different layers.
These physical qubits are aggregated into $\dsl d_{\ell}^2, 1, d_{\ell}\dsr$ codes $\cRS_{\ell}$.
See Figure~\ref{fig:bilayer-intro} (b).
There are $2L^2$ tiles in total.
Even though we are only implementing a quantum memory, we still need to understand how to perform a limited set of logical operations on tiles to implement the syndrome-extraction circuit for the outer code.
The advantage of the bilayer architecture is that it allows for transversal $\cnot$ to implement logical $\cnot$.
We propose a new technique to perform logical $\swapp$ operations between tiles.
This yields all required logical Clifford operations between tiles to perform syndrome-extraction for the outer code.

We note that the existence of a threshold does not depend on using the bilayer architecture.
For example, tiles can be arranged in a single layer and Clifford gates can be implemented via lattice surgery \cite{litinski2019game,horsman2012surface}.
For an alternative implementation in the context of measurement-based quantum computation, see \cite{bombin2021logical}.
Although we do not prove it here, it is possible to show that a threshold exists also in this setting using similar techniques.

\begin{figure}[h]
    \centering
    \begin{tikzpicture}
        \node at (-3,2.25) {(a)};
        \node at (0,0) {\includegraphics[width=0.3\columnwidth]{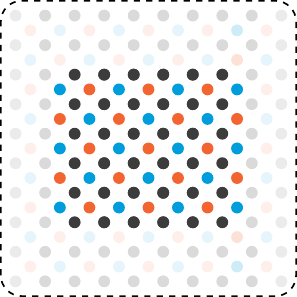}};
        \node at (4.25,2.25) {(b)};
        \node at (7.7,0) {\includegraphics[width=0.5\columnwidth]{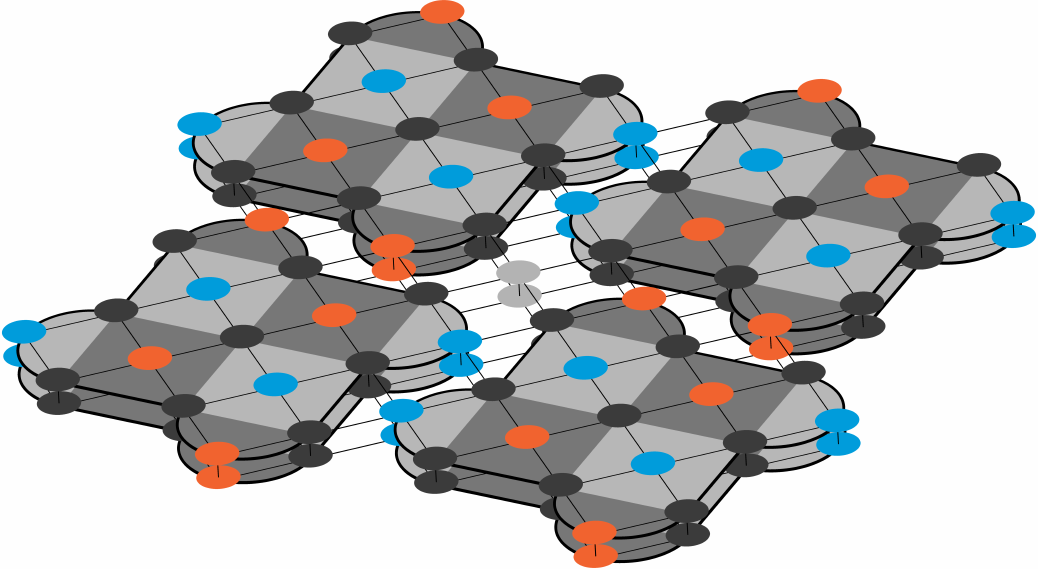}};
    \end{tikzpicture}
    \caption{The bilayer architecture used to implement the syndrome-extraction circuit $\CFT_N$ for the hierarchical code $\cH_N$.
    (a) represents a single layer of the bilayer architecture.
    Colored dots represent syndrome qubits and gray dots represent data qubits.
    Transparent dots represent inactive qubits.
    At any given time step, the qubits that participate in the circuit are depicted as opaque dots and form a lattice of side length $L \ell$; its location within the larger lattice can shift relative to the second layer.
    This is used to facilitate logical Clifford operations.
    (b) represents parallel tiles of distance $d_{\ell}$.
    Each tile represents an outer qubit of the hierarchical code construction.
    Light gray dots will be used to facilitate Clifford operations but are not used in the syndrome-extraction circuit for $\cRS_{\ell}$.
    }
    \label{fig:bilayer-intro}
\end{figure}

The circuits $\CFT_N$ are constructed such that each lattice position remains connected to a fixed and constant-sized set of other lattice positions for any $R = \omega(1)$.
Furthermore, the connectivity does not change dynamically over the course of the circuit.
This way, the wiring can be decided ahead of time.

\paragraph{Section~\ref{sec:numerical estimates} : Comparisons to surface code} Finally, we compare the hierarchical memory $\cH_N$ with a simple memory that only uses rotated surface codes.
At the outset, it may seem unclear whether the use of extra resources to execute the constant-rate LDPC code's syndrome-extraction circuit can be better spent simply building better surface codes which are ideally suited for $2$-dimensional local interactions.
We let $\{\cB_{M}\}$ refer to the \emph{basic} encoding where each logical qubit is encoded in a separate surface code; for some distance $d_{M}$, we let $\cB_M$ be the $K$-fold product of the surface code, i.e.\ $\cB_M = \cRS_{\ell_M}^{\otimes K}$.
The index $M$ represents the total number of qubits in this encoding, i.e.\ $M = \Theta(K d_{M}^2)$.
To contrast $\cH_N$ and $\cB_M$, we present both an asymptotic comparison as well as numerical estimates based on some conservative assumptions.

\begin{theorem}[Informal]
\label{thm:asymptotic-comparison}
    Let $\cH_N$ be a specific $\dsl N,K,D \dsr$ hierarchical code family such that the (outer) constant-rate LDPC code $\cQ_n$ has distance $d = \Theta(n^{\delta})$.
    Let $\cB_M$ be the basic encoding $\cRS_{\ell_M}^{\otimes K}$ that encodes $K$ qubits in separate rotated surface codes of distance $d_{M}$.
    Let $\Cplain_M$ be the corresponding family of syndrome-extraction circuits for $\cB_M$.
    Let $p_{\cB}(M)$ denote the logical failure probability under minimum-weight decoding for $\cB_M$ where we declare failure if \emph{any} logical qubit fails.
    Suppose the gate error rate $p$ %
    is below the thresholds for both the basic encoding and the hierarchical code.
    To achieve $p_{\cB}(M) < p_{\cH}(N)$, we require
    \begin{align*}
        \Space(\Cplain_M) = \Omega\left[ \left(\frac{N}{\log(N)}\right)^{1+2\delta} \right]~, \qquad
        \Depth(\Cplain_M) = \Omega\left[ \left(\frac{N}{\log^2(N)} \right)^{\delta}\right]~.
    \end{align*}
\end{theorem}
We can compare this with parameters for $\CFT_N$ from Theorem \ref{thm:hierarchical-ppties}.
For all $\delta > 0$, the width $\Space$ of $\CFT_N$ is less than that of $\Cplain_M$.
Furthermore, if the outer code has a single-shot decoder, i.e.\ if a constant number of applications of $\CFT_N$ are sufficient to achieve a threshold, then the depth $\Depth$ of $\CFT_N$ is also less than that of $\Cplain_M$.
Efficient single-shot decoders are known to exist for constant-rate  LDPC codes \cite{leverrier2015quantum,fawzi2018constant, fawzi2018efficient}.

Having said this, it is unclear whether this advantage manifests for practically-relevant code sizes and error rates.
To make such a comparison, we use numerical estimates.
We choose the size $M = M(N)$ such that the syndrome-extraction circuits for the hierarchical scheme $\cH_N$ and the basic encoding $\cB_M$ use the same number of physical qubits.
Fixing the total number of qubits in this manner, we look for a \emph{crossover point}, the gate error rate $q_0$ at which the hierarchical code achieves a lower logical failure rate than the basic encoding.

We estimate the circuit-level failure rate using some assumptions about the sub-threshold scaling of the logical failure rate for LDPC codes.
We assume the threshold of the surface code is $10^{-2}$ and the threshold for constant-rate LDPC codes under circuit-level noise is $10^{-3}$.
Our model takes into consideration how the logical failure rate %
depends on the depth of the circuit $\CFT_N$, and how hook errors could reduce the effective distance.
Hook errors are harmful errors that spread from the ancilla qubits to the data qubits during syndrome extraction.
These are explained in Section~\ref{subsec:hook_errors}.

We offer evidence that against circuit-level depolarizing noise, the crossover happens at a gate error rate as high as $5 \times 10^{-3}$ depending on the choice of outer code family and inner/outer code sizes.
See the left-most plot in Figure~\ref{fig:logical-failure-rate-intro}.
\begin{figure}
    \centering
    \includegraphics[scale=0.75]{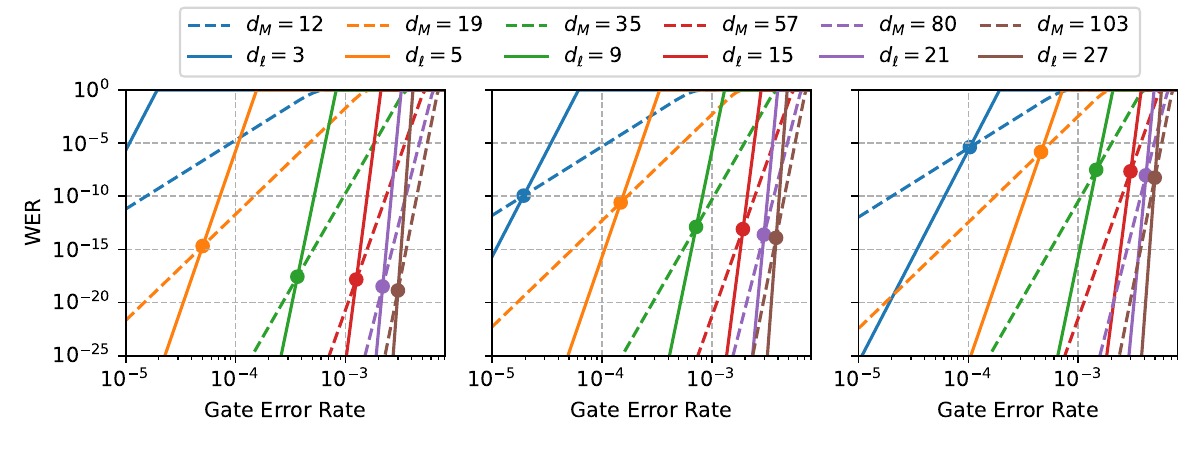}
    \caption{Comparing the logical failure rate for the hierarchical memory versus the logical failure rate for the basic encoding.
    The outer LDPC code has parameters $\dsl 1\; 116\; 416, 112\; 896, 119 \dsr$.
    Each color represents an inner code of distance $d_{\ell}$.
    The solid and dashed lines are estimates for the WER for the hierarchical memory and basic encoding respectively.
    The legend shows the size of the surface codes in each setting.
    For example, the solid blue line represents a hierarchical code with inner code lattice length $d_{\ell} = 3$.
    The dashed blue line represents a basic encoding that uses surface codes of lattice length $12$.
    The three panels correspond to three different assumptions about the error rate in $\swapp$ gates, as described in the text.
    In the left-most plot, $\swapp$ gates are assumed to fail at the same rate as entangling gates.
    In contrast, in the middle and right plots, $\swapp$ gates have a fidelity $10\times$ and $100\times$ better than entangling gates respectively.
    }
    \label{fig:logical-failure-rate-intro}
\end{figure}
These numbers are merely a proof-of-concept and depend on the aforementioned assumptions which are discussed in Section~\ref{subsec:setup-num-est}.

We arrive at these estimates assuming all gates fail with the same probability.
While such an assumption is convenient for proofs, in some architectures, it may be possible to perform $\swapp$ operations with higher fidelity than $\cnot$ or $\cz$ \cite{endres2016atom,barredo2016atom,bluvstein2022quantum,hensinger2006t,kaufmann2017fast}.
For example, in ion trap and neutral atom trap architectures, $\swapp$ gates can be performed by moving the traps.
The mechanism is entirely different than that used to perform other two-qubit gates and, in principle, could have much better fidelity.
These considerations are especially important to us as the main source of noise in the hierarchical scheme stems from $\swapp$ gates.
We present variations of our numerical estimates when the $\swapp$ gates have better fidelity than the $\cnot$ gates.
The middle plot and right-most plot in Figure~\ref{fig:logical-failure-rate-intro} represent estimates for the failure rate when the $\swapp$ gates are $10\times$ and $100\times$ better than entangling gates respectively.

As mentioned, our estimates are predicated on some assumptions.
We re-examine these assumptions in Section~\ref{sec:future} and propose ways to improve the failure rate for hierarchical codes.
We show how we can reduce the effect of hook errors by designing noise-biased qubits.
A qubit is said to have a noise bias if $\ssX$ and $\ssY$ errors are suppressed with respect to $\ssZ$ errors.
\begin{figure}
    \centering
    \begin{tikzpicture}
        \node at (-5,3) {(a)};
        \node at (-2,0) {\includegraphics[width=0.45\columnwidth]{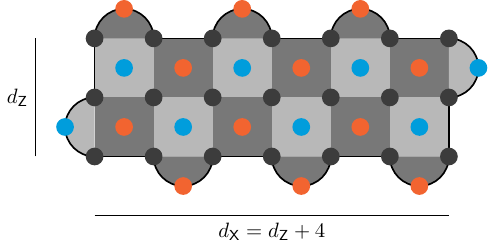}};
        \node at (2.25,3) {(b)};
        \node at (6,0) {\includegraphics[width=0.45\columnwidth]{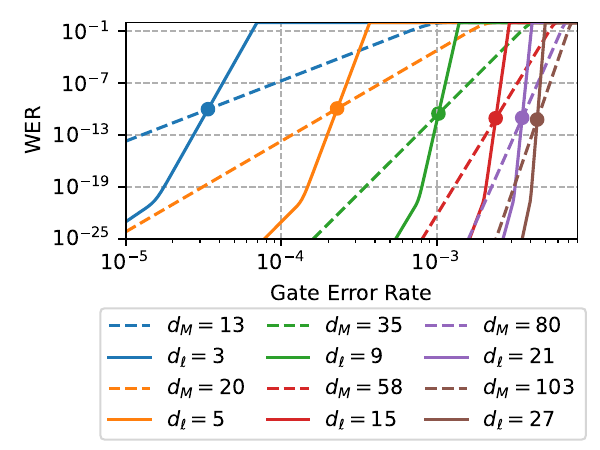}};
    \end{tikzpicture}
    \caption{(a) Creating Level-1 qubits such that the probability of logical $\ssX$ failure is less than the probability of $\ssZ$ failure.
    This is accomplished by changing the aspect ratio of the tiles.
    (b) Estimating crossover points when $\swapp$ gates are $10\times$ better than entangling gates.}
    \label{fig:biased-noise-intro}
\end{figure}
We can introduce a bias on Level-1 qubits using unbiased Level-0 (physical) qubits.
As the inner code is a surface code, we can engineer a bias simply by making the surface code longer in one direction of our choosing.
See Figure~\ref{fig:biased-noise-intro} (a).
Based on our estimates, we expect this can reduce the size of the code considerably.
Figure~\ref{fig:biased-noise-intro} (b) shows the crossover points for the hierarchical code and the basic encoding with the assumption that $\swapp$ gates are $10\times$ better than entangling gates using a much smaller outer code.

Secondly, we believe that decoders for the hierarchical code can take advantage of their concatenated structure.
To achieve this, we propose using message-passing decoders between the outer and inner codes.
These ideas can be used in soft decoders for the outer code to partially overcome the problems of degeneracy \cite{poulin2008iterative}.
Similar ideas have proved useful in the context of concatenating GKP codes and LDPC codes \cite{raveendran2022finite}.

Lastly, we expect the hierarchical scheme to be resistant to \emph{burst} errors.
Unlike typical errors which affect only one or two qubits at a time, burst errors can wipe out entire patches of qubits.
This can happen when there exists poorly localized error mechanisms such as the absorption of cosmic rays in superconducting circuits \cite{vepsalainen2020impact,mcewen2022resolving,thorbeck2022tls,acharya2022suppressing,cardani2023disentangling} or blackbody radiation mediated transitions to other Rydberg states in neutral atom platforms \cite{festa2022blackbody,zeiher2016many}.
Large deviations may also occur in single points of failure in the control hardware such as power supplies, local oscillators, lasers, etc \cite{arute2019quantum, kraglund2019entanglement, madjarov2020high, ebadi2021quantum, ma2022universal}. %
Protection of surface codes from burst errors was initially studied in \cite{xu2022distributed} by concatenating a small constant-sized stabilizer code with surface codes.
The hierarchical scheme is robust to these errors because each inner surface code represents a qubit of the outer code which we know is resistant to some number of erasure errors.

\textbf{Related work:} Gottesman \cite{gottesman2000fault} demonstrated that it is possible to find a threshold using only local gates and concatenation.
Svore, Divincenzo and Terhal \cite{svore2005local,svore2006noise} studied this issue further and established a numerical lower bound on the threshold in a scheme with many layers of concatenation.
Yamasaki and Koashi \cite{yamasaki2022time} show that concatenated codes can be used to achieve constant overhead quantum computation that is also time efficient.

In contrast to these approaches, we consider a qualitatively different setting.
In our hierarchical model, the concatenated code has only two layers.
The outer LDPC code grows quickly to improve the error rate, while the inner code grows slowly to achieve a threshold.
The number of encoded logical qubits in the code therefore increases (sublinearly) with the size of the code.
Consequently, the rate of error suppression is significantly better.

Finally, Baspin \emph{et al}.\ ~\cite{baspin2023lower} have recently generalized the result of Delfosse \emph{et al}.\ in another direction.
In contrast to the constructive approach in this paper, they approach this problem top-down --- given access to \emph{arbitrary} local operations and classical communication (not merely Clifford operations), they study syndrome-extraction circuits for LDPC codes and their ability to suppress stochastic errors.
They prove the existence of a tradeoff between the parameters of the syndrome-extraction circuit and the sub-threshold error scaling (See Theorem 28 of \cite{baspin2023lower}).
For fixed gate error rate $p$, suppose we use an $\dsl N,K,D \dsr$ code $\cH_N$ and desire a sub-threshold scaling of the logical failure rate $p_{\cH}(N) = \exp(-f(N))$ for some function $f(N)$.
Let $C_N$ be the corresponding family of syndrome-extraction circuits.
Assuming $f(N) = O(N)$, we express Theorem 28 of \cite{baspin2023lower} in our notation
\begin{align}
    \label{eq:bfs22}
    \frac{\Space(C_N)}{K} = \Omega\left(\frac{\sqrt{f(N)}}{\Depth(C_N)} \right)~.
\end{align}
To compare with our result, suppose we only use $\swapp$ gates of constant range, i.e.\ $R = O(1)$.
From Theorem~\ref{thm:hierarchical-ppties}, the syndrome-extraction circuit $\CFT_N$ achieves $p_{\cH}(N) = \exp(-\Theta(N^{\delta}/\log^{2\delta}(N))$ with $\Space(\CFT_N) = \Theta(N)$ and $\Depth(\CFT_N) = O(\sqrt{N})$.
\begin{align}
    \frac{\Space(\CFT_N)}{K} = O(\log(N))~, \qquad \frac{\sqrt{f(N)}}{\Depth(\CFT_N)} = O\left(\frac{N^{(\delta-1)/2}}{\log^{\delta}(N)}\right)~.
\end{align}
Comparing with Equation~\eqref{eq:bfs22}, we can see that the bound is satisfied for any constant $\delta > 0$.
Note that such a low logical error rate is only feasible because our syndrome-extraction circuit $\CFT_N$ has polynomially growing depth.

\input{2-background}
\input{3-explicit-design}
\input{4-bilayer-arch}

\input{5-asymptotic}

\input{6-quant-estim}

\section{Conclusions}

We have constructed a quantum memory with a threshold using geometrically local gates to simulate long-range connectivity.
We did so by constructing a code family $\{\cH_N\}$ that we refer to as the hierarchical code.
$N$ indexes the size of the code; the $N$\textsuperscript{th} element $\cH_N$ of this code is obtained by concatenating a constant-rate quantum LDPC code $\cQ_{n}$ (the $n$-qubit outer code) and the surface code $\cRS_{\ell}$ (the inner code).
The outer code has a number of encoded qubits $k(n) = \rho \cdot n$ and distance $d(n) = \Theta(n^{\delta})$ for positive constants $\rho, \delta$.
Our construction builds on Gottesman's proof of the existence of a threshold using quantum LDPC codes (Theorem 4 from \cite{gottesman2014fault}).
The central idea in Gottesman's construction is that if the failure rate \emph{per round} of syndrome extraction, denoted $\pround$, is a sufficiently small constant, then logical errors can be suppressed exponentially in the distance of the code. We showed that the requisite constant error rate per round can be achieved using geometrically-local gates if the inner code has suitable properties. 

Although $\cH_N$ is no longer an LDPC code, local operations suffice for extracting the error syndrome.
In Section~\ref{sec:surface-code}, we presented an explicit family of syndrome-extraction circuits $\{\CFT_N\}$ for $\cH_N$.
This circuit has width $\Space(\CFT_N) = \Theta(N)$ and depth $\Depth(\CFT_N) = O(\sqrt{N}/R)$, where $R$ denotes the range of physical $\swapp$ gates.
To describe this circuit for the hierarchical code, we first presented a construction of the syndrome-extraction circuit $C_n$ for the outer LDPC code $\cQ_n$ in Section~\ref{subsec:sec-from-routing}. 
This circuit is based on a bilayer architecture --- physical qubits are laid out in two layers in $2$ dimensions.
In our concatenated construction, the outer qubits of $\cQ_n$ are replaced by rotated surface codes referred to as tiles.
In Section~\ref{subsec:swap-bilayer}, we demonstrated how to perform Level-1 logical Clifford operations on tiles using physical nearest-neighbor gates, including a novel technique for performing nearest-neighbor logical $\swapp$ gates.
We also discussed how to perform logical $\swapp$ operations on tiles with range $R_1$ using physical $\swapp$ operations with range $R_0$.

In Section~\ref{sec:ft-asymptotics}, we showed that for fixed values of the physical failure rate $\pphys$, the error rate per round of syndrome-extraction, $\pround$, is a polynomial function of the depth $\Depth(\CFT_N)$.
Using an inner surface code with linear size $\ell$, which can suppress errors exponentially in $\ell$, we can guarantee that the Level-1 error rate per round is a constant by choosing $\ell = \Theta(\log(n))$.
The resulting concatenated code $\cH_N$ encodes a number of encoded qubits $K = \Omega(N/\log(N)^2)$.
Furthermore, if the distance of the LDPC code $\cQ_n$ is $d(n) = \Theta(n^{\delta})$, then $\cH_N$ can suppress errors superpolynomially; the Word Error Rate (WER) satisfies $p_{\cH}(N) < \exp(-\Theta[N^{\delta}/\log^{2\delta}(N)])$.
Given access to physical $\swapp$ operations of range $R$, the syndrome-extraction circuit $\CFT_N$ has depth $O(\sqrt{N}/R)$.

Using this architecture we made numerical estimates of the WER $p_{\cH}(N)$ in Section~\ref{sec:numerical estimates}.
We contrasted this with the WER $p_{\cB}(M)$ of the \emph{basic encoding} $\cB_M$, where all logical qubits are encoded using only the surface code.
We first made comparisons in the asymptotic regime, showing in Section~\ref{subsec:asymptotic-comparison} that if the outer constant-rate LDPC code has an efficient single-shot decoder, then a target logical error rate can be achieved more efficiently using the hierarchical encoding rather than the basic encoding.

We then proceeded with numerical estimates probing whether this advantage holds for practical code sizes and noise parameters.
For this purpose, we compared the WERs of the basic encoding and hierarchical encoding when both schemes use the same total number of physical qubits. 
We found that the physical error rate has a \emph{crossover point}; when the physical error rate is below this value, the hierarchical code outperforms the basic encoding.
To perform these estimates, we made assumptions about the noise model and about the WER for surface codes and LDPC codes, and we assessed the impact of these assumptions on our conclusions.
We also discussed some ways to reduce the WER of hierarchical codes by modifying the syndrome-extraction circuit, improving the fidelity of $\swapp$ operations, and using more sophisticated decoding algorithms. 
\begin{enumerate}
\item We made the conservative assumption that propagation of error from Level-1 ancilla qubits to Level-1 data qubits reduces the effective distance of the outer code by a factor of $\Delta_g$, the degree of the outer-code stabilizer generators. This error propagation can be mitigated if the noise in Level-1 qubits is highly biased, with $\ssX$ errors occurring much less frequently than $\ssZ$ errors. Even if the noise afflicting the physical qubits is unbiased, this Level-1 noise bias can be enforced by using an asymmetric surface code as the inner code of the hierarchical scheme. 
\item The failure rate of the outer code grows in proportion to the depth of the permutation routing, and hence is sensitive to the error rate of Level-1 $\swapp$ operations.
By improving the error rate of physical $\swapp$ gates we can improve the performance of the hierarchical code significantly.
\item We assumed that the decoding algorithm for the outer code makes no use of the syndrome information from the inner code blocks. We expect that a much better decoding scheme for the hierarchical code can be achieved by exploiting such information from the inner code when decoding the outer code. 

\end{enumerate}
Finally, we also highlighted that a hierarchical architecture might deal effectively with ``burst'' errors that damage a large cluster of physical qubits simultaneously. A severe burst error could corrupt several of the inner-code tiles, but the resulting Level-1 erasure errors can be adequately addressed by the decoder for the outer code.

\section{Acknowledgements}

AK is supported by the Bloch Postdoctoral Fellowship from Stanford University. AK acknowledges funding from NSF award CCF-1844628.
CAP acknowledges funding from the Air Force Office of Scientific Research (AFOSR), FA9550-19-1-0360.
JP acknowledges funding from  the U.S. Department of Energy Office of Science, Office of Advanced Scientific Computing Research, (DE-NA0003525, DE-SC0020290), the U.S. Department of Energy QuantISED program ({DE}-{SC0018407}), the U.S. Department of Energy Quantum Systems Accelerator, the Air Force Office of Scientific Research ({FA9550}-{19}-{1}-{0360}), and the National Science Foundation ({PHY}-{1733907}). The Institute for Quantum Information and Matter is an NSF Physics Frontiers Center. 
We thank Nicolas Delfosse, Mary Wootters, Anthony Leverrier, Nou\'edyn Baspin, Bailey Gu, Alex Kubica, David Schuster, Manuel Endres, Michael Vasmer, and Pavel Panteleev for helpful conversations.

\printbibliography

\pagebreak

\input{appendix}

\end{document}

%% file: 2-background.tex
\section{Background \& Notation}
\label{sec:background}

In this section, we begin by formally defining concepts needed to state our results.
Section \ref{subsec:basic_defs} defines syndrome-extraction circuits.
We review gadgets used to construct them and how to use these gadgets to obtain a syndrome-extraction circuit given an error correcting code.
Section \ref{subsec:noise-and-imperfect-syndromes} reviews locally decaying distributions that describe errors on states and faults on circuits.
These are general error models that can describe the types of correlated errors that we might witness in a circuit.
A noise model is parameterized by a failure rate which quantifies the probability of errors.
We described how error correcting codes and their associated syndrome-extraction circuits are robust to some amount of errors occurring below a \emph{threshold} failure rate.
Section \ref{subsec:concat-review} reviews syndrome-extraction circuits for concatenated codes.
The hierarchical code is constructed by concatenating a constant-rate quantum LDPC code and the surface code.
These are defined in Section \ref{subsec:ldpc-review} and Section \ref{subsec:surface-review} respectively.
We review Gottesman's requirements \cite{gottesman2014fault} for the existence of a threshold.
This will be an important idea in the proof of the existence of a threshold for hierarchical codes.

\subsection{Basic definitions}
\label{subsec:basic_defs}

Let $\cP = \langle \ssX, \ssZ\rangle/\{\pm i, \pm 1\}$ denote the (projective) single-qubit Pauli group (where we ignore phases); for $n \in \bbN$, let $\cP_n$ denote the $n$-fold tensor product $\cP^{\otimes n}$.
For $\ssP \in \cP_n$, $\supp(\ssP) \subseteq [n]$ denotes the support of $\ssP$, i.e.\ the set of qubits on which $\ssP$ acts non-trivially.
The \emph{weight} of a Pauli operator $\ssP$ is $|\supp(\ssP)|$, the number of qubits in its support.
For brevity, we denote this as $|\ssP|$.

For $\ba$, $\bb \in \{0,1\}^n$, let $\ssX(\ba) = \otimes_{i} \ssX^{a_i}$, and $\ssZ(\bb) = \otimes_j \ssZ^{b_j}$.
Any Pauli operator $\ssP \in \cP_n$ can be expressed uniquely as $\ssP = \ssX(\ba) \ssZ(\bb)$ for $\ba$, $\bb \in \{0,1\}^n$.
We use $\ssP|_{\ssX}, \ssP|_{\ssZ} \in \{0,1\}^n$ to denote the $\ssX$ and $\ssZ$ components of $\ssP$ respectively, i.e.\ $\ssP|_{\ssX} := \ba$ and $\ssP|_{\ssZ} := \bb$.

\textbf{Stabilizer codes:}
An $n$-qubit quantum error correcting code is the simultaneous $+1$-eigenspace of a set of commuting Pauli operators.
These Pauli operators form a subgroup $\cS$ of the Pauli group called the stabilizer group.
The stabilizer group $\cS$ is generated by elements $\ssS_1,...,\ssS_{m}$.
The codespace $\cQ$ is then defined as
\begin{align*}
  \cQ = \{\ket{\psi} \in (\bbC^2)^{\otimes n} |\; \ssS_i \ket{\psi} = \ket{\psi} \; \forall i \in [m] \}~.
\end{align*}
The number of encoded qubits $k$ is the base $2$ logarithm of the number of linearly-independent vectors in $\cQ$.
Equivalently, given $\cS$, it is simply $k = n-m$ (where we assume the stabilizer generators are independent).

Intuitively, the minimum distance $d$ of the code is the minimum weight Pauli operator such that we can map one element of $\cQ$ to a distinct element of $\cQ$.
Formally, we write
\begin{align*}
  d = \min_{\substack{\ssP \in \cP_{n} \setminus \cS\\ [\ssP,\ssS_i] = 0}} |\ssP| ~.
\end{align*}
We say such a code is an $\dsl n,k,d \dsr$ (stabilizer) code.

The code is said to be a CSS code if every generator can be chosen such that it is a tensor product of only $\ssX$ or $\ssZ$ Pauli operators \cite{calderbank1996good,steane1996multiple}.
We can define the $\ssX$- and $\ssZ$-distances $d_{\ssX}$ and $d_{\ssZ}$ of a CSS code as
\begin{align*}
  d_{\ssX} = \min_{\substack{\ssP \in \{\ssI, \ssZ\}^{\otimes n} \setminus \cS \\ [\ssP,\ssS_i] = 0}} |\ssP| \qquad
  d_{\ssZ} = \min_{\substack{\ssP \in \{\ssI, \ssX\}^{\otimes n} \setminus \cS \\ [\ssP,\ssS_i] = 0}} |\ssP|~.
\end{align*}

Let $1 \leq b \leq m_{\ssX}$ and $1 \leq c \leq m_{\ssZ}$ index the $\ssX$-type and $\ssZ$-type stabilizer generators $\{\ssS_b^{\ssX}\}$ and $\{\ssS_c^{\ssZ}\}$.

\textbf{Syndrome-extraction circuits \& measurement gadgets:} A syndrome-extraction circuit $C$ for a code $\cQ$ can be composed of the following elements that are allowed to be classically controlled.
\begin{definition}[Clifford operations]
\label{def:local-cliff}
Consider a set of qubits arranged in a lattice in $2$ dimensions.
We define the set $\cK$ of  \emph{elementary Clifford operations} as follows:
\begin{enumerate}
  \item Initialization of new qubits in state $\ket{0}$ or $\ket{+}$,
  \item Single-qubit Pauli gates,
  \item Two-qubit Clifford gates $\cnot$ between nearest-neighbor qubits,
  \item Single-qubit Pauli $\ssX$ and $\ssZ$ measurements,
  \item Physical $\swapp$ operation with range up to $R$.
\end{enumerate}
\end{definition}
At any given time step, a qubit in $C$ can be involved in at most one of these operations.
In addition, we assume instantaneous classical communication and access to classical computation for processing measurement data.

To obtain the syndrome, we use gadgets to measure Pauli operators which are described as follows.
Consider a CSS code $\cQ$ with $m_{\ssX}$ $\ssX$-type stabilizer generators $\cS^{\ssX} = \{\ssS_b^{\ssX}\}_{b=1}^{m_{\ssX}}$ and $m_{\ssZ}$ $\ssZ$-type stabilizer generators $\cS^{\ssZ} = \{\ssS_{b}^{\ssZ}\}_{c=1}^{m_{\ssZ}}$.
The entire set of stabilizer generators is $\cS = \cS^{\ssX} \union \cS^{\ssZ}$.
\begin{enumerate}
  \item \textbf{For $1 \leq b \leq m_{\ssX}$}, measure a product of $\ssX$ operators:
  \begin{enumerate}
    \item Initialize the $b$\textsuperscript{th} ancilla qubit in $\ket{+}_b$.
    \item Perform a $\cnot$ gate controlled on the $b$\textsuperscript{th} ancilla qubit and targeted on each qubit in the support of $\ssS_b^{\ssX}$.
    \item Perform a measurement of the $b$\textsuperscript{th} ancilla in the $\ssX$ basis.
  \end{enumerate}
  \item \textbf{For $1 \leq c \leq m_{\ssZ}$,} measure a product of $\ssZ$ operators:
  \begin{enumerate}
    \item Initialize the $c$\textsuperscript{th} ancilla qubit in $\ket{0}_c$.
    \item Perform a $\cnot$ gate targeted on the $c$\textsuperscript{th} ancilla qubit and controlled on each data qubit in the support of $\ssS_c^{\ssZ}$.
    \item Perform a measurement of the $c$\textsuperscript{th} ancilla in the $\ssZ$ basis.
  \end{enumerate}
\end{enumerate}
Figure \ref{fig:syndrome-extraction} illustrates gadgets for measuring an $\ssX$ operator of weight $5$ and a $\ssZ$ operator of weight $4$.
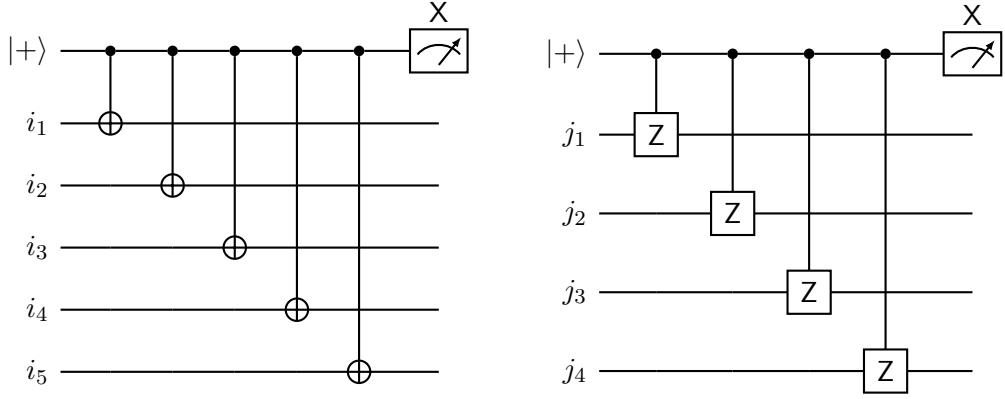
\begin{figure}[h]
  \centering
  \begin{quantikz}[row sep={0.6cm,between origins}]
    \lstick{$\ket{+}$}& \ctrl{1} & \ctrl{2} & \ctrl{3} & \ctrl{4} & \ctrl{5} & \meter{$\ssX$} & \midstick[2,brackets=none]{} & \lstick{$\ket{0}$} & \targ{}   & \targ{}   & \targ{}   & \targ{}   & \targ{}   & \meter{$\ssZ$} \\
    \lstick{$i_1$}    & \targ{}  & \qw      & \qw      & \qw      & \qw      & \qw            &                              & \lstick{$j_1$}     & \ctrl{-1} & \qw       & \qw       & \qw       & \qw       & \qw \\
    \lstick{$i_2$}    & \qw      & \targ{}  & \qw      & \qw      & \qw      & \qw            &                              & \lstick{$j_2$}     & \qw       & \ctrl{-2} & \qw       & \qw       & \qw       & \qw \\
    \lstick{$i_3$}    & \qw      & \qw      & \targ{}  & \qw      & \qw      & \qw            &                              & \lstick{$j_3$}     & \qw       & \qw       & \ctrl{-3} & \qw       & \qw       & \qw \\
    \lstick{$i_4$}    & \qw      & \qw      & \qw      & \targ{}  & \qw      & \qw            &                              & \lstick{$j_4$}     & \qw       & \qw       & \qw       & \ctrl{-4} & \qw       & \qw \\
    \lstick{$i_5$}    & \qw      & \qw      & \qw      & \qw      & \targ{}  & \qw            &                              & \lstick{$j_5$}     & \qw       & \qw       & \qw       & \qw       & \ctrl{-5} & \qw
  \end{quantikz}
  \caption{Performing the syndrome extraction corresponding to the operator $\ssX_{i_1}\ssX_{i_2}\ssX_{i_3}\ssX_{i_4}\ssX_{i_5}$ on the left and $\ssZ_{j_1}\ssZ_{j_2}\ssZ_{j_3}\ssZ_{j_4}\ssZ_{j_5}$ on the right.
  The measurements are performed on some qubits $\{i_1,...,i_5,j_1,...,j_5\} \subseteq [n]$.}
  \label{fig:syndrome-extraction}
\end{figure}
Given a circuit $C$, we use two figures-of-merit to quantify its size:
\begin{enumerate}
  \item $\Space(C)$: the width of the circuit, i.e.\ the total number of physical qubits, data and ancilla, used in the circuit.
  \item $\Depth(C)$: the depth of the circuit, i.e.\ the number of time steps required to measure all syndromes.
\end{enumerate}

We assume operations in $\cK$ can be performed in parallel (subject to the constraint that each qubit is only involved in one operation at a time).
We shall present one way of using parallel operations to build efficient syndrome-extraction circuits for quantum LDPC codes in Section \ref{sec:design-routing}.

\subsection{Noise \& imperfect syndrome-extraction circuits}
\label{subsec:noise-and-imperfect-syndromes}

In practice, $C$ is imperfect.
In general, errors on multiple locations with complicated correlations can arise at the end of a syndrome-extraction circuit.
Under the action of a two-qubit gate for instance, single-qubit errors which occur with probability $p$ can transform into two-qubit correlated errors which occur with probability $p$.
Two-qubit gates themselves can fail and introduce errors on both qubits where there were none before.
As yet another example, small clusters of qubits that are near each other can also fail together, for example, due to crosstalk, stray magnetic fields, etc.
These errors are outside the scope of an i.i.d.\ errors model and hence, we consider a generalization.

\begin{definition}
    Let $n \in \bbN$ and $\Pow(n) = \{E : E \subseteq [n]\}$.
    Consider a probability distribution $\widehat{\Pr}:\Pow(n) \to [0,1]$ and for $E \subseteq [n]$, let $\Pr(E)$ be the total probability
    \begin{align*}
        \Pr(E) = \sum_{E' \supseteq E} \widehat{\Pr}(E')~.
    \end{align*}
    We say the distribution $\Pr$ is locally decaying with rate $p \in [0,1]$ if for all $E \subseteq [n]$,
    \begin{align*}
        \Pr(E) \leq p^{|E|}~.
    \end{align*}
\end{definition}

We first consider general errors on an $n$-qubit state.
We assume every set of qubits has some probability of being corrupted by an arbitrary Pauli error.
Consider a Pauli operator $\ssE' \in \cP_n$ such that $\ssE' = \ssX(\bx')\ssZ(\bz')$.
Let $\widehat{\cE}(\bx', \bz')$ be the probability of the error $\ssE'$.
By definition, $\widehat{\cE}$ is itself a map from $\Pow(n) \times \Pow(n)$ to $[0,1]$.
Let $\cX(\bx): \Pow(n) \to \bbR$ and $\cZ(\bz):\Pow(n) \to \bbR$ denote
\begin{align}
   \cX(\bx) = \sum_{\bx' \supseteq \bx} \sum_{\bz'} \widehat{\cE}(\bx', \bz')~,\qquad
   \cZ(\bz) = \sum_{\bx'} \sum_{\bz' \supseteq \bz} \widehat{\cE}(\bx', \bz')~.
\end{align}
In other words, $\cX$ and $\cZ$ denote the probability that a random error $\ssE$ distributed according to $\widehat{\cE}$ has $\ssX$ and $\ssZ$ components that contain $\bx$ and $\bz$ respectively.
For brevity, we have used $\bx' \supseteq \bx$ and $\bz' \supseteq \bz$ to mean that the supports of $\bx$, $\bz$ are contained in $\bx'$, $\bz'$ respectively.
Treating $\cX$ and $\cZ$ separately in this way does not prevent correlations between $\ssX$ and $\ssZ$ errors.

\begin{definition}[Locally decaying errors on qubits]
    \label{def:lde-qubits}
    Given an $n$-qubit state with Pauli errors distributed according to $\widehat{\cE}$.
    We say that errors are described by a \emph{locally decaying errors model} to mean that $\cX$ and $\cZ$ are both locally decaying distributions with failure rate $p$.
\end{definition}

We want to extend this idea to describe errors caused by faulty circuits.
A \emph{location} in a circuit $C$ refers to a one- or two-qubit gate (including identity), single-qubit preparation or single-qubit measurement operation at some time step $1 \leq t \leq \Depth(C)$.
A fault location is a location which performs a random Pauli operation following the desired Clifford operation.
We assume that a fault location introduces a Pauli operator on the qubits in its support chosen according to some distribution $\widehat{\cF}$.
Given a set $F$ of fault locations in $C$, the support of $F$ is the set $\supp(F) \subseteq [\Space(C)]$ of qubits that are in some location in $F$.

For a set $F'$ of locations, let $\widehat{\cF}(F')$ denote the probability of the set of locations $F'$ being faulty.
For a set $F$ of fault locations, the total probability $\cF(F)$ is
\begin{align}
    \cF(F) = \sum_{F' \supseteq F} \widehat{\cF}(F')~.
\end{align}

\begin{definition}[Locally decaying faults on circuits]
    \label{def:lde-circuits}
    Let $C$ be a depth $1$ circuit with faults distributed according to $\widehat{\cF}$.
    We say that the faults are described by a \emph{locally decaying faults model} if $\cF$ is a locally decaying distribution with failure rate $\pphys$---for all sets of locations $F$,
    \begin{align*}
        \cF(F) \leq \pphys^{|F|}~.
    \end{align*}
\end{definition}
Note that the probability of failure falls with the number of locations $|F|$ and not the number of qubits $|\supp(F)|$.

In practice, different locations may have different failure rates.
To prove the existence of a threshold, we assume that $\pphys$ is the maximum failure probability across all gates.
We return to this assumption in Section~\ref{sec:numerical estimates}, where we discuss how the logical failure rate behaves if gates have different failure rates.

Definition~\ref{def:lde-circuits} pertains to circuits of depth $1$---we assume faults in successive time steps are independent.
In a more general model for faults, we could include arbitrary fault patterns for a circuit of growing depth so long as the probability of a particular fault path falls exponentially with the size of the fault path.

As a state undergoes circuit operations, errors can spread and accumulate.
Consider a $\cnot$ gate acting on two qubits.
Figure~\ref{fig:cnot-evoln} illustrates how a generating set of $2$-qubit Pauli operators $\{\ssX \ssI, \ssI \ssZ, \ssI \ssX, \ssZ \ssI\}$ on these two qubits evolve under ideal $\cnot$.
The error doubles in size in the worst-case scenario.
As shorthand, we say that Pauli operators `flow' within circuits to refer to this spreading.
$\ssX$ operators flow down a $\cnot$ and $\ssZ$ operators flow up.

  \begin{figure}[h]
    \centering
    \begin{tikzpicture}
        \node at (-0.9,1.2) {$(a)$};
        \draw (0,0) to (1,0);
        \draw[fill=black] (1,1) circle (0.07);
        \draw (1,0) circle (0.13);
        \draw (1,1) to[-] (1,-0.13) (0.87,0) to[-] (1.13,0);
        \draw[red,thick] (1,1) to[-] (2,1);
        \draw[color=red,rounded corners,thick] (0,1) -- ++ (1,0) |- (2,0);
        \node at (-0.3,1) {{\color{red} $\ssX$}};
        \node at (-0.3,0) {$\ssI$};
        \node at (2.3,1) {{\color{red} $\ssX$}};
        \node at (2.3,0) {{\color{red} $\ssX$}};
    \begin{scope}[shift={(5,0)}]
        \node at (-0.9,1.2) {$(b)$};
        \draw (0,1) to (1,1);
        \draw[fill=black] (1,1) circle (0.07);
        \draw (1,0) circle (0.13);
        \draw (1,1) to[-] (1,-0.13) (0.87,0) to[-] (1.13,0);
        \draw[red,thick] (1,0) to[-] (2,0);
        \draw[color=red,rounded corners,thick] (0,0) -- ++ (1,0) |- (2,1);
        \node at (-0.3,0) {{\color{red} $\ssZ$}};
        \node at (-0.3,1) {$\ssI$};
        \node at (2.3,1) {{\color{red} $\ssZ$}};
        \node at (2.3,0) {{\color{red} $\ssZ$}};
    \end{scope}
    \begin{scope}[shift={(0,-2.5)}]
        \node at (-0.9,1.2) {$(c)$};
        \draw (0,1) to (2,1) (0,0) to (2,0);
        \draw[fill=black] (1,1) circle (0.07);
        \draw (1,0) circle (0.13);
        \draw (1,1) to[-] (1,-0.13) (0.87,0) to[-] (1.13,0);
        \node at (-0.3,0) {$\ssX$};
        \node at (-0.3,1) {$\ssI$};
        \node at (2.3,1) {$\ssI$};
        \node at (2.3,0) {$\ssX$};
    \end{scope}
    \begin{scope}[shift={(5,-2.5)}]
        \node at (-0.9,1.2) {$(d)$};
        \draw (0,1) to (2,1) (0,0) to (2,0);
        \draw[fill=black] (1,1) circle (0.07);
        \draw (1,0) circle (0.13);
        \draw (1,1) to[-] (1,-0.13) (0.87,0) to[-] (1.13,0);
        \node at (-0.3,0) { $\ssI$};
        \node at (-0.3,1) {$\ssZ$};
        \node at (2.3,1) {$\ssZ$};
        \node at (2.3,0) {$\ssI$};
    \end{scope}
    \end{tikzpicture}
    \caption{Evolution of Pauli errors under the action of $\cnot$.
    The first qubit is the control qubit and the second qubit is the target.
    The operators $\ssX \otimes \ssI$ and $\ssI \otimes \ssZ$ double in size.
    The red paths show how $\ssX$ `flows down' a $\cnot$ gate and $\ssZ$ `flows up' a $\cnot$ gate.
    }
    \label{fig:cnot-evoln}
  \end{figure}
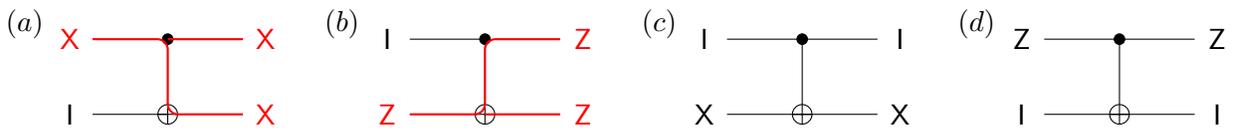

The structure of syndrome-extraction circuits is special.
For $\ssP \in \{\ssX, \ssZ\}$, a controlled-$\ssP$ gate within the syndrome-extraction circuit uses ancilla qubits as control qubits and data qubits as target qubits (See Figure~\ref{fig:syndrome-extraction}).
This means that errors only flow in limited ways---for example, $\ssX$ errors always flow from ancilla qubits to data qubits, and $\ssZ$ errors flow from data qubits to ancilla qubits when $\cnot$ gates are applied.

Implementing the imperfect circuit $C$, we obtain an imperfect syndrome.
To overcome this problem, we repeat the syndrome-measurement circuit for $r$ \emph{rounds}.
Let $\bsig = (\bsig_{\ssX}^{(1)}, \bsig_{\ssZ}^{(1)},...,\bsig_{\ssX}^{(r)},\bsig_{\ssZ}^{(r)})$ be the $r$ faulty syndromes.

\textbf{Failure rate per round:}

Consider a corrupted code state $\ssE \ket{\psi}$ where $\psi$ is a code state and $\ssE = \ssX(\be_x) \ssZ(\be_z)$ is some Pauli operator.
If the syndrome-extraction circuit $C$ has no faults, the joint state of the data and ancilla qubits after one round of syndrome extraction is described by
\begin{align}
    \ssE \ket{\psi} \otimes \ssZ(\bsig) \ket{+}^{\otimes m}~,
\end{align}
where $\bsig$ represent the ideal syndromes for $\ssX$- and $\ssZ$-type stabilizer generators.

However, because of faults in the circuit, the state after the circuit is
\begin{align}
    (\ssD \otimes \ssA )(\ssE \ket{\psi} \otimes \ssZ(\bsig) \ket{+}^{\otimes m})~,
\end{align}
where $\ssD$ and $\ssA$ represent errors on the data and ancilla qubits respectively caused by faults in $C$ that then spread.

Let $\widehat{\cE}'(\ssD \otimes \ssA)$ denote the probability of errors \emph{per round} on the qubits.
Let $\cX'$, $\cZ'$ denote the induced distributions for errors on data and ancilla qubits of $\ssX$ and $\ssZ$ type respectively.

\begin{definition}[Probability of errors per round]
\label{def:pround}
We say that the probability of errors per round is locally decaying with failure rate $\pround \in [0,1]$ such that $\cX'$, $\cZ'$ are locally decaying distributions with failure rates $\pround$ respectively.
\end{definition}
Definition~\ref{def:pround} thus considers one round of syndrome extraction not as individual operations, but in aggregate; it then associates a failure probability $\pround$ with the entire round associated with the probability of witnessing $\ssX$ and $\ssZ$ errors.
Thus, $\pround$ can be a function of the code size $N$, as well as other details of the implementation such as the specific syndrome-extraction circuit used.

A priori, $\widehat{\cE}'$ can depend on $r$ and the input error $\ssE$.
However, as entangling gates restrict the direction of error propagation, errors do not propagate from one data qubit to another or from one ancilla qubit to another.
In Section~\ref{subsec:evoln-syndrome-ext}, we use this to show that $\pround$ does not depend on how many prior rounds of the syndrome-extraction circuit have already been applied.
We show that $\pround$ is a function of $\pphys$ of the form $a \cdot \pphys^{b}$, where $a$ is a function of the depth $\Depth(C)$ and $b$ is a function of the degrees $\Delta_q$ and $\Delta_g$.

\paragraph{Recovering the state:}
After performing $r$ rounds of syndrome extraction, a \emph{decoding algorithm} $\dec: (\bbF_2^m)^{\times r} \to \cP_n$ maps the observed syndrome $\bsig$ to a deduced error.

The applied correction may not completely correct all errors due to faults in the syndrome extraction circuit.
We declare success if, after applying the correction, the final state is `not too far' from the desired output of the ideal circuit $C$.
To this end, we consider the ideal recovery map $\cR$ \cite{aliferis2005quantum}---a fictitious quantum channel that is not subject to geometric constraints or noise.
We gauge the accuracy of the circuit $\widetilde{C}$ using the logical failure probability $p_{\cQ}$, which is the probability that the residual error is correctable by the ideal recovery map.
To be precise, $p_{\cQ}$ is the probability that \emph{any} logical qubit fails in one round of error correction.
The probability $p_{\cQ}$ also referred to as the \emph{Word Error Rate} (WER).

\paragraph{Ideal recovery map \& Thresholds:}

To understand whether a scheme is scalable, we are interested in properties of a \emph{family} of codes $\{\cQ_n\}$ to process an ever increasing number $n$ of qubits.
Consider a code family $\{\cQ_n\}$ and suppose errors are described by a locally decaying distribution $\cE$ with failure rate $\pin$.
Let $\{C_n\}$ be the corresponding set of syndrome-extraction circuits to $\{\cQ_n\}$, where faults are described by $\cF$, a locally decaying distribution with failure rate $\pphys \in [0,1]$.
We can compute $\pround$ as a function of $\pphys$ as shown in Section~\ref{subsec:evoln-syndrome-ext}.

For our purposes, we say that the family has a \emph{threshold} with respect to the noise model and decoding algorithm if there exists a pair $\qin$, $\qround \in (0,1]$ such that if
\begin{align}
    \pin < \qin~, \qquad \pround < \qround ~,
\end{align}
the probability of logical failure $p_{\cQ} \to 0$ as the size of the code $n \to \infty$.
The logical probability of failure is defined with respect to family of ideal recovery maps.
It depends on $\pin$ and $\pphys$ and the thresholds.

Whether a threshold exists with respect to a given noise model, the exact value of the threshold, as well as how quickly the logical failure probability decreases as a function of $n$ (e.g. polynomially or exponentially), depend not only on the choice of quantum error-correcting code $\cQ_n$, but also the implementation of the syndrome-extraction circuit $C_n$ and the decoding algorithm.
In our construction, the code family is a concatenated code where the syndrome-extraction circuit is subject to constraints on geometric locality.

While the state after error correction is `close enough' to the codespace, undoing the deduced error may not correct all errors.
The remaining errors on the state are described by $\cE_{\mathrm{res}}$ that is a locally decaying distribution with failure rate $p_{\mathrm{res}}$.
We can perform another round of error correction and thereby keep the state alive for arbitrary duration if $p_{\mathrm{res}} < \pin$.
For this reason, we will specify the residual failure rate after error correction in addition to the logical failure probability $p_{\cQ}$.

\subsection{Concatenated codes}
\label{subsec:concat-review}

A concatenated code is a quantum code obtained via the composition of two codes, an inner code $\Qin$ and an outer code $\Qout$.
We consider the simple case of a $\dsl n_{0}, 1, d_{0}\dsr$ code $\Qin$ that only encodes $1$ qubit and a suitable $\dsl n,k,d\dsr$ code $\Qout$.

\definecolor{blue2}{HTML}{DBF0F5}
\definecolor{blue1}{HTML}{63A8B9}
\definecolor{blue0}{HTML}{52747D}
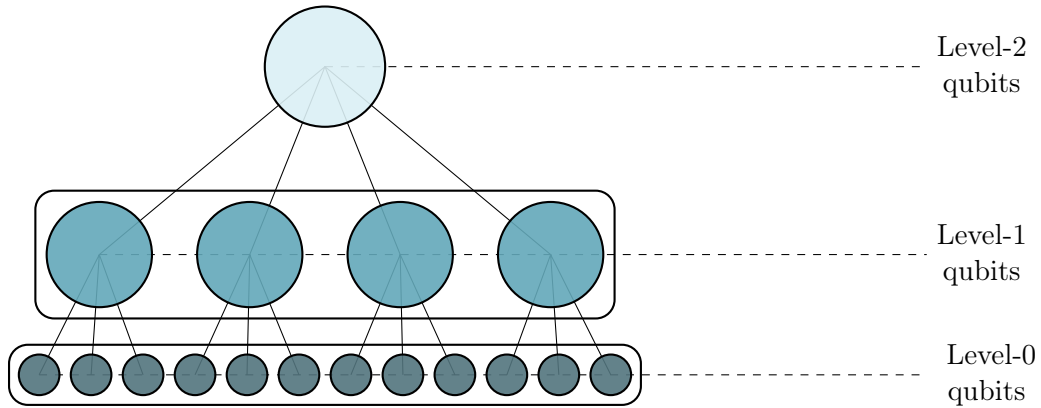
\begin{figure}[h]
    \centering
    \begin{tikzpicture}
        \tikzmath{
            \levelzerorad = 0.27;
            \levelonerad = 0.7;
            \levelzerospace = 2*3.8/11;
            \levelzeroboxpad = 0.4;
            \leveloneboxpad = 0.85;
        }

        \draw[dashed] (0,1.8)--(8,1.8) node[right,align=center] {Level-2 \\ qubits};
        \draw[dashed] (-3,-0.7)--(8,-0.7) node[right,align=center] {Level-1 \\ qubits};
        \draw[dashed] (-3.8,-2.3)--(8,-2.3) node[right,align=center] {{ Level-0} \\ qubits};
        
        \draw (0, 1.8) node(C0) {};
        \foreach \i in {0,...,3} {
            \draw[] (-3,-0.7) ++(2*\i,0) node(B\i) {};
        }
        \foreach \i in {0,...,11} {
            \draw (-3.8,-2.3) ++(\levelzerospace*\i,0) node(A\i) {};
        }

        \draw[]
            (C0.center) -- ( B0.center)
            (C0.center) -- ( B1.center)
            (C0.center) -- ( B2.center)
            (C0.center) -- ( B3.center)
            (B0.center) -- ( A0.center)
            (B0.center) -- ( A1.center)
            (B0.center) -- ( A2.center)
            (B1.center) -- ( A3.center)
            (B1.center) -- ( A4.center)
            (B1.center) -- ( A5.center)
            (B2.center) -- ( A6.center)
            (B2.center) -- ( A7.center)
            (B2.center) -- ( A8.center)
            (B3.center) -- ( A9.center)
            (B3.center) -- (A10.center)
            (B3.center) -- (A11.center);
        
        \filldraw[thick, fill=blue2, fill opacity=0.9] (C0) circle(0.8);
        \foreach \i in {0,...,3} {
            \filldraw[thick, fill=blue1, fill opacity=0.9] (B\i) circle(\levelonerad);
        }
        \foreach \i in {0,...,11} {
            \filldraw[thick,fill=blue0, fill opacity=0.9] (A\i) circle(\levelzerorad);
        }
        
        \draw[rounded corners=0.25cm, thick] ($(A0) + (-\levelzeroboxpad,-\levelzeroboxpad)$) rectangle  ($(A11) + (\levelzeroboxpad,\levelzeroboxpad)$) {};
        \draw[rounded corners=0.25cm, thick] ($(B0) + (-\leveloneboxpad,-\leveloneboxpad)$) rectangle  ($(B3) + (\leveloneboxpad,\leveloneboxpad)$) {};
    \end{tikzpicture}
    \caption{Visualizing a concatenated code $\cH$.}
    \label{fig:concat}
\end{figure}

\paragraph{Code parameters:}
The concatenated code, denoted $\Qconcat$ with parameters $\dsl N,K,D\dsr$, is constructed by replacing each qubit of the code $\Qout$ by a copy of $\Qin$, resulting in $n$ copies of the inner code $\Qin$.
The benefit of this construction is that the distance $D$ of the code $\Qconcat$ is amplified with respect to the constituent codes.
To be precise,
\begin{align*}
    N = n\cdot n_{0}~, \qquad K = k~, \qquad D = d \cdot d_{0}~.
\end{align*}
The physical qubits are referred to as Level-0 qubits, the logical qubits of $\Qin$ which form the block $\Qout$ are referred to as Level-1 qubits, and the logical qubits of $\Qconcat$ are referred to as Level-2 qubits.
See the schematic in Figure~\ref{fig:concat}.

When errors on qubits are distributed in an i.i.d.\ manner, the advantage of concatenation becomes apparent when we ``coarse grain'' details of the concatenated code.
Consider a simple setting where qubits are subject to independent $\ssX$ and $\ssZ$ errors.
Suppose we use the code $\Qout$ without concatenation.
By assumption, the probability of failure of each of the physical qubits is $p$.
However, after concatenation, the probability of failure of the Level-1 qubits is suppressed---it fails with probability proportional to $p^{d_{0}/2}$.
This is because at least $d_0/2$ errors are required to cause a logical error for $\Qin$.
The inner code thus adds an extra layer of protection and consequently, the logical failure rate for the outer code is that much lower.
As we shall see, we have to be more careful when making this sort of argument in the context of circuits.

\paragraph{Syndrome-extraction circuit:}

Let $C^{\cQ_0}$ and $C^{\cQ}$ denote the syndrome-extraction circuits for $\Qin$ and $\Qout$ respectively such that both can be implemented in $2$ dimensions using $\cK$, the set of local Clifford operations and $R$-local $\swapp$ gates.
To implement a $\cnot$ between distant qubits, we may need to permute qubits using $\swapp$ gates to bring them within range of a two-qubit gate.
We discuss to how to design such a permutation in Section~\ref{sec:design-routing}.

A syndrome-extraction circuit $\CFT$ for the concatenated code $\Qconcat$ can be expressed in terms of the syndrome-extraction circuits $C^{\cQ_0}$ and $C^{\cQ}$.
Each data and ancilla qubit in the syndrome-extraction circuit for $C^{\cQ}$ is now replaced with a copy of $\Qin$.
Each gate in $C^{\cQ}$ is replaced by the corresponding logical Clifford gate between Level-1 qubits.
Thus, even for constructing a quantum memory, we need to understand how to perform a restricted set of inner code logical operations in a fault-tolerant manner.
We perform error correction either after the logical gate or in an interleaved manner.
We discuss this in the context of our explicit architecture in Section~\ref{sec:surface-code}.

The ideal recovery map $\cRh$ for $\cH$ is obtained by first decoding the $n$ copies of the inner code $\Qin$ using $\cR_0$ and then decoding the outer code using $\cR_{\cQ}$.
Here $\cR_{\cQ_0}$ and $\cR_{\cQ}$ refer to the ideal recovery maps for $\Qin$ and $\Qout$ respectively.
Thus $\cRh = \cR_{\cQ} \circ (\cR_{\cQ_0})^{\otimes n}$.

We generalize the notion of location in the context of circuits.
A Level-1 location refers to a Level-1 gate, including the error correction rounds.
The location is faulty if it implements the incorrect logical operation on the Level-1 qubits in its support.
In the context of $C^{\cQ}$, a single Level-1 location in the circuit could refer to a $\swapp$ gate or an entangling gate or a preparation or measurement of a logical qubit of $\cQ_0$.

When ``coarse graining'' circuits for concatenated codes, more care is needed than the i.i.d.\ errors setting.
We illustrate using the following examples.

\textbf{Problem \# 1: Level-1 failure rates are not additive}

Consider a $n_0$-qubit code state of the inner code $\rho^{(0)}_{\mathrm{in}}$ with Level-0 errors $\ssE_{\mathrm{in}}$.
The error $\ssE_{\mathrm{in}}$ is not catastrophic---the ideal decoder $\cR_0$ can correct it.
The state is therefore correctable.

Consider the syndrome-extraction circuit $C^{\cQ_0}$ with Level-0 faulty locations $F$.
Suppose there is some error supported on $\supp(F)$ but that this error is not a logical error.
We may then be tempted to extend the notion of correctability to include circuits and declare the circuit $C^{\cQ_0}$ correctable.
However, this is misleading as a correctable circuit acting on a correctable input state need not produce a correctable output state.

Let $\rho^{(0)}_{\mathrm{out}}$ denote the output state and $\ssE^{(0)}_{\mathrm{out}}$ denote the errors on this state.
Suppose the faulty locations $F$ result in an error $\ssE_F$.
The product $\ssE_{\mathrm{in}} \cdot \ssE_F$ might not be correctable.
In addition, the errors $\ssE_{\mathrm{in}}$ and $\ssE_F$ can spread in unpredictable ways within the circuit.
We therefore cannot calculate the Level-1 output failure probability by merely knowing the input state and the faults individually resulted in correctable errors.
We need additional structure.

\textbf{Problem \# 2: Level-0 failure rate is not always sustainable}

Secondly, the thresholds are decoder dependent.
By definition, the ideal decoder $\cR_0$ has no faults; if the errors on the state $\rho^{(0)}_{\mathrm{out}}$ are correctable, then $\cR_{\cQ_0}$ is successful.
On the other hand, $C^{\cQ_0}$ can contain faults and may be unable to deal with as many errors as the ideal decoder $\cR_{\cQ_0}$.
This can result in instances where the output state $\rho^{(0)}_{\mathrm{out}}$ will be correctable by $\cR_0$; by our criteria for success, the output state is decodable.
However, the number of residual errors may be above the threshold for error correction.
In other words, as error correction is itself faulty, these faults can combine with existing errors to cause a logical failure.

In our construction, we address these problems in Section~\ref{subsec:coarse-grain-concat}.
We shall show that for sufficiently low failure rates, we can indeed ignore dealing with the syndrome-extraction circuit for the outer code assuming a failure rate that depends on the inner code.
This statement relies on the structure of LDPC codes and surface codes.
We now proceed to review these codes.

\subsection{Constant-rate LDPC codes}
\label{subsec:ldpc-review}

An $\dsl n,k,d\dsr$ code family $\{\cQ_n\}$ is said to be a low-density parity-check code if
\begin{enumerate}
  \item each stabilizer generator $\ssS_i$, $i \in [m]$, only acts non-trivially on at most a constant number $\Delta_g$ of qubits for all elements in $\{\cQ_n\}$.
  \item each qubit only participates in at most a constant number $\Delta_q$ of stabilizer generators for all elements in $\{\cQ_n\}$.
\end{enumerate}
To include the degree of stabilizer generators and qubits, we shall say that a code family $\{\cQ_n\}$ is an $\dsl n,k,d, \Delta_q, \Delta_g \dsr$ LDPC family.

We will choose the outer code to be a code with constant rate, i.e. $k = \Theta(n)$.
Constructing a constant-rate LDPC code is a non-trivial task because there is a conflict between the constraints on stabilizer generators.
On one hand, all stabilizer generators need to commute with each other to form a well-defined stabilizer code; on the other hand, the stabilizer generators need to have weight at most $\Delta_g$.
Despite these difficulties, there exists constant-rate quantum LDPC codes, i.e.\ $k(n) = \Theta(n)$, with distance $d(n) = \Theta(n^{\delta})$ for $0 < \delta \leq 1$.

LDPC codes have a threshold \cite{kovalev2013fault,gottesman2014fault} \emph{if} operations in $\cK$ are not subject to any locality constraints.
In this setting, we can construct syndrome-extraction circuits where each qubit is involved in a constant number of two-qubit gates.

Consider a family of $\dsl n,k,d, \Delta_q, \Delta_g \dsr$ quantum LDPC codes $\{\cQ_n\}$ where $k = \rho \cdot n$ for some constant $\rho > 0$ and distance $d = \Theta(n^{\delta})$ for some $\delta > 0$.
Suppose qubits are subject to the following errors:
\begin{enumerate}
    \item the input state is subject to locally decaying errors with failure rate per qubit $\pin$.
    \item the syndrome-extraction circuit is subject to locally decaying faults with failure rate per gate $\pphys$.
\end{enumerate}
We restate a result from Gottesman \cite{gottesman2014fault} (Theorem 4) which guarantees the existence of a threshold for arbitrary LDPC codes.
In this construction, we require $r = d(n)$ rounds of syndrome extraction.
After syndrome extraction, the (imperfect) syndromes are processed by a minimum-weight decoder $\dec$.
We do not describe the decoder in detail here and merely note that it exists.
For generic LDPC codes, the minimum-weight decoder is not necessarily efficient.

There exist $\qin$, $\qround$ in the interval $(0,1]$ such that when
\begin{equation}
    \pin \leq \qin~, \qquad \pround \leq \qround~,
\end{equation}
the following is true.

The minimum-weight decoder $\dec$ yields a correction such that:
\begin{enumerate}
  \item the final state is recoverable by an ideal recovery operator $\cR_{\cQ}$ with probability at least to $1-p_{\cQ}(n)$ where $ p_{\cQ}(n) := \exp[-\Theta(d(n))]$.
  To be precise, $p_{\cQ}(n)$ is the probability of failure \emph{per round of syndrome extraction}.
  \item the physical qubits have residual errors that are described by a locally decaying error model with failure rate at most $\pround$.
\end{enumerate}
The first condition guarantees that the probability of logical failure falls exponentially with the distance of the code.
It is worth noting that we declare a logical failure if \emph{any} logical qubit fails.
This is \emph{qualitatively} different from codes that only encode a constant number of qubits.

The second condition on the residual error is not what is in the theorem statement of Theorem 4 of \cite{gottesman2014fault}; however, the proof implies it.
For sufficiently low values of $\pphys$, it guarantees that we can continue to perform error correction for arbitrarily many rounds (conditioned on no logical errors).
In other words, we require $\pround < \pin$.

We highlight that this result applies to arbitrary LDPC codes, i.e.\ it is independent of the rate of the code.
In particular, it applies to the surface code.

We note that the threshold is stated in terms of $\pround$, and not directly in terms of $\pphys$.
This is for two reasons: (1) this is how Theorem 4 of \cite{gottesman2014fault} is itself stated, and (2) in our construction, the dependence of $\pround$ on $\pphys$ can change depending on the depth of the syndrome-extraction circuit.
Stating the thresholds in this manner will allow us to derive the functional dependence between $\pround$ and the depth of the syndrome-extraction circuit.
In Gottesman's construction \cite{gottesman2014fault}, the syndrome-extraction circuit is constant depth and therefore $\pround$ is also a constant.
In contrast, our construction is more complicated because of constraints on geometric locality.

It is known that codes defined by geometrically-local stabilizer generators in $2$ dimensions cannot achieve both constant rate and growing distance \cite{bravyi2009no,bravyi2010tradeoffs}.
To achieve a constant rate and distance $d = \Theta(n^{\delta})$ with fixed degrees $\Delta_q$ and $\Delta_g$, the amount of non-locality scales with the parameters $k$ and $d$ \cite{baspin2021connectivity,baspin2021quantifying}.
In other words, there exist $\Theta(n)$ stabilizer generators such that qubits in their support cannot be close to each other in the 2-dimensional lattice.
In the context of syndrome-extraction circuits, the result by Delfosse \emph{et al}.\ \cite{delfosse2021bounds} states that the depth of the syndrome-extraction circuit will grow when we only have geometrically-local gates and a limited number of ancilla qubits.
(Recall Equation~\eqref{eq:dbt}.)

In Section \ref{sec:ft-asymptotics}, we show that $\pround$ grows if the syndrome-extraction circuit $C$ is constrained by geometric locality.
In other words, it is \emph{not} constant and we need an approach different than Gottesman's to prove the existence of a threshold.
In our alternative approach using the hierarchical code, the growth of the inner code suppresses Level-1 logical errors sufficiently to ensure that the Level-2 logical failure rate drops rapidly as the outer LDPC code scales up. 

Finally, we discuss the choice of quantum LDPC code.
While the result above applies to generic quantum LDPC codes, more is known about specific constructions.
Quantum expander codes are one family of constant-rate quantum LDPC codes for which $d = \Theta(\sqrt{n})$ \cite{tillich2014quantum,leverrier2015quantum}.
It has been rigorously proven that these codes can be equipped with an \emph{efficient} decoder called $\ttt{small}$-$\ttt{set}$-$\ttt{flip}$ \cite{fawzi2018constant,fawzi2018efficient}.
Furthermore, it was shown that the decoder is \emph{single shot} meaning that it only requires a constant number of rounds of syndrome measurements for the decoder to function even when the syndrome is noisy.
Similar to Gottesman's requirements for the existence of a threshold, all that is needed in Fawzi \emph{et al}.\ \cite{fawzi2018constant} is for $\pround$ to remain constant.
However, unlike Gottesman's construction, it was shown that these codes have an efficient decoding algorithm that only require a constant number of rounds of syndrome measurement.
Thus, if we wish to implement a quantum expander code, we can use the same machinery presented in this paper to justify an efficient single-shot decoder for the outer code.

We refer to LDPC codes with distance scaling as $d = \Theta(n)$ as good codes.
For nearly 2 decades, it was unclear whether good codes even existed.
Following a series of breakthroughs, \cite{panteleev2020quantum,evra2020decodable,kaufman2021new,hastings2021fiber,breuckmann2021balanced}, this impasse was famously crossed first by Panteleev \& Kalachev \cite{panteleev2021asymptotically} and later by Leverrier \& Z\'emor \cite{leverrier2022quantum}.
Furthermore, these codes have the single-shot property (albeit with an inefficient decoder) as guaranteed by Quintavalle \emph{et al}.\ \cite{quintavalle2021single}.

In this paper, we do not place any constraints on the outer LDPC code other than it have constant rate $\rho >0$ and distance $d = \Theta(n^{\delta})$ for $1/2 \leq \delta \leq 1$.
Our construction works for all $\delta$, but we choose $\delta \geq 1/2$ to simplify some theorem statements.

\subsection{Surface codes}
\label{subsec:surface-review}

We consider the rotated surface code \cite{bravyi1998quantum,horsman2012surface}, arguably the simplest code that can be laid out on a $2$-dimensional lattice.
The surface code is an LDPC code, albeit with vanishing asymptotic rate.

An example is shown in Figure \ref{fig:surface}.
The code is implemented on a \emph{rotated} lattice, i.e.\ the points of the lattice correspond to the vertices of squares that run in $45$ degree angles relative to the $x$ and $y$ axes.
The points of the lattice are labeled $(a,b)$ where $a,b \in \mathbb{Z}/2$.
Each black dot represents a data qubit; these are located on integer points, i.e. on points $(a,b)$ where $(a,b) \in \mathbb{Z}$.
Each colored dot represents a syndrome qubit; these are located on half-integer points, i.e. on points $(a+1/2,b+1/2)$ where $(a,b) \in \mathbb{Z}$.
Corresponding to each blue face, we define an $\ssX$-type stabilizer generator that jointly measures $\ssX^{\otimes 4}$ on adjacent data qubits.
Similarly, corresponding to each red face, we define a $\ssZ$-type stabilizer generator that jointly measures $\ssZ^{\otimes 4}$ on adjacent qubits.
The semi-circles represent stabilizers that only act on two qubits in their support, i.e.\ they measure $\ssX^{\otimes 2}$ or $\ssZ^{\otimes 2}$ jointly.

The rotated surface code $\cRS_{\ell}$ encodes exactly one qubit and has distance $d_{\ell}$.
It uses $d_{\ell}^2$ data qubits and $d_{\ell}^2 -1$ syndrome qubits.
The total number of qubits is thus $\ell^2 = 2d_{\ell}^2-1$.
We use $\ell$ to parameterize the code family.
We also refer to each code as a tile.

\begin{figure}[h]
  \centering
  \includegraphics[scale=0.7]{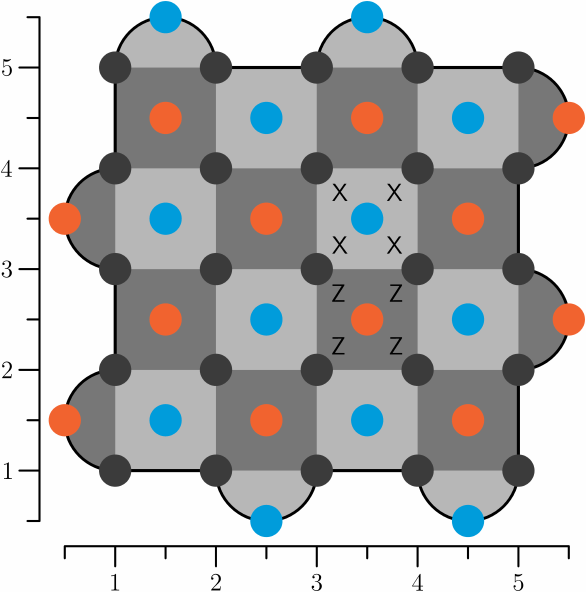}
  \caption{A surface code of distance $d_{\ell} = 5$.
  Each dark gray dot represents a data qubit.
  Light faces correspond to $\ssX$ checks and dark faces correspond to $\ssZ$ checks.
  They are measured using the qubit represented as a blue or orange dot respectively in the center of each face.
  Note that data qubits reside at integer points $\mathbb{Z}^2$ and ancilla qubits reside at the points of this lattice shifted by $(1/2, 1/2)$.}
  \label{fig:surface}
\end{figure}

Using operations $\cK$, the syndrome-extraction circuit for the surface code has depth $6$. 

\textbf{Thresholds for error correction:}
To motivate our noise model, we consider a simple setting where $n=1$, i.e.\ we have a single tile $\cRS_{\ell}$.
Suppose we are given physical qubits, each qubit in some fixed computational-basis state, and use the syndrome-measurement circuit to \emph{project} this state onto a (fixed) code state of the surface code.
If the physical qubits are subject to locally decaying errors at failure rate $\pphys^{(0)}$, we can derive the Level-1 probability of failure $p_{\cRS}^{(1)}(\ell)$ for the surface code.
The superscripts denote the noise on Level-1 and Level-0 qubits respectively.

Contrast this with the scenario where we obtain the surface code from another party.
Upon receipt, we are only informed that the tile has already failed with failure rate $\pin^{(1)}$; we do not have additional information, such as syndrome histories from prior rounds of error correction.
If the code has not \emph{already} failed, then we are guaranteed that the physical failure rate is $\pphys^{(0)}$.
Failure after error correction can thus result in two ways: either the tile fails prior to us receiving the state with probability $\pin^{(1)}$ \emph{or} conditioned on it being correct, it fails because of error correction with probability $p_{\cRS}^{(1)}(\ell) = \exp(-c_{\mathrm{EC}} \cdot \ell)$ for some positive number $c_{\mathrm{EC}}$ that does not depend on $\ell$.
The Level-1 failure rate after error correction is thus $\pin^{(1)} + p_{\cRS}^{(1)}(\ell)$.

When performing \emph{repeated} rounds of error correction, we require the Level-1 failure rate $\pin^{(1)}$ to bound the probability that the code has already failed in prior rounds.

Surface codes will form the inner code in our concatenated construction.
Consider the syndrome-extraction circuit $C$ for the constant-rate LDPC code $\cQ_n$.
Suppose $\Space = \Space(C)$ is the width of the circuit.
We require an arrangement of $\cRS_{\ell}^{\otimes \Space}$ in two dimensions.
In Section~\ref{sec:surface-code}, we introduce a \emph{bilayer} architecture for arranging tiles in two parallel layers.
We return to the explicit description of this layout in Section~\ref{sec:surface-code}.

Consider an input state of the code $\cRS_{\ell}^{\otimes \Space}$.
The errors are distributed in the following manner:
\begin{enumerate}
    \item Level-1 errors are described by $\cE^{(1)}$, a locally decaying distribution with failure rate $\pin^{(1)}$.
    \item Level-0 errors are described by $\cE^{(0)}$, a locally decaying distribution with failure rate $\pin^{(0)}$.
\end{enumerate}
Faults in the syndrome-extraction circuit are described by $\cF^{(0)}$, a Level-0 locally decaying distribution with failure rate $\pphys^{(0)}$.

The code $\cRS_{\ell}^{\otimes \Space}$ is itself an LDPC code (with vanishing rate asymptotically), and therefore, we can apply Theorem 4 from \cite{gottesman2014fault}.
We note that although the original theorem is itself is not stated in this way, the proof implies the following.

There exist thresholds $\qin^{(0)}$, $\qround^{(0)}$ on Level-0 failure rates such that, below threshold, the probability of logical failure after error correction is described by a locally decaying Level-1 error $\pin^{(1)} + p_{\cRS}^{(1)}(\ell)$ where $p_{\cRS}^{(1)}(\ell) = \exp(-c_{\mathrm{EC}} \cdot \ell)$ for some positive number $c_{\mathrm{EC}}$ that is independent of $\ell$.

In addition, the state after error correction is described by locally decaying errors with failure rate proportional to $\pround^{(0)}$.
This guarantees that if we are sufficiently below threshold, then the number of residual errors is low enough such that we can apply another round of error correction.

Unlike the case for general LDPC codes, surface codes possess a minimum weight decoder that runs in $\poly(n)$ time by mapping the decoding problem to a minimum-weight perfect matching problem.

\textbf{Logical Clifford operations:}
As highlighted in the subsection on concatenated codes, we need to implement logical Clifford operations for the surface code to be able to use it within a concatenated construction.
Extending our notation from Definition \ref{def:local-cliff}, we let $\cK_0$ denote the physical geometrically-local Clifford gates and $R$-local $\swapp$ operation on the physical qubits.
Let $C_0$ be the syndrome-extraction circuit for $\cRS_{\ell}$ constructed using $\cK_0$.

Let $\cK_1$ denote the corresponding logical operations on the surface code.
Single-tile operations in $\cK_1$ --- state preparation in a fixed stabilizer state, (destructive) measurement of logical Pauli operators and applying Pauli corrections --- can be performed using operations in $\cK_0$ regardless of how two-tile gates are implemented.
The only Clifford operations we require are two-tile operations: $\cnot$ and $R$-local $\swapp$.
These are discussed in Section~\ref{sec:surface-code}.

%% file: 3-explicit-design.tex
\section[]{Permutation routings on sparse graphs in two dimensions}
\label{sec:design-routing}

In this section, we prove Theorem~\ref{thm:permutnrouting}, restated here for convenience.
\begin{theorem*}
  For $R$ even, there is an efficient construction of a degree-12 graph $G = (V,E)$ whose vertex set $V$ is identified with an $L \times L$ lattice with edges of length at most $R$.
  Any permutation $\alpha: V \to V$ can be performed in depth $3L/R + O(\log^2 R)$.
\end{theorem*}
We use this result in the next section to construct syndrome-extraction circuits for quantum LDPC codes.
We shall study permutation routings on graphs and focus on $\nn_2(L,R)$, the $L \times L$ lattice in $2$ dimensions where two vertices share an edge if they are separated by a distance of at most $R$.
Based on the idea of a permutation routing on product graphs, we demonstrate that we can implement an arbitrary permutation in depth $O(L/R)$.
For the special case of the 2D lattice, we can make heavy use of sorting networks to find implementations of target permutations.

Using sorting networks to implement long-range connectivity is itself not a new idea~\cite{beals2013efficient}.
For instance, it was used in Delfosse {\it et al.} \cite{delfosse2021bounds} to construct syndrome-extraction circuits for quantum expander codes to match the bound in Equation~\eqref{eq:dbt}.
The results in this section generalize this idea to arbitrary syndrome-extraction circuits with constant spatial overhead.
To the best of our knowledge, this is the first work to construct sparse syndrome-extraction circuits when $R$ can scale as a function of $L$.

\subsection{Permutation Routing on product graphs}
\label{subsec:routing}

A \emph{permutation routing} is sometimes explained in terms of a pebble-exchange game, where pebbles are placed on the vertices of an (undirected) connected graph $G = (V,E)$.
The pebble on vertex $u \in V$ has an address $\alpha(u)$.
The addresses of all the pebbles together specify a permutation $\alpha$ on the vertices of $G$.
We are allowed to swap any two pebbles along an edge of $G$.
Formally, every vertex has a label and for every edge $(u,v) =: e \in E$, we are equipped with an edge permutation $\pi(e)$ that exchanges the labels of $u$ and $v$.
Edge permutations can be performed in parallel as long as every pebble is involved in at most one edge permutation in one time step.
We say $\beta$ is a simple permutation on $G$ if it is the product of edge permutations $\{\pi(e)\}_e$ that commute.
The objective of the pebble-exchange game is to find a minimum sequence of simple permutations so that the pebble that began at $u$ is located on the vertex $\alpha(u)$ afterwards.
In other words, we wish to find the smallest sequence of simple permutations $\beta_1,...,\beta_{\Depth(\alpha)}$ such that $\alpha = \beta_{\Depth(\alpha)}\circ \cdots \circ \beta_1$.
Here $\Depth(\alpha)$ denotes the minimum number of simple permutations required to perform $\alpha$.
We represent permutations using the one-line notation \cite{contributors2022permutation} where $\alpha = \begin{pmatrix} \alpha_1 & \alpha_2 & \cdots & \alpha_n \end{pmatrix}$ means $1$ is mapped to $\alpha_1$, $2$ is mapped to $\alpha_2$ and so on.
Given any permutation $\alpha$, the permutation $\alpha^{-1}$ can be computed efficiently by applying the permutation to a list of consecutive integers $[n]$.

\paragraph{The $R$-nearest-neighbor graph:}
The $R$-nearest-neighbor graph in $1$ dimension of length $L$ is denoted $\nn_1(L,R) = (V,E)$ where
\begin{equation}
    V = \{1,...,L\}, \qquad E = \{(u,v): |u-v|_2 \leq R \}~,
\end{equation}
where $|\cdot|_2$ represents the standard $2$-norm.
This is the graph for which the vertices are simply the positive integers up to $L$ and two vertices are connected by an edge if their difference is less than $R$.
In particular, consider the graph $\nn_1(L,1)$ which corresponds to the path graph. 

\textbf{Fact:} We can perform an arbitrary permutation $\alpha$ of pebbles placed on the vertices of the path graph $\nn_1(L,1)$ in depth $L-1$ \cite{knuth1997art}.
The explicit permutation routing algorithm $\ttt{Path}$-$\ttt{Routing}$ is presented in Algorithm \ref{alg:perm-rout-path}.

\begin{algorithm}[H]
    \begin{algorithmic}[0]
        \State \textbf{Input:} Permutation $\alpha$.
        \State \textbf{Output:} simple permutations $\beta_1,...,\beta_{L-1}$ such that $\alpha = \beta_{L-1} \circ \cdots \circ \beta_{1}$.
    \end{algorithmic}
    \begin{algorithmic}[1]
        \State $\ttt{labels} \leftarrow \{\alpha^{-1}(1),...,\alpha^{-1}(L)\}$
        \State $t = 1$.
        \While{$t \leq L-1$}
            \State $\beta_t \leftarrow \{1,\dots, L\}$
            \For{$i \in \{1,...,\floor{L/2}\}$}
                \State $a \leftarrow$ $2i-1$ if $t$ is even else $2i$
                \State $b \leftarrow$ $2i$\hspace{2ex}if $t$ is even else $2i+1$
                \If{$\ttt{label}(a) < \ttt{label}(b) $} \Comment{Swap}
                    \State $\beta_t(a) \leftarrow b$
                    \State $\beta_t(b) \leftarrow a$
                    \State Exchange $\ttt{label}(a)$ and $\ttt{label}(b)$.
                \EndIf
            \EndFor
        \State $t \leftarrow t+1$.
        \EndWhile
        \State \Return $\beta_1,...,\beta_{L-1}$.
    \end{algorithmic}
    \caption{$\ttt{Path}$-$\ttt{Routing}(\alpha)$}
    \label{alg:perm-rout-path}
\end{algorithm}

To illustrate, we consider a permutation $\alpha$ on the path graph on $8$ vertices in Figure~\ref{fig:permutation-on-paths}.
Here, $\alpha = \begin{pmatrix} 6 & 7& 2& 5& 3& 4& 8& 1 \end{pmatrix}$.

\begin{figure}[h]
  \centering
  \includegraphics[width=0.9\columnwidth]{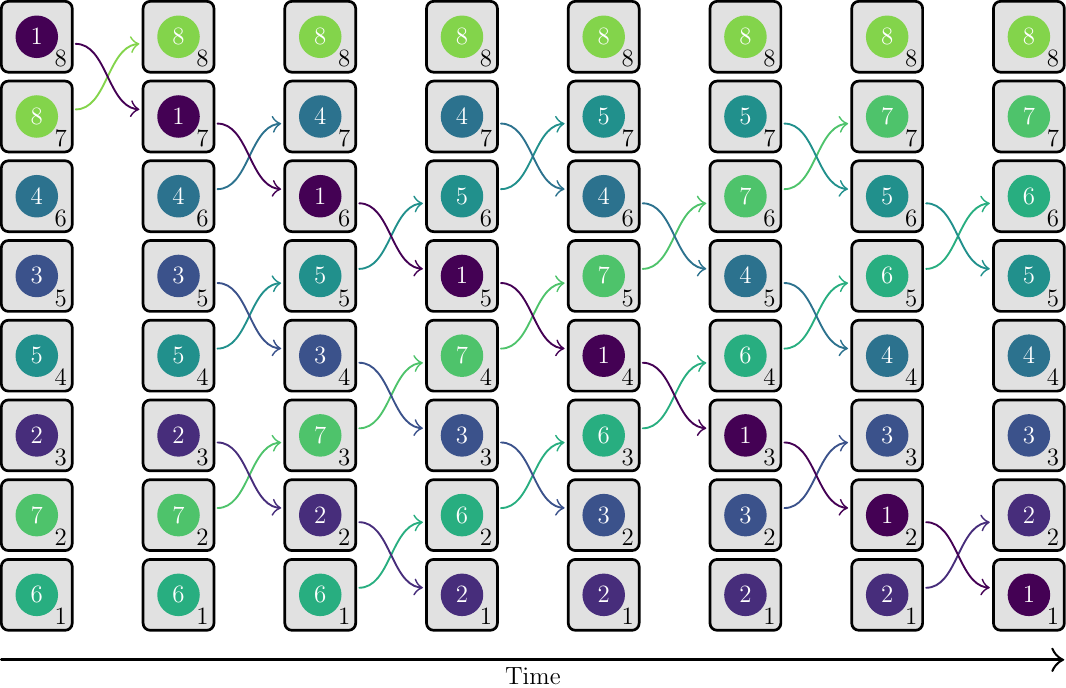}
  \caption{
    Example: implementing the permutation $\alpha = \begin{pmatrix}6 & 7 & 2 & 5 & 3 & 4 & 8 & 1\end{pmatrix}$ with nearest-neighbor swaps using Algorithm \ref{alg:perm-rout-path}.
  }
  \label{fig:permutation-on-paths}
\end{figure}

We can generalize this concept and define the $R$-nearest-neighbor graph in $2$ dimensions which we denote by $\nn_2(L,R)$.
It has vertices $\{\vec{u} : \vec{u} = (u_x,u_y) \in [L] \times [L]\}$; two vertices $\vec{u}$, $\vec{u}'$ share an edge if $|\vec{u}-\vec{u}'|_2 \leq R$.
Our objective is to build up to a routing algorithm on this graph.
Before considering this general case, we study the case where $R = 1$.
The idea used there will be used again for general $R$ in the next subsection.

\paragraph{Routing on graph products:}
The main idea we present in this subsection are techniques due to Annexstein and Baumslag \cite{annexstein1990unified} to route on Cartesian products of graphs.
They showed that we can derive routing algorithms for the Cartesian product $G_1 \times G_2$ using routing algorithms for graphs $G_1 = (V_1,E_1)$ and $G_2 = (V_2,E_2)$.

The general routing algorithm $\mathtt{Product}$-$\mathtt{Routing}$ presented in Algorithm \ref{alg:product-routing} below applies to any two graphs $G_1$ and $G_2$ for which routing routines are known.
Each $v_1 \in V_1$, defines a ``row'' $\mathscr{R}_{v_1} = \{v_1\} \times V_2$, and each $u_2 \in V_2$ defines a ``column'' $\mathscr{C}_{u_2} = V_1 \times \{u_2\}$ \footnote{This notation is inspired by thinking about the vertices arranged as a matrix of size $V_1 \times V_2$.}.
We call a permutation $\alpha$ of the vertices of $G_1\times G_2$ a \emph{row permutation} if the permutation respects a decomposition into rows i.e. for all $v_1 \in V_1$, $\alpha: \mathscr{R}_{v_1} \to \mathscr{R}_{v_1}$.
Likewise, for a \emph{column permutation}, we have that for all $u_2 \in V_2$, $\alpha: \mathscr{C}_{u_2} \to \mathscr{C}_{u_2}$.
A row or column permutation can be implemented using routing routines for $G_2$ or $G_1$ by applying the routine to each copy in the Cartesian product.

\begin{lemma}[Annexstein \& Baumslag \cite{annexstein1990unified}]
    \label{lem:routing}
    For any routing $\alpha_{12}$ on $G_1 \times G_2$, there exist column permutations $\alpha_1, \alpha_1'$ and row permutations $\alpha_2$ on $G_2$ such that:
    $\alpha_{12} = \alpha_1 \circ \alpha_2 \circ \alpha_1'$.
    These permutations can all be computed in polynomial time.
    If the permutations $\alpha_1$ and $\alpha_1'$ require depth at most $\Depth_1$ and $\alpha_{2}$ requires depth at most $\Depth_2$.
    Then $\alpha_{12}$ requires depth at most $2\Depth_1 + \Depth_2$.
\end{lemma}

We provide some intuition for this lemma.
At first glance, one might expect that a row permutation $\alpha_1$ followed by a column permutation $\alpha_2$ ought to suffice.
However, this will not always work---if two pebbles in a row share the same destination column, then no row permutation will be able to send both the pebbles to the correct column.

To avoid collisions, we start the procedure with an additional step.
We first send pebbles to rows in which no other pebbles shares the same destination column, so that the routing procedure performs a column permutation, a row permutation, and finally a column permutation.
This problem will be rephrased as an edge-coloring problem where each color corresponds to the intermediate row that qubits will be routed through. 

We first construct a bipartite multigraph $B$ over the vertices $(V_2 \sqcup V_2')$ with left and right vertex sets both copies of $V_2$ \footnote{A multigraph is a generalization of a graph where two vertices are allowed to share multiple edges}.
To each pebble, we associate an edge between the initial column on the left and the destination column on the right.
$B$ is bipartite and has degree at most $|V_1|$, so there exists an efficiently computable edge coloring with $|V_1|$ colors~\cite{schrijver2003combinatorial} i.e.\ a decomposition into $|\mathscr{R}|$ disjoint matchings.

To each color, $\tau \in [V_1]$, we will assign an arbitrary row.
For each pebble (edge), we will first pre-route it to the assigned row (color) before completing a final routing along the rows then columns.
In a valid coloring, no two edges (pebbles) of the same color (intermediate row) are incident to the same vertex (column).
In the first step, this means that, for every column, each pebble has a unique intermediate destination row.
Further, in the row permutation step, for every row, each pebble has a unique destination column.
Finally, in the last column permutation step, each pebble is in its destination column, so, for each column, each pebble has a unique destination row.

We assume we are given blackbox access to an efficient edge coloring algorithm for bipartite graphs~\cite{schrijver2003combinatorial} and call it via a subroutine $\mathtt{Edge}$-$\mathtt{Coloring}$ in Algorithm \ref{alg:product-routing}.

\begin{algorithm}[H]
    \begin{algorithmic}[0]
        \State \textbf{Input:} Permutation $\alpha \colon V_1 \times V_2 \to V_1 \times V_2$
        \State \textbf{Output:} Row permutations $\alpha_1$, $\alpha_1'$ and a column permutation $\alpha_2$ such that $\alpha = \alpha_1 \circ \alpha_2 \circ \alpha_1'$.
    \end{algorithmic}
        \begin{algorithmic}[1]
        \State Initialize bipartite graph $B \leftarrow (V_2 \sqcup V_2', \emptyset)$ with no edges.
        \For{Every $(v_1,u_2) \in V_1 \times V_2$}
            \State Draw an edge between $u_2 \in V_2$ and $u_2' \in V_2'$ if $\alpha(v_1,u_2) = (v_1',u_2')$.
        \EndFor
        \State $\tau \leftarrow \mathtt{Edge}$-$\mathtt{Coloring}(B)$
        
        \For{$(v_1,u_2) \in E$}
            \State $\alpha_1'(v_1, u_2) \leftarrow (\tau(e), u_2)$
            \State $\alpha_2(\tau(e), u_2) \leftarrow (\tau(e), u_2')$
            \State $\alpha_1(\tau(e), u_2') \leftarrow (v_1',  u_2')$
        \EndFor
    \end{algorithmic}
    \caption{$\mathtt{Product}$-$\mathtt{Routing}(\alpha)$}
    \label{alg:product-routing}
\end{algorithm}

To illustrate Algorithm \ref{alg:product-routing}, we describe how to obtain the permutation routing on the nearest-neighbor graph in $2$-dimensions $\nn_2(L,1)$.
Noting that $\nn_2(L,1)\cong \nn_1(L,1) \times \nn_1(L,1)$, Lemma \ref{lem:routing} implies that an arbitrary permutation on $\nn_2(L,1)$ can be implemented using a product of permutations on the components $\nn_1(L,1)$.
Recalling that any permutation on the path graph $\nn_1(L,1)$ can be done in depth $L-1$ implies the following corollary.

\begin{corollary}
\label{cor:depth-nn2}
    Any permutation $\alpha$ on $\nn_2(L,1)$ can be performed in depth $3L-3$.
\end{corollary}

\begin{proof}
    $\nn_2(L,1)\cong \nn_1(L,1) \times \nn_1(L,1)$, and we can route on $\nn_1(L,1)$ using an even-odd sorting network (Algorithm \ref{alg:perm-rout-path}).
    Using Algorithm \ref{alg:product-routing}, we have that the total number of steps is $3(L-1)$.
\end{proof}
In a different context, the above claim was also made Thompson \& Kung \cite{thompson1977sorting}.

Figure~\ref{fig:spacetime-route} shows an example of a permutation of such a lattice using Algorithm~\ref{alg:product-routing}.

\begin{figure}[h]
  \centering
  \includegraphics[width=\textwidth]{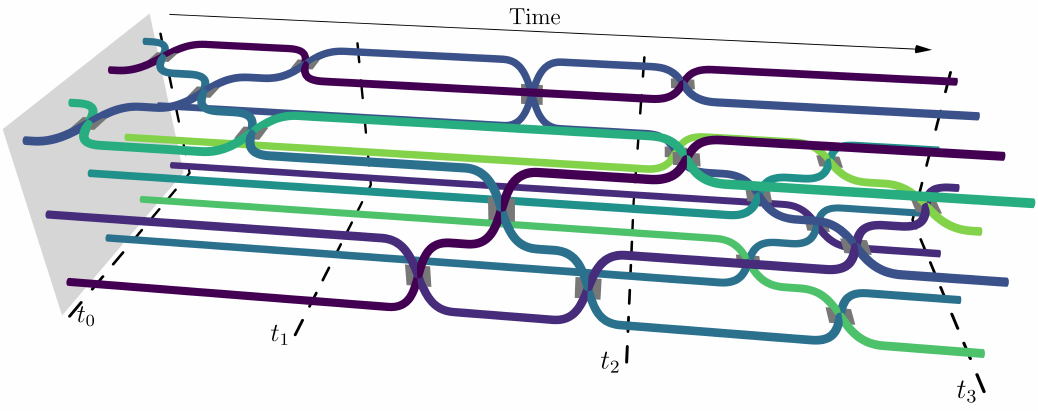}
  \caption{
    Visualizing the routing algorithm via the space-time path of individual qubits.
    For each swap location, a gray rectangle indicated the plane of the swap is drawn for visualization purposes.
    Note that in the intervals $(t_0,t_1)$ and $(t_2,t_3)$ there are only row swaps, and in the interval $(t_1,t_2)$ there are only column swaps.
    Each of the swaps within a single row or column we obtain by a 1D even-odd sorting network (Figure~\ref{fig:permutation-on-paths}).
  }
  \label{fig:spacetime-route}
\end{figure}

\subsection{Permutation routing given long-range gates}
\label{subsec:permutn-routing-long-range}

In this subsection, we show how to route on $\nn_2(L,R)$.
We do this by finding graph approximations -- subgraphs of our original graph that we can route on nearly as well.
We will approximate $\nn_2(L,R)$ using a two-step approach---first we show that the complete graph times a $2$-dimensional nearest-neighbor graph approximates $\nn_2(L,R)$ well; we then show that a sparse graph approximates the complete graph well.
Together, this will result in a circuit with sparse connectivity that exploits long-range connectivity of range $R$.

\paragraph{The complete graph \& sparse approximations:}
The complete graph $K_{m}$ is a graph on $m$ vertices with edges between every pair of vertices. 
Any permutation $\alpha$ on $K_{m}$ can trivially be accomplished in depth $2$.
Using a complete graph will simplify some of the analysis in this section.
However, $K_m$ is a dense graph; in turn, the corresponding syndrome-extraction circuit we construct from it will require that qubits are involved in a super constant number of two-qubit gates.
To avoid this problem, we replace $K_{m}$ by a sparse graph in exchange for a modest increasing in the depth of permutations.
We state some facts about sparse approximations to the complete graph $K_m$.

\begin{definition}[Spectral Expander]\label{def:expander}
Let $G$ be a $d$-regular graph on $m$ vertices where all eigenvalues of the adjacency matrix except for the largest $\{\lambda_i\}_{i=2}^m$ satisfy $\lvert \lambda_i \rvert \le \lambda < d$. $G$ is said to be an $(m, d, \lambda)$-\emph{spectral expander}.
\end{definition}

\begin{fact}[\cite{friedman2008proof}]
\label{fact:random_regular_graphs}
For even $d\ge 4$ and any $\epsilon > 0$, there is an efficient randomized\footnote{We note that the result in \cite{friedman2008proof} only shows that a regular random graph will be an expander with high probability. However, spectral expansion is efficiently checkable, so this process may be repeated until success.} algorithm that returns a random $d$-regular graph such that it is an $(m,d,\lambda)$-spectral expander for  $\lambda = 2 \sqrt{d-1} + \epsilon$.
\end{fact}

This fact establishes that a random regular graph is a good spectral expander with high probability.
For $d=4$, such a graph can be defined on any even number of vertices, so this family is extremely flexible.
For convenience, we will set $d=4$ and $m$ even.
We will call a random 4-regular graph picked in this way on $m$ vertices $\cE_m$. 

The next fact concerns routing on spectral expanders in an efficiently computable manner.

\begin{fact}[\cite{alon1993routing}]
\label{fact:sparsify}
Let $G$ be an $(m, d, \lambda)$-spectral expander.
Then, any permutation $\sigma : [m] \to [m]$ can be performed in depth $O\left(\frac{d^2}{(d-\lambda)^2}\log^2(m)\right)$.
\end{fact}

Take together, we can replace $K_m$ by a random 4-regular subgraph $\cE_m$ on which we can route in depth $O(\log^2(m))$.
For our purposes, we assume that the routing algorithm for these sparse graphs can be accessed in a black-box manner.

Now we are prepared to move on to implementing permutations on $\nn_2(L,R)$.
The depth we will find is nearly optimal, even when $R > 1$.
For convenience, let us assume that $L/R$ is an integer.
First, note that at distances shorter than $R$, $\nn_2(L,R)$ locally ``looks'' like a complete graph on $R$-vertices.
We can leverage this to find a spanning subgraph\footnote{A subgraph $H$ of a graph $G$ is said to be a spanning subgraph if all vertices of $G$ are contained in $H$.} of $\nn_2(L,R)$ that is a product of graphs that we know how to route on: $K_R$ and $\nn_2(L/R,1)$.

\begin{lemma}\label{lem:spanning_subgraph_KR}
    If $R$ divides $L$, $\nn_1(L/R,1)\times K_{R}$ is a spanning subgraph of $\nn_1(L,R)$.
\end{lemma}
\begin{proof}
    Using the coordinates $[L/R]\times [R]$ for $\nn_1(L/R,1)\times K_{R}$ and $[L]$ for $\nn_1(L,R)$, we can map the vertices of $\nn_1(L/R,1)\times K_{R}$ to those of $\nn_1(L,R)$ using the bijection $\eta \colon [L/R]\times [R] \to [L]$ 
    \begin{equation*}
        (a, b) \xrightarrow{\eta} (a - 1) R + b~.
    \end{equation*}
    Away from the boundary, the neighbors of an arbitrary vertex $(a,b)$ of $\nn_1(L/R,1)\times K_{R}$ are $(a\pm1, b)$ and $(a, [R]\setminus\{b\})$.
    A vertex $(a,b)$ at a boundary is adjacent to the vertices $(a, [R]\setminus\{b\})$ and one of $(a-1, b)$ or $(a+1, b)$; whichever is in the graph.
    All elements of these sets are at most a distance $R$ from $(a,b)$ under $\eta$, so it is a valid edge in $\nn_1(L,R)$. %
\end{proof}

Clearly, $\nn_1(L,R)\times \nn_1(L,R)$ is a spanning subgraph of $\nn_2(L,R)$ given by retaining only those edges connecting vertices within a single row or column.

\begin{corollary}\label{cor:dense_nn2_routing}
    Any permutation on $\nn_2(L,R)$ can be performed in depth $3L/R + 9$.
\end{corollary}

\begin{proof}
    Denote $\nn_1(L/R,1)\times K_{R}$ by $H$.
    By Lemma \ref{lem:spanning_subgraph_KR}, $H$ is a spanning subgraph of $\nn_1(L,R)$, and $\nn_1(L,R) \times \nn_1(L,R)$ is a spanning subgraph of $\nn_2(L,R)$.
    It follows that $H\times H$ is a spanning subgraph of $\nn_2(L,R)$, so any simple permutation for $H \times H$ is a simple permutation for $\nn_2(L,R)$.
    
    Permutations on $K_{R}$ and $G_2=\nn_1(L/R,1)$ can be implemented in depth 2 and $L/R-1$, respectively. Setting $G_1 = K_{R}$ and $G_2=\nn_1(L/R,1)$ in Lemma \ref{lem:routing}, any permutation on $H$ can therefore be implemented in depth $(L/R-1) + 2\cdot 2 = L/R + 3$. 
    Invoking Lemma \ref{lem:routing} again with $G_1=G_2=H$, we can implement any permutation on $H\times H$ and hence, $\nn_2(L,R)$ in depth $3 L/R + 9$.
\end{proof}

Note even though there are many edges that we do not use, the lowest depth routing can be no better than the graph diameter of $\nn_2(L,R)$ which is roughly $\sqrt{2} L/R$, so this is nearly optimal.
Furthermore, owing to the translation invariance of $\nn_2(L/R,1)$, the embedding of $K_{R} \times K_{R} \times \nn_2(L/R,1)$ is also translation invariant---far from the boundaries, the graph remains the same locally under translation of $R$ units in the vertical and horizontal directions.

Now, $\nn_2(L,R)$ is not sparse: the degree of each vertex grows as $R^2$.
For practical purposes, it would be convenient if only a sparse subgraph of $\nn_2(L,R)$ were used in the routing routine.
In Corollary \ref{cor:dense_nn2_routing}, we used only the edges contained in a subgraph $\nn_1(L/R,1) \times K_{R} \times \nn_1(L/R,1) \times K_{R}$.
$\nn_1(L/R,1)$ is sparse, but $K_{R}$ is not.
However, we can replace $K_{R}$ by a sparse expander graph so that the subgraph we use is sparse.
Fact \ref{fact:sparsify} supplies such a family of graphs and a depth $O(\log^2 R)$ routing subroutine.

\begin{figure}[h]
  \label{fig:lattice_embedding}
  \centering
  \includegraphics[width=\textwidth]{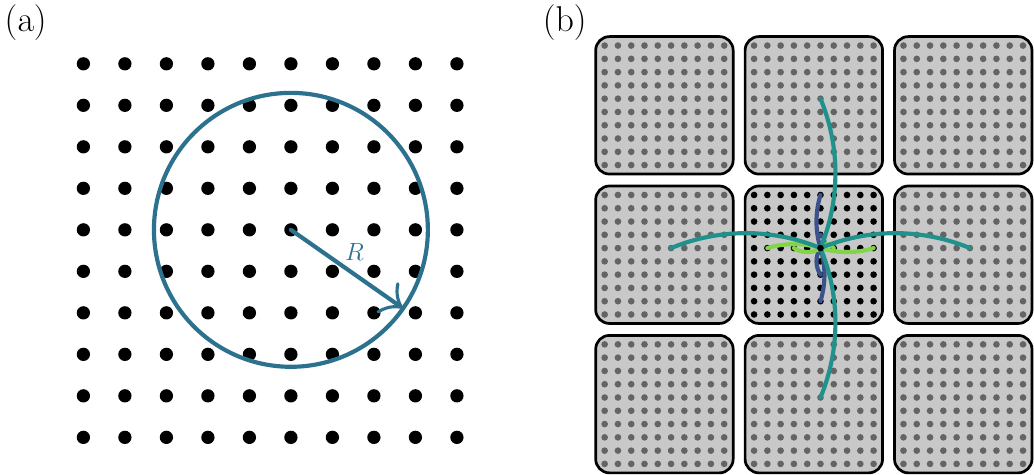}
  \caption{(a) A square lattice $\nn_2(L,R)$ with a qubit on each lattice point.
  The blue circle of radius $R$ denotes the interaction radius for a qubit in the lattice.
  Such a circle exists around each qubit; we only draw one for clarity.
  (b) Approximating the lattice using the sparse product graph $\cE_R \times \cE_R \times \nn_2(L/R, 1)$.
  Different colored edges come from different factors in the product.
  }
  \label{fig:latticelr}
\end{figure}

We now bring these ideas together formally in the following corollary.
\begin{corollary}
  \label{cor:sparse-depth-perm}
  For $R$ even, there is an efficiently constructable degree-12 spanning subgraph of $2$-dimensional $R$-nearest-neighbor lattice $\nn_2(L,R)$ on which any permutation can be performed in depth $3L/R + O(\log^2 R)$. 
\end{corollary}
\begin{proof}
    We will use replace the use of the fully connected graph in Corollary~\ref{cor:dense_nn2_routing} with a sparse expander.
    Consider a 4-regular random graph $\cE_{R}$ generated according to fact~\ref{fact:random_regular_graphs}.
    By fact~\ref{fact:sparsify}, we can route on $\cE_{R}$ in depth $O(\log^2 R)$.
    
    Now consider the graph $H=\cE_{R}\times \nn_1(L/R,1)$.
    $\cE_{R}$ is a spanning subgraph of $K_R$, so it follows by lemma \ref{lem:spanning_subgraph_KR} that $H$ is a spanning subgraph of $\nn_1(L/R,1)$.
    Further, using Lemma \ref{lem:routing}, we can implement any permutation on $H$ in depth $L/R + O(\log^2 R)$, so we can also implement any permutation on $H\times H$ in depth $3L/R + O(\log^2 R)$.
    
    By an identical argument to Corollary \ref{cor:dense_nn2_routing}, we have that $H\times H$ is a spanning subgraph of $\nn_2(L,R)$.
    Further, $H\times H$ has vertex degree 12\footnote{Degree-8 can be achieved by replacing the $\cE_{R}\times\cE_{R}$ factor in the decomposition with a single expander graph after some modification of parameters.} since the max vertex degree of the product of graphs is the sum of the max vertex degrees of the factors.
\end{proof}

This subgraph is illustrated in figure \ref{fig:latticelr} (b).
Later, we will use the contents of the corollary to construct syndrome-extraction circuits with two-qubit gates of range at most $R$ and where each qubit only needs to interact with a constant number of other qubits.

%% file: 4-bilayer-arch.tex
\section{Bilayer implementation of hierarchical codes}
\label{sec:surface-code}

In this section, we will prove Theorem~\ref{thm:hierarchical-ppties}.
Note that \(R\) can be taken to be any number (such as 1) and the required connectivity is always sparse.
We state Theorem~\ref{thm:hierarchical-ppties} here in two parts and present the proof for each part in turn.

\begin{theorem}[Theorem~\ref{thm:hierarchical-ppties}, Part 1]
\label{thm:hierarchical-ppties-1}
     The $\dsl N,K,D \dsr$ hierarchical code $\cH_N$ is constructed by concatenating an outer code, a constant-rate $\dsl n,k,d, \Delta_q, \Delta_g\dsr$ quantum LDPC code $\cQ_n$ and an inner code, a rotated surface code $\cRS_{\ell}$ where $d_{\ell} = \Theta(\log(n))$.
    Let $\rho > 0$ and $\delta \geq 1/2$, such that $k = \rho \cdot n$ and $d = \Theta(n^{\delta})$.
    The code $\cH_N$ has parameters
    \begin{align*}
        K(N) = \Theta\left(\frac{N}{\log(N)^2} \right)~, \qquad
        D(N) = \Omega\left(N^{\delta}/ \log^{2\delta-1} \left[ \frac{N}{\log(N)} \right]\right)~.
    \end{align*}
\end{theorem}
\begin{proof}
We first define the hierarchical code family $\{\cH_N\}$ and corresponding syndrome-extraction circuits $\{\CFT_N\}$.
The element of this family indexed by $N = N(n)= n \cdot d_{\ell}^2$ is created by concatenating an outer LDPC code $\cQ_n$ with an inner surface code $\cRS_{\ell}$, where $\ell = \Theta(\log(n))$.
Recall $\ell = 2d_{\ell}^2-1$ is the total number of qubits used to construct the rotated surface code $\cRS_{\ell}$.
The justification for this choice of $\ell$ will follow in the next section.

To express $n$ in terms of $N$, the following bounds will be useful:
\begin{align}
\label{eq:bounds-n-from-N}
    n = O\left(\frac{N}{\log(N)}\right)~, \qquad n = \Omega\left(\frac{N}{\log^2(N)}\right)~.
\end{align}

It follows from the definition of a concatenated code that the number of encoded qubits is $K = k(n)$, the code distance is $D = d(n) \cdot d_{\ell}$ (See Section~\ref{subsec:concat-review}).
As $\delta \geq 1/2$, $1-2\delta \leq 0$.
Using Equation~\eqref{eq:bounds-n-from-N}, we can write
\begin{align}
    K = \Omega\left( \frac{N}{\log^2(N)} \right)~, \qquad D(N) = \Omega\left(N^{\delta}\;\log^{1-2\delta}\left[\frac{N}{\log(N)} \right]\right)~.
\end{align}
This completes the proof.
\end{proof}

The second portion of Theorem~\ref{thm:hierarchical-ppties}, stated below, guarantees the existence of a syndrome-extraction circuit for the hierarchical code constructed in Theorem~\ref{thm:hierarchical-ppties-1}.

\begin{theorem}[Theorem~\ref{thm:hierarchical-ppties}, Part 2]
\label{thm:hierarchical-ppties-2}
    There exists an explicit and efficient construction of an associated family of syndrome-extraction circuits $\CFT_N$ using only local Clifford operations and $\swapp$ gates of range $R$ such that
    \begin{align*}
        \Space(\CFT_N) = \Theta(N)~, \qquad \Depth(\CFT_N) = O\left( \frac{\sqrt{N}}{R}  \right)~.
    \end{align*}
\end{theorem}

The rest of this section is dedicated to the proof of Theorem~\ref{thm:hierarchical-ppties-2}.
We construct the syndrome-extraction circuit $\CFT_N$ for the concatenated code with the stated parameters.
This circuit is constructed in a bilayer architecture and is described in detail below.
A bilayer construction of the syndrome-extraction circuit $C_n^{\cQ}$ is described in Section~\ref{subsec:sec-from-routing}.
To obtain $\CFT_N$, each outer qubit in the syndrome-extraction circuit $C_n^{\cQ}$ is replaced by a copy of the inner code $\cRS_{\ell}$ as described in Section~\ref{subsec:concat-review}.

If $C_n^{\cQ}$ requires $\Space = \Space(C_n^{\cQ})$ qubits, then we need a layout for $\Space$ surface codes in $2$ dimensions.
In Section~\ref{subsec:inner-code-construction}, we propose an implementation of $\cRS_{\ell}^{\otimes \Space}$ using a bilayer $2$-dimensional architecture.
The advantage of this architecture is that entangling gates between codes can be performed in a transversal manner which reduces the number of extra ancilla qubits.
In Section~\ref{subsec:swap-bilayer}, we describe a novel implementation of $\swapp$ gates for this architecture.
This completes the set of logical Clifford operations $\cK_1$.
We will bring these elements together to construct the syndrome-extraction circuit $\CFT_N$ in Section~\ref{subsec:sec-hierarchical}.

Before doing so, we take a brief detour in Section~\ref{subsec:bilayer-biased} to design Level-1 qubits with noise bias.
We will return to this construction in Section~\ref{sec:numerical estimates} to deal with hook errors.

\subsection[]{Syndrome-extraction circuit $C_n^{\cQ}$ for the outer code}
\label{subsec:sec-from-routing}

In this section, we design a family of syndrome-extraction circuits $\{C_n^{\cQ}\}$ for a constant-rate $\dsl n,k,d, \Delta_q, \Delta_g\dsr$ LDPC code $\{\cQ_n\}$.
We assume $k = \rho \cdot n$ for $\rho > 0$ and that $m = n-k$ is the number of stabilizer generators.
In Section~\ref{subsec:basic_defs}, we described measurement gadgets to measure each stabilizer generator.
We now describe how these gadgets can be implemented in parallel subject to constraints on geometric locality.
We first state the existence of an \emph{ideal} circuit $(C_n^{\cQ})^{\mathrm{ideal}}$ which is not constrained by geometric locality.
While we include the proof of this construction for the sake of completeness, we note that the idea and the proof itself have been used before---for example, see \cite{delfosse2021bounds}.
For this reason, the proof is relegated to Appendix~\ref{app:proof-ideal-sec}.

Define constants $\Delta$ and $m_0$ such that
\begin{align*}
    \Delta := \max(\Delta_q, \Delta_g)~, \qquad m_0 := \max(m_{\ssX},m_{\ssZ})~.
\end{align*}
The circuit $(C_n^{\cQ})^{\mathrm{ideal}}$ is divided into two phases, where in each phase we measure either $\ssX$ or $\ssZ$ syndromes.
Each phase requires at most $(\Delta+2)$ stages.
It satisfies
\begin{align}
    \Space := \Space[(C_n^{\cQ})^{\mathrm{ideal}}] = n + m_0~, \qquad
    s := \Depth[(C_n^{\cQ})^{\mathrm{ideal}}] = 2(\Delta+2)~.
\end{align}

The first phase proceeds as follows:
\begin{enumerate}
  \item All $m_{\ssX}$ ancilla qubits are prepared in the state $\ket{+}$.
  \item In each intermediate step $1 < t \leq \Delta+1$, there is a subset $P_t$ of all $\Space$ qubits such that $P_t$ is a disjoint union of $m_{\ssX}$ pairs of qubits, where each pair contains one ancilla and one data qubit respectively. These pairs correspond to control and target qubits, respectively, for $\cnot$.
  \item All $m_{\ssX}$ ancilla qubits are measured in the $\ssX$ basis.
\end{enumerate}

The second phase is structurally similar with minor modifications:
\begin{enumerate}
  \item All $m_{\ssZ}$ ancilla qubits are prepared in the state $\ket{0}$.
  \item In each intermediate step $1 < t \leq \Delta+1$, there is a subset $P_t$ of all $\Space$ qubits such that $P_t$ is a disjoint union of $m_{\ssZ}$ pairs of qubits, where each pair contains one ancilla and one data qubit respectively. These pairs correspond to target and control qubits, respectively, for $\cnot$.
  \item All $m_{\ssZ}$ ancilla qubits are measured in the $\ssZ$ basis.
\end{enumerate}

We now use this circuit to construct the syndrome-extraction circuit $C^{\cQ}_n$ that is constrained by geometric locality.
It will have the same space footprint $\Space$, but its depth will be different.

\textbf{Setup:} Qubits are arranged in two parallel layers where each layer is a grid of dimensions $L \times L$.
We assume we have access to Clifford operations $\cK$ where, in addition to nearest-neighbor gates between two qubits in the same layer, we can perform nearest-neighbor gates between two qubits that are adjacent but in different layers.
We also assume that $\swapp$ gates of range $R > 1$ are restricted to a single layer.

\textbf{Initialization:} To accommodate all $\Space$ qubits required for $(C_n^{\cQ})^{\rm ideal}$, it is sufficient to choose the smallest integer $L$ that satisfies $2L^2 \geq \Space$.
Initially, data qubits and syndrome qubits are distributed arbitrarily.
While further optimization is likely possible, this will not affect the asymptotics, and certainly results in an upper bound on the circuit volume.

\textbf{Partition $C_n^{\cQ}$ into stages:} The circuit $C_n^{\cQ}$ will be partitioned into $s$ \emph{stages}, where in each stage, we prepare and measure ancilla qubits or simulate long-range entangling gates between pairs of qubits specified by $P_t$.
To simulate a long-range entangling gate, we use a series of $\swapp$ gates which bring each pair specified by $P_t$ close together, followed by the desired entangling gate when they are sufficiently close.
The preparation and measurement stages are straightforward.
We now describe how to perform the long-range entangling gate.

\textbf{Simulating long-range:} In each simulation stage, qubits are arranged such that each pair of $P_t$ are adjacent but in different layers.
In the first step of each stage, we apply a permutation to the ancilla qubit in each pair of $P_t$ to ensure that both qubits in the pair are not in the same layer.
Qubits that are not in $P_t$ remain stationary.
For simplicity, we only permute qubits in the top layer and keep the bottom layer stationary.
This specifies a permutation $\alpha_{t}$ on the top layer.
As this layer is an $L\times L$ lattice, we can use Algorithm \ref{alg:product-routing} to design a circuit so that $\swapp$ operations can be performed in parallel.
Using Theorem~\ref{thm:permutnrouting}, we can construct a sparse spanning subgraph of $\nn_2(L,R)$ with vertex degree 12 such that any permutation $\alpha_{t}$ can be accomplished in depth at most $3L/R + O(\log^2(R))$.
This is followed by nearest-neighbor entangling gates as specified by $P_t$.
The simulation stages have depth
\begin{align}
    3\frac{L}{R} + O(\log^2(R)) + 1~.
\end{align}

Accounting for preparation and measurement steps in each phase, the circuit $C_n^{\cQ}$ has parameters
\begin{align}
\label{eq:Cn-params}
    \Space(C_n^{\cQ}) = \Space = n+m_0~, \qquad \Depth(C_n^{\cQ}) \leq 2\Delta\left(3\frac{L}{R} + O(\log^2(R)) + 1\right) + 4~.
\end{align}
We note that if $R = o(L)$, then the depth is $O(\sqrt{n}/R)$ (as $L = O(\sqrt{n}$).
We shall assume that this is the case for the rest of the paper.
We include the bound in Equation~\eqref{eq:Cn-params} with constants (ex.\ the $3$ preceding $L/R$) and the dependence on the degree $\Delta$ to highlight that, for $R=1$, this is not merely an asymptotic result and can actually be executed in practice.

The bounds in Equation~\eqref{eq:Cn-params} represent an achievability result---any quantum LDPC code can be simulated in depth $O(\sqrt{n}/R)$ as stated in the theorem above.
However, it is not asymptotically tight for all code families (for example, consider the surface code).
We expect future versions of this bound to depend on $k$ and $d$, and how they scale as functions of $n$.

\textbf{Circuit connectivity:} In addition to providing a bound on the depth of permutations, Theorem~\ref{thm:permutnrouting} guarantees that each lattice position interacts with at most 12 other locations all within a range $R$.
This implies that the connectivity of the circuit $C_n^{\cQ}$ we have constructed can be `static'---once qubits have been connected by wires of length at most $R$, we do not change it afterwards.

\subsection[]{Implementation of the inner code}
\label{subsec:inner-code-construction}

We have shown that $C_n^{\cQ}$ is constructed using two parallel layers of qubits, where each layer is a lattice of dimensions $L \times L$.
Here $L$ is the smallest integer such that $2L^2 \geq \Space$, where $\Space$ is the number of qubits used by $C_n^{\cQ}$.
To construct the syndrome-extraction circuit $\CFT_N$, we use two parallel \emph{rotated} lattices.
Each qubit in $C_n^{\cQ}$ is replaced by a rotated surface code $\cRS_{\ell}$ where $\ell = \Theta(\log(n))$.
As each tile uses $\ell^2 = 2d_{\ell}^2-1$ physical qubits, the circuit $\CFT_N$ requires at least $2L^2 \cdot \ell^2$ physical qubits.
We also use additional physical qubits which we refer to as \emph{buffer} qubits to facilitate logical Clifford operations between tiles.
See Figure~\ref{fig:cell}.

Buffer qubits are either placed along the periphery of each lattice or between tiles:
\begin{enumerate}
\item First, we include a thin band of ``buffer'' qubits along the perimeter of each layer for reasons that we will explain shortly.
The band has thickness $(\ell+1)/2$ and therefore this adds at most $2L (\ell+1)^2$ ancilla qubits per layer.
These are denoted as transparent dots in Figure~\ref{fig:cell} (a).
\item Second, for each surface code, we have $d_{\ell}^2$ data qubits, $d_{\ell}^2 - 1$ ancilla qubits, and $1$ extra buffer qubit (light gray) for later convenience.
These are denoted using dark gray, orange/blue and light gray respectively in Figure~\ref{fig:cell} (b).
\end{enumerate}

In total, accounting for both qubits used in tiles as well as buffer qubits, we use at most $2(\ell+1)^2(L + 1)^2$ physical qubits.
These are arranged in two parallel (rotated) lattices of side length $(L+1)\cdot (\ell+1)$\footnote{Note that there are $\sqrt{2} L \ell$ qubits to a side due to the rotated lattice.}.

\begin{figure}[h]
  \centering
  \begin{tikzpicture}
    \node at (-3,2.25) {(a)};
    \node at (0,0) {\includegraphics[width=0.3\columnwidth]{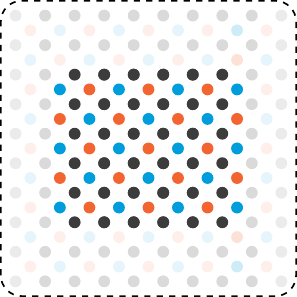}};
    \node at (4.25,2.25) {(b)};
    \node at (7.7,0) {\includegraphics[width=0.5\columnwidth]{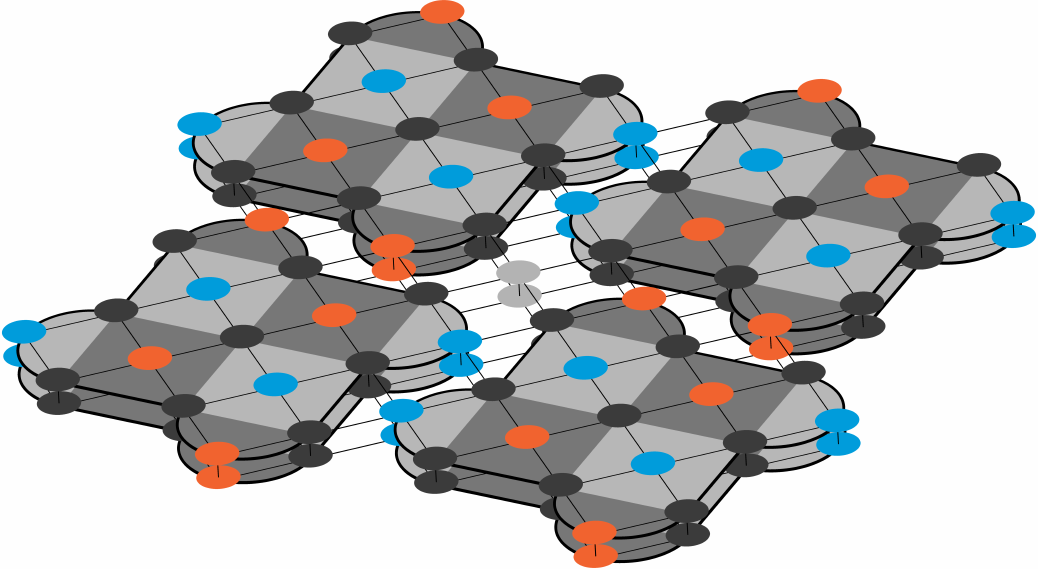}};
  \end{tikzpicture}
  \caption{
  (a) A top-down view of one layer of the physical layout.
  At any given time, only some of the qubits are active; these are denoted using dots with solid color.
  In contrast, there are qubits that are inactive that can be used for performing logical operations; these are denoted using dots that are transparent.
  (b) A small $2 \times 2 \times 2$ unit cell of the physical layout containing 8 distance-3 rotated surface code tiles.
  Physical qubits are drawn as dark gray dots.
  $\ssX$- and $\ssZ$-type stabilizer generators within a tile are indicated by a light or dark gray region with the colored dot used as an ancilla.
  Thin lines indicate gate connectivity: Each qubit has 5 neighbors, 4 in-plane and 1 out-of-plane. 
  The light gray qubit in the center is unused when the tiles are idle. Note that this layout does not contain additional ancilla qubits between tiles for lattice surgery: all operations will be performed transversally.
  }
  \label{fig:cell}
\end{figure}

Physical gates can act either on two neighboring qubits in the same layer, or on adjacent qubits in different layers.
Using only $\cK_0$ operations, we construct the necessary primitives to implement the syndrome-extraction circuits $\CFT_N$.

As described in Section~\ref{subsec:sec-from-routing}, the circuit $C_n^{\cQ}$ is divided into $s$ stages.
Implementing the syndrome-extraction circuit for the outer code requires logical Clifford operations $\cK_1$.

At the outset, both data and ancilla tiles are arranged arbitrarily.
Syndromes are measured in two phases---first we measure the $\ssX$-type syndromes and then the $\ssZ$-type syndromes.
The first and last stages of each phase correspond to single-tile logical state preparation and single-tile logical measurements.
Single-tile logical operations of state preparation and measurements can be done using only nearest-neighbor gates.
If Level-0 qubits can be prepared in $\ket{0}$ or $\ket{+}$, then we can prepare Level-1 $\ket{\overline{0}}$ and $\ket{\overline{+}}$ simply by performing the syndrome-extraction circuit which projects the state into the code space.
Similarly, we can perform destructive measurements of logical Pauli operators $\logX$ and $\logZ$ using single-qubit measurements of $\ssX$ and $\ssZ$.

For each stage where we simulate long-range entangling gates, there exists a partition of outer qubits $P_t$; Here, $P_t$ is the set of $m_0 = \max(m_{\ssX},m_{\ssZ})$ pairs, where each pair has one outer ancilla qubit and one outer data qubit respectively.
Depending on whether we are measuring $\ssX$ or $\ssZ$ syndromes, we perform the logical $\cnot$ gate using the outer data qubit as target or the outer ancilla qubit as target.
Data and syndrome tiles that are involved in entangling gates are always arranged such that the tiles are adjacent but in different layers.
As the rotated surface code is a CSS code, we can perform $\cnot$ gates between surface code blocks corresponding to data and ancilla qubits using transversal operations.
For all these operations --- single-tile preparation, single-tile measurement, and transversal $\cnot$ --- we perform surface code error correction after the operation.

To complete the description of the Level-1 syndrome extraction circuit, it only remains to explain how the logical $\swapp$ operation is implemented.
We propose a novel way to perform this gate when $R<\ell$; this is the focus of Section \ref{subsec:swap-bilayer}.
We show that an arbitrary permutation requires $O(L \cdot d_{\ell})$ steps.
Error correction for $\swapp$ gates is performed in an \emph{interleaved} manner---we perform a single round of error correction after each step as described below.

We conclude Section~\ref{subsec:inner-code-construction} by discussing connectivity requirements.
As mentioned in the introduction, we construct syndrome-extraction circuits such that once two lattice positions have been connected, this does not need to change dynamically over the course of the circuit.
Furthermore, each lattice position only ever interacts with a constant-sized set of other lattice locations.
For both error correction and logical operations, lattice positions (that store a physical qubit) will be involved in $\cnot$ gates.
It would be preferable if the pairs of lattice positions that need to interact did not change dynamically over the course of the circuit, and instead could be chosen ahead of time.

If the circuit $C_n^{\cQ}$ is implemented such that each lattice position requires only connectivity to a constant sized set of other lattice positions, then the entire syndrome-extraction circuit $\CFT_N$ for the hierarchical code $\cH_N$ will only use sparse connectivity of two-qubit physical gates.
The proof of this claim is straightforward for single-tile logical state preparation and measurement --- these are accomplished using local physical operations.
Secondly, by construction, logical entangling gates are implemented transversally and therefore the connectivity does not change.
Finally, we will show in Section~\ref{subsec:swap-bilayer} that for a given $L$ and $R$, the connectivity required to implement an arbitrary permutation of tiles will be chosen ahead of time and will not change dynamically.

\subsection[SWAP gate]{$\swapp$ Gate}
\label{subsec:swap-bilayer}
As discussed above, we can perform logical $\cnot$ transversally. 
To complete $\cK_1$, the final ingredient we need are $\swapp$ gates.
We first focus on the special case $R=1$, and then generalize the construction to arbitrary $R$.
To implement the permutation returned by Algorithm \ref{alg:product-routing}, it suffices to perform $\swapp$ gates only along one orientation of the lattice at a time, either vertically or horizontally.
This restriction is utilized to create a resource efficient $\swapp$ operation that requires no additional ancilla qubits.
The key insight is that movement of individual tiles may be accomplished by moving \emph{all} the tiles within a single layer.
By performing the transversal $\swapp$ after movement and then moving back, we can accomplish a $\swapp$ operation between two surface codes that are not directly on top of each other.

\subsubsection[Nearest-neighbor logical SWAP gates]{Nearest-neighbor logical $\swapp$ gates}
\label{subsubsec:nn-swap}

\textbf{High-level overview:}
We first provide a high-level overview of the $\swapp$ operation and refer to Figure~\ref{fig:sideview}.
Consider a two parallel rows of tiles in the bilayer architecture as shown in Figure~\ref{fig:sideview} (a).
The tiles in the top row are labeled $a_1$,..., $a_4$ and the tiles in the bottom row are labeled $b_1$,...,$b_4$.
The tiles labeled $\emptyset$ are buffer qubits along the periphery.
For ease of visualization, the picture depicts a single-tile width of buffer qubits along the top layer.
In practice, we use two half-tile width of buffer qubits in both layers.
In this example, we swap tiles $a_2$ and $a_3$; however, this process can be generalized to swap tiles in parallel.
This is accomplished in $5$ steps:
\begin{enumerate}
\item The logical $\swapp$ operation begins by exchanging alternate tiles between two rows.
In the bilayer architecture, this exchange is performed in a checkerboard pattern, i.e.\ we swap alternate tiles along both rows and columns.
This can be accomplished using what we call the \emph{staggered $\swapp$ primitive}.
\item We can then slide the entire top layer one tile width to the left.
In the bilayer architecture, this will be accomplished using what we call the \emph{walking primitive} that we describe below.
The top layer will shift a half-tile width in one direction while the bottom layer will move a half-tile width in the other.
This is the reason we use buffer qubits along the periphery---to accommodate tiles after the walk step.
\item Pairs of tiles that we wish to exchange are now adjacent in different layers.
For each pair that we wish to swap, we perform a swap operation using nearest-neighbor $\swapp$ gates between adjacent layers.
\item The last two steps are the inverse of the first two steps---we apply the walk primitive and then perform a swap operation between layers on alternate tiles.
\end{enumerate}

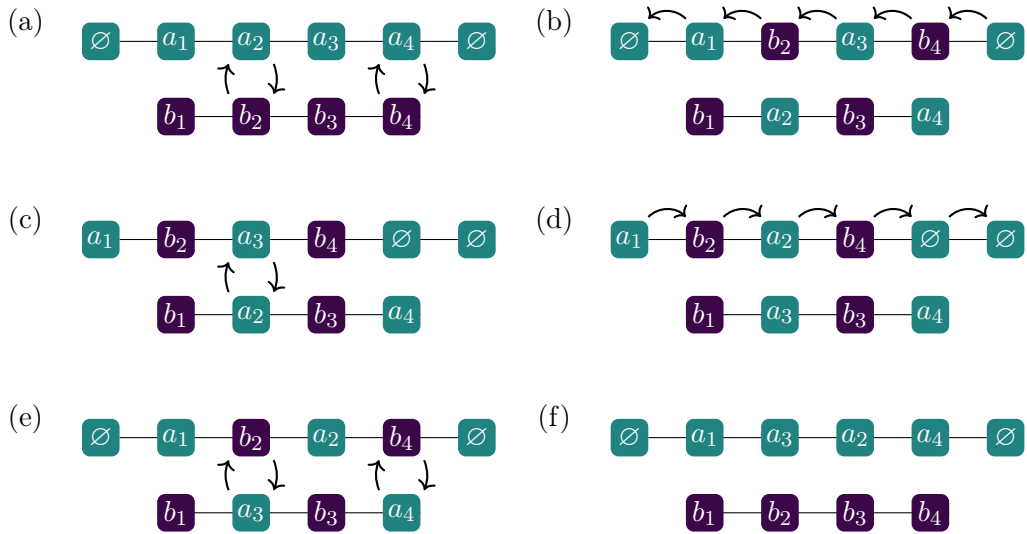
\begin{figure}[h]
    \centering
    \begin{tikzpicture}
    \node at (-0.75, 1.5) {(a)};
    \draw (1,0.25)--(4.5,0.25);
    \draw (0.5,1.25)--(5,1.25);

    \fill[tiletop,text=white,rounded corners=3pt] (0,1) rectangle node[align=center] {$\emptyset$} (0.5,1.5);
    \fill[tiletop,text=white,rounded corners=3pt] (1,1) rectangle node[align=center] {$a_{1}$} (1.5,1.5);
    \fill[tilebot,text=white,rounded corners=3pt] (1,0) rectangle node[align=center] {$b_{1}$} (1.5,0.5);
    \fill[tiletop,text=white,rounded corners=3pt] (2,1) rectangle node[align=center] {$a_{2}$} (2.5,1.5);
    \fill[tilebot,text=white,rounded corners=3pt] (2,0) rectangle node[align=center] {$b_{2}$} (2.5,0.5);
    \fill[tiletop,text=white,rounded corners=3pt] (3,1) rectangle node[align=center] {$a_{3}$} (3.5,1.5);
    \fill[tilebot,text=white,rounded corners=3pt] (3,0) rectangle node[align=center] {$b_{3}$} (3.5,0.5);
    \fill[tiletop,text=white,rounded corners=3pt] (4,1) rectangle node[align=center] {$a_{4}$} (4.5,1.5);
    \fill[tilebot,text=white,rounded corners=3pt] (4,0) rectangle node[align=center] {$b_{4}$} (4.5,0.5);
    \fill[tiletop,text=white,rounded corners=3pt] (5,1) rectangle node[align=center] {$\emptyset$} (5.5,1.5);
    
    \path[->,thick] (2.55,0.95) edge [bend left=20] (2.55,0.55);
    \path[->,thick] (1.95,0.55) edge [bend left=20] (1.95,0.95);
    \path[->,thick] (4.55,0.95) edge [bend left=20] (4.55,0.55);
    \path[->,thick] (3.95,0.55) edge [bend left=20] (3.95,0.95);
    \begin{scope}[xshift=200]
    \node at (-0.75, 1.5) {(b)};
    \foreach \x in {1,...,5} {
        \path[->,thick] ({\x},1.55) edge [bend right=45] ({\x-0.5},1.55);
    }
    \draw (1,0.25)--(4.5,0.25);
    \draw (0.5,1.25)--(5,1.25);
    \fill[tiletop,text=white,rounded corners=3pt] (0,1) rectangle node[align=center] {$\emptyset$} (0.5,1.5);
    \fill[tiletop,text=white,rounded corners=3pt] (1,1) rectangle node[align=center] {$a_{1}$} (1.5,1.5);
    \fill[tilebot,text=white,rounded corners=3pt] (1,0) rectangle node[align=center] {$b_{1}$} (1.5,0.5);
    \fill[tilebot,text=white,rounded corners=3pt] (2,1) rectangle node[align=center] {$b_{2}$} (2.5,1.5);
    \fill[tiletop,text=white,rounded corners=3pt] (2,0) rectangle node[align=center] {$a_{2}$} (2.5,0.5);
    \fill[tiletop,text=white,rounded corners=3pt] (3,1) rectangle node[align=center] {$a_{3}$} (3.5,1.5);
    \fill[tilebot,text=white,rounded corners=3pt] (3,0) rectangle node[align=center] {$b_{3}$} (3.5,0.5);
    \fill[tilebot,text=white,rounded corners=3pt] (4,1) rectangle node[align=center] {$b_{4}$} (4.5,1.5);
    \fill[tiletop,text=white,rounded corners=3pt] (4,0) rectangle node[align=center] {$a_{4}$} (4.5,0.5);
    \fill[tiletop,text=white,rounded corners=3pt] (5,1) rectangle node[align=center] {$\emptyset$} (5.5,1.5);
    \end{scope}
    \begin{scope}[yshift=-75]
    \node at (-0.75, 1.5) {(c)};
    \draw (1,0.25)--(4.5,0.25);
    \draw (0.5,1.25)--(5,1.25);

    \fill[tiletop,text=white,rounded corners=3pt] (5,1) rectangle node[align=center] {$\emptyset$} (5.5,1.5);
    \fill[tiletop,text=white,rounded corners=3pt] (0,1) rectangle node[align=center] {$a_{1}$} (0.5,1.5);
    \fill[tilebot,text=white,rounded corners=3pt] (1,1) rectangle node[align=center] {$b_{2}$} (1.5,1.5);
    \fill[tilebot,text=white,rounded corners=3pt] (1,0) rectangle node[align=center] {$b_{1}$} (1.5,0.5);
    \fill[tiletop,text=white,rounded corners=3pt] (2,1) rectangle node[align=center] {$a_{3}$} (2.5,1.5);
    \fill[tiletop,text=white,rounded corners=3pt] (2,0) rectangle node[align=center] {$a_{2}$} (2.5,0.5);
    \fill[tilebot,text=white,rounded corners=3pt] (3,1) rectangle node[align=center] {$b_{4}$} (3.5,1.5);
    \fill[tilebot,text=white,rounded corners=3pt] (3,0) rectangle node[align=center] {$b_{3}$} (3.5,0.5);
    \fill[tiletop,text=white,rounded corners=3pt] (4,1) rectangle node[align=center] {$\emptyset$} (4.5,1.5);
    \fill[tiletop,text=white,rounded corners=3pt] (4,0) rectangle node[align=center] {$a_{4}$} (4.5,0.5);

   \path[->,thick] (2.55,0.95) edge [bend left=20] (2.55,0.55);
    \path[->,thick] (1.95,0.55) edge [bend left=20] (1.95,0.95);
    \end{scope}
    \begin{scope}[xshift=200,yshift=-75]
    \node at (-0.75, 1.5) {(d)};
    \foreach \x in {0,...,4} {
        \path[->,thick] ({\x+0.5},1.55) edge [bend left=45] ({\x+1},1.55);
    }
    \draw (1,0.25)--(4.5,0.25);
    \draw (0.5,1.25)--(5,1.25);
        
    \fill[tiletop,text=white,rounded corners=3pt] (0,1) rectangle node[align=center] {$a_{1}$} (0.5,1.5);
    \fill[tilebot,text=white,rounded corners=3pt] (1,1) rectangle node[align=center] {$b_{2}$} (1.5,1.5);
    \fill[tilebot,text=white,rounded corners=3pt] (1,0) rectangle node[align=center] {$b_{1}$} (1.5,0.5);
    \fill[tiletop,text=white,rounded corners=3pt] (2,1) rectangle node[align=center] {$a_{2}$} (2.5,1.5);
    \fill[tiletop,text=white,rounded corners=3pt] (2,0) rectangle node[align=center] {$a_{3}$} (2.5,0.5);
    \fill[tilebot,text=white,rounded corners=3pt] (3,1) rectangle node[align=center] {$b_{4}$} (3.5,1.5);
    \fill[tilebot,text=white,rounded corners=3pt] (3,0) rectangle node[align=center] {$b_{3}$} (3.5,0.5);
    \fill[tiletop,text=white,rounded corners=3pt] (4,1) rectangle node[align=center] {$\emptyset$} (4.5,1.5);
    \fill[tiletop,text=white,rounded corners=3pt] (4,0) rectangle node[align=center] {$a_{4}$} (4.5,0.5);
    \fill[tiletop,text=white,rounded corners=3pt] (5,1) rectangle node[align=center] {$\emptyset$} (5.5,1.5);

    \end{scope}
    \begin{scope}[yshift=-150]
    \node at (-0.75, 1.5) {(e)};
    \draw (1,0.25)--(4.5,0.25);
    \draw (0.5,1.25)--(5,1.25);
        
    \fill[tiletop,text=white,rounded corners=3pt] (0,1) rectangle node[align=center] {$\emptyset$} (0.5,1.5);
    \fill[tiletop,text=white,rounded corners=3pt] (1,1) rectangle node[align=center] {$a_{1}$} (1.5,1.5);
    \fill[tilebot,text=white,rounded corners=3pt] (1,0) rectangle node[align=center] {$b_{1}$} (1.5,0.5);
    \fill[tilebot,text=white,rounded corners=3pt] (2,1) rectangle node[align=center] {$b_{2}$} (2.5,1.5);
    \fill[tiletop,text=white,rounded corners=3pt] (2,0) rectangle node[align=center] {$a_{3}$} (2.5,0.5);
    \fill[tiletop,text=white,rounded corners=3pt] (3,1) rectangle node[align=center] {$a_{2}$} (3.5,1.5);
    \fill[tilebot,text=white,rounded corners=3pt] (3,0) rectangle node[align=center] {$b_{3}$} (3.5,0.5);
    \fill[tilebot,text=white,rounded corners=3pt] (4,1) rectangle node[align=center] {$b_4$} (4.5,1.5);
    \fill[tiletop,text=white,rounded corners=3pt] (4,0) rectangle node[align=center] {$a_{4}$} (4.5,0.5);
    \fill[tiletop,text=white,rounded corners=3pt] (5,1) rectangle node[align=center] {$\emptyset$} (5.5,1.5);

    \path[->,thick] (2.55,0.95) edge [bend left=20] (2.55,0.55);
    \path[->,thick] (1.95,0.55) edge [bend left=20] (1.95,0.95);
    \path[->,thick] (4.55,0.95) edge [bend left=20] (4.55,0.55);
    \path[->,thick] (3.95,0.55) edge [bend left=20] (3.95,0.95);
    \end{scope}
    \begin{scope}[xshift=200,yshift=-150]
    \node at (-0.75, 1.5) {(f)};
    \draw (1,0.25)--(4.5,0.25);
    \draw (0.5,1.25)--(5,1.25);
        
    \fill[tiletop,text=white,rounded corners=3pt] (0,1) rectangle node[align=center] {$\emptyset$} (0.5,1.5);
    \fill[tiletop,text=white,rounded corners=3pt] (1,1) rectangle node[align=center] {$a_{1}$} (1.5,1.5);
    \fill[tilebot,text=white,rounded corners=3pt] (1,0) rectangle node[align=center] {$b_{1}$} (1.5,0.5);
    \fill[tiletop,text=white,rounded corners=3pt] (2,1) rectangle node[align=center] {$a_{3}$} (2.5,1.5);
    \fill[tilebot,text=white,rounded corners=3pt] (2,0) rectangle node[align=center] {$b_{2}$} (2.5,0.5);
    \fill[tiletop,text=white,rounded corners=3pt] (3,1) rectangle node[align=center] {$a_{2}$} (3.5,1.5);
    \fill[tilebot,text=white,rounded corners=3pt] (3,0) rectangle node[align=center] {$b_{3}$} (3.5,0.5);
    \fill[tiletop,text=white,rounded corners=3pt] (4,1) rectangle node[align=center] {$a_4$} (4.5,1.5);
    \fill[tilebot,text=white,rounded corners=3pt] (4,0) rectangle node[align=center] {$b_{4}$} (4.5,0.5);
    \fill[tiletop,text=white,rounded corners=3pt] (5,1) rectangle node[align=center] {$\emptyset$} (5.5,1.5);
    \end{scope}
    \end{tikzpicture}
    \caption{Consider two parallel rows of tiles $a_1,...,a_4$ and $b_1,...,b_4$ in different layers, one on top of another.
    Tiles are depicted as squares.
    The tiles with the label $\emptyset$ represent a tile width of buffer qubits on the periphery of the top layer.
    This example demonstrates how to swap two tiles $a_2$ and $a_3$.
    To begin, we swap alternate tiles in each row as shown in (a).
    We then use the walk operation to move tiles one unit as shown in (b).
    For every pair of tiles we wish to exchange, we perform a inter-layer $\swapp$ as shown in (c).
    In this example, we only wish to swap tiles $a_2$ and $a_3$ so the other tiles remain stationary.
    We then undo the transformation by reversing the walk in (d) and undoing the alternate exchange in (e).
    The final panel (f) is the desired state.
    }
    \label{fig:sideview}
\end{figure}

\textbf{Walking primitive:} By placing a $1/2$-tile wide strip of buffer qubits on the periphery of the lattice, we can ``walk'' the entire memory by swapping the physical-level ancilla qubits and surface code data qubits (Figure \ref{fig:lattice_shift}).\footnote{Not to be confused with the extra buffer qubit per tile.}
Using this walking primitive, we can move an entire layer an entire tile width in depth $2 d_{\ell}$ using only $\swapp$-gates. 
This is a global operation.
We will use this primitive in two ways: First, we implement a transversal $\swapp$ between two tiles that are in different layers in the staggered-$\swapp$ primitive (explained below).
Second, we will use this repeatedly to move an entire layer half-a-tile width in some direction.

\textbf{Staggered $\swapp$ primitive:} When possible, we would like to avoid applying gates directly between data qubits of surface code blocks as this would introduce (small) extra correlations in the logical failure probability of tiles. 
Instead of a direct transversal swap between data blocks, the vertical $\swapp$ can instead be performed between data qubits in one layer and syndrome qubits in the other layer.
This is accomplished via the staggered $\swapp$ primitive.
See Figure~\ref{fig:stag-swap}.

\begin{figure}[h]
    \centering
    \includegraphics{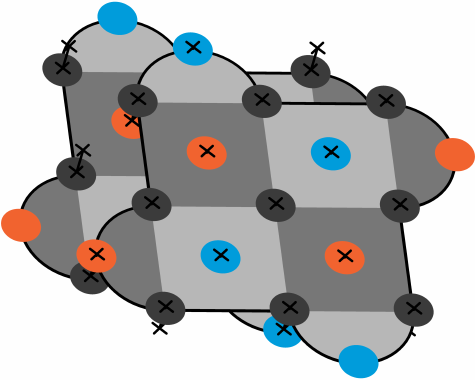}
    \caption{Performing a staggered $\swapp$ operation between layers.
    Physical qubits are arranged such that data (ancilla) qubits in the top layer are adjacent to ancilla (data) qubits in the bottom layer.
    Each qubit in the top layer is swapped with the qubit immediately below it.
    This allows us to exchange tiles between layers without two data qubits directly interacting with each other.}
    \label{fig:stag-swap}
\end{figure}

By default, the qubits in the two layers are positioned such that a data qubit in the top layer is above a data qubit in the bottom layer.
This facilitates performing logical $\cnot$ gates via transversal physical $\cnot$ gates.
We can perform a stagger operation using the walking primitive.
This positions data qubits in the top layer above ancilla qubits in the bottom layer.
We can then apply a transversal $\swapp$ between layers, and undo the stagger operation if need be.
Over the course of the Level-1 logical $\swapp$, however, we only need to undo the stagger operation at the very end.
Throughout the logical $\swapp$, the two layers remain staggered.
If syndrome qubits are reset before use in a syndrome-extraction round, the surface code blocks have undergone a somewhat complicated idle operation with no correlated errors generated between the two surface code blocks.

\begin{figure}[h]
  \centering
  \includegraphics[width=7cm]{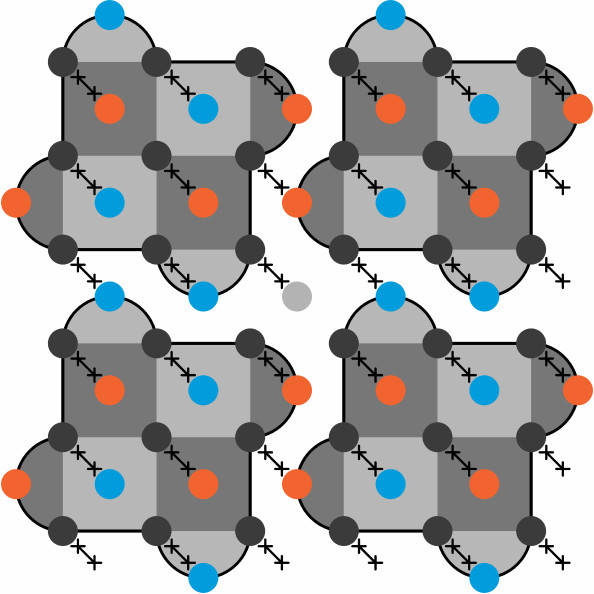}
  \caption{
    Walking primitive used in the $\swapp$ implementation with qubits outside of the $2\times 2$ unit cell not drawn.
    Swaps of Level-0 qubits are drawn as black lines with crosses. 
    Data qubits become ancilla qubits and ancilla qubits become data qubits, so that syndrome extraction is possible at every step.
    A full $1$-unit step to the right of the data qubits can be accomplished following this $1/2$ unit step by swapping each (initially) syndrome qubit with the data qubit up and to the right.
    Later, we will increase the speed at which the top and the bottom layers are shifted relative to each other by shifting the top layer in one direction and the bottom layer in the other.
    }
  \label{fig:lattice_shift}
\end{figure}

\textbf{Level-1 logical $\swapp$:} We are ready to describe the logical $\swapp$ operation.

We begin by exchanging every other tile.
This procedure is illustrated in Figure~\ref{fig:transverse_swap}; this can be compared to Figure~\ref{fig:sideview}.
Steps 1 and 3 (half step shifts) are performed globally with step 2A (vertical swap) or 2B (half step shifts) performed whether or not a logical swap or logical identity is scheduled for a given tile.
The step 2B is necessary for the logical identity gate, because step 2A would otherwise lead to data qubits of adjacent tiles directly next to each other instead of separated by an ancilla qubit.
In this way, syndrome extraction can optionally be performed after every layer of $\swapp$ gates.

\begin{figure}[]
  \centering
  \includegraphics[height={\textheight-6cm}]{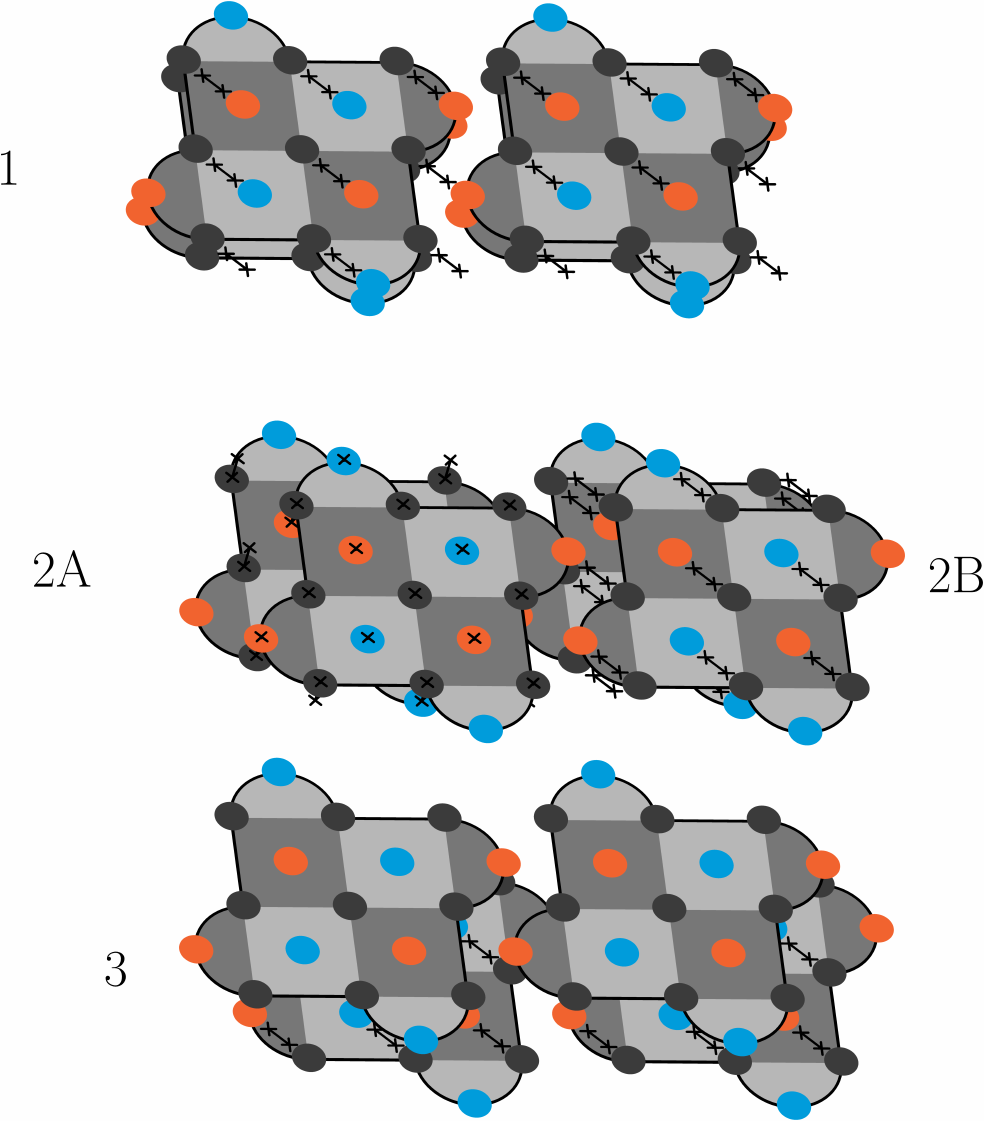}
  \caption{
    3-step transversal swap implementation between two stacked syndrome tiles that avoids directly swapping data qubits. The pair of tiles on the right undergoes an identity operation while the pair of tiles on the left are swapped.
    In step 1, the top layer is shifted by a half-unit using $\swapp$ gates on the top layer.
    In step 2, the pairs of tiles are either swapped using vertical $\swapp$ gates (2A) or shifted using horizontal $\swapp$ gates to keep alignment (2B).
    Finally, in step 3, the lower layer is shifted back by a half-unit using $\swapp$ gates in the bottom layer.
    This operation has the property that syndrome extraction can be performed in all three timesteps. 
    The perspective is inclined slightly to show both layers.
  }
  \label{fig:transverse_swap}
\end{figure}

To summarize, our $\swapp$-gate is implemented as follows where we account for the depth of each operation:
\begin{enumerate}
    \item Use the staggered $\swapp$ on every other tile in a checkerboard pattern to put them in different layers (depth-3).
    \item Use the walking primitive to translate the top and bottom layers $d_{\ell}$ lattice sites in opposite directions so that originally adjacent tiles are now stacked (depth-$d_{\ell}$).
    \item Optionally perform a staggered $\swapp$ (depth-3).
    \item Translate the top and bottom layers $d_{\ell}$ lattice sites back (depth-$d_{\ell}$)
    \item Bring the tiles back to the same layer by undoing a staggered $\swapp$ (depth-3).
\end{enumerate}

At every step, we have the necessary ancilla qubits to perform surface code syndrome extraction.
At first, we might think to perform $d_{\ell}$ rounds of error correction after each of the 5 steps above.
However, because we are working with transversal $\swapp$ operations, we only perform a \emph{single} round of error correction after each step.
Ideal $\swapp$ gates do not spread errors, and therefore, the $\swapp$ operation can be seen as a syndrome-extraction circuit on each tile with a higher failure rate.
While performing just a single round of error correction may reduce the threshold, we expect this change to be minimal.
Furthermore, it reduces the depth of the logical $\swapp$ operation, so our tile $\swapp$-gate takes $2 d_{\ell}+9$ steps of physical $\swapp$-gates.

Recall from Section~\ref{subsec:sec-from-routing} that the syndrome-extraction circuit $C_n^{\cQ}$ is split into $s = 2\Delta + 4$ stages.
Besides the preparation and measurement stages, we simulate a long-range $\cnot$ between pairs of data and ancilla qubits in each stage.
The stage begins by picking an element of each pair and ensuring that they are in different layers.
We can then use the logical $\swapp$ described here to permute tiles.

Note that the $\swapp$ operations can be performed in parallel between tiles on the same layer; the $\swapp$ operations exchange tiles in the same row or column as required by the routing algorithm presented in Algorithm \ref{alg:product-routing}.
We can use Lemma \ref{lem:routing} to show that any permutation of tiles on an $L \times L$ lattice can be accomplished in $3L-3$ steps.
This allows any desired permutation on the $L\times L \times 2$ lattice of tiles to be accomplished in depth $(2d_{\ell}+9)(3 L - 3)$.

In fact, we can optimize this further to avoid repeating redundant operations.
For all but the first and last swap operations:
1) Steps 1 and 5 can be omitted.
2) Within step 3, we may omit the walking step that offsets the upper and lower layers by a half lattice site, so the staggered $\swapp$ becomes a simple transversal swap.
Using these optimizations, any permutation on the $L\times L \times 2$ lattice of tiles can be accomplished in depth $t_\mathrm{route}$ where
\begin{equation}\label{eq:t_route_opt}
    t_\mathrm{route} := (2d_{\ell}+1)(3 L - 3) + 8~.
\end{equation}

\subsubsection[Range R SWAP gates]{Logical permutation routings}
We restrict our attention to range $R$ $\swapp$ gates in a single layer.
Interlayer operations are strictly nearest-neighbor gates, and can be accomplished using the primitives discussed in Section~\ref{subsubsec:nn-swap}.
In the following lemma, we show that an arbitrary permutation routing of tiles can be accomplished in depth $O(L \ell/R)$.

\begin{lemma}
\label{lem:range-R-swap-depth}
Consider access to physical $\swapp$ operations with range $R = o(L \cdot \ell)$.
We can implement any arbitrary permutation of Level-1 qubits in the bilayer architecture in depth
\begin{align*}
    O\left(L \ell/R \right).
\end{align*}
\end{lemma}
\begin{proof}
    We proceed in two cases, $R \ge \ell$ and $R<\ell$.

    \textbf{Case 1: $R \ge \ell$}
    
    Level-0 $\swapp$ gates of range $R$ can be used to implement Level-1 transversal gates of range $R_1 = \floor{R/\ell}$.
    We can route on the Level-1 lattice $\nn_2(L, R_1)$ using Corollary~\ref{cor:sparse-depth-perm} with Level-1 tiles swapped transversally.
    This guarantees that the depth of any permutation routing of tiles is $O(L/R_1)$.
    Each transversal $\swapp$ is followed by a single round of syndrome extraction of the rotated surface code; this requires constant depth and does not affect the depth of permutation routing.

    \textbf{Case 2: $R < \ell$}
    
    The range $R$ Level-0 $\swapp$ gate can be used to speed up the walking primitive presented in Section~\ref{subsubsec:nn-swap}.
    Parallelized Level-1 nearest-neighbor $\swapp$ gates implemented in this way take time $\Theta(\ell/R)$.
    Combined with the Corollary~\ref{cor:depth-nn2}, we have that routing takes time $O(L\ell/R)$.
\end{proof}

\subsection{Biased-noise qubits}
\label{subsec:bilayer-biased}
In Section~\ref{subsec:biased-noise}, we will be interested in suppressing certain kinds of Pauli errors.
When a qubit experiences $\ssX$ or $\ssZ$ errors with an asymmetric rate, it is said to be noise-biased.
In this section, we will explain how to introduce such a noise bias on Level-1 qubits by modifying the bilayer architecture.

Let $\eta \ge 1$ be the desired noise bias of the Level-1 qubits, and suppose $\ssZ$ errors occur with a probability $p$ and are $\eta$-times more likely to occur than $\ssX$ or $\ssY$ errors.
We can introduce a noise bias $\eta > 1$ on Level-1 qubits by elongating the surface code into rectangular regions where the minimum-weight $\ssX$ logical operator is longer than the minimum weight $\ssZ$ logical operator.
See Figure~\ref{fig:biased_Level-1}.
Suppose we have a rectangular surface code patch of dimensions $d_\ssX$ by $d_\ssZ$ such that the minimum weight $\ssX$ logical operator has weight $d_\ssX$ and the minimum weight $\ssZ$ logical operator has weight $d_\ssZ$.
Considering the minimum weight logical operators, we expect a failure rate $\propto p^{\ceil{d_\ssX /2}}$ in the $\ssX$ basis and $\propto p^{\ceil{d_\ssZ/2}}$ in the $\ssZ$ basis.
By taking $d_\ssX > d_\ssZ$, we can introduce a noise bias.
\begin{figure}[h]
    \centering
    \includegraphics{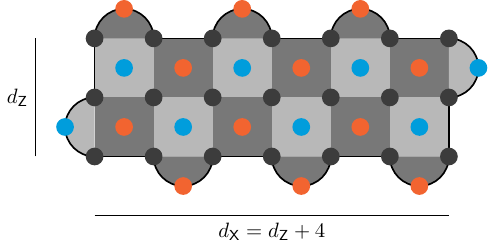}
    \caption{Creating a biased Level-1 qubit by using a rectangular surface code.
    In this picture $d_{\ssX} = 7$ and $d_{\ssZ} = 3$.
    The bias is therefore $\eta = O(p^{-2})$.
    }
    \label{fig:biased_Level-1}
\end{figure}

For simplicity, we assume $d_{\ssX} = d_{\ssZ} + \lceil \log(\eta)/\log(p) \rceil$ to guarantee a bias of \emph{at least} $\eta$.
We also assume that all Level-1 qubits, data and ancilla both, have been biased.
We update the bilayer architecture to use two $L_\ssX \times L_\ssZ$ grids of rotated surface code tiles with $\ssX$ and $\ssZ$ distances $d_\ssX$ and $d_\ssZ$ respectively.
Consider a constant-rate LDPC code $\cQ_n$ with rate $\rho$.
To accommodate $\Space$ outer qubits, we let $L$ be defined by $2L^2 = \Space$.
We then define $L_{\ssZ}$ and $L_{\ssX}$ to be the smallest integers satisfying $L_{\ssZ} \geq L \cdot \sqrt{d_{\ssX}/d_{\ssZ}}$ and $L_\ssX \geq L\cdot \sqrt{d_{\ssZ}/d_{\ssX}}$.

It is not sufficient that the qubits themselves are noise biased.
In addition, we also require that all gates preserve this bias.
We consider each Clifford operation in turn:

\begin{enumerate}
\item \textbf{Preparation \& Measurement:} Within the syndrome-extraction circuit for the outer code where all measurement ancilla qubits are prepared in the $\encket{+}$ state (Section~\ref{subsec:basic_defs}), $\ssX$ errors on the ancilla qubits are suppressed by a factor of $\eta$.
\item \textbf{Entangling gates:} In the bilayer architecture, entangling gates are performed transversally.
Transversal gates are naturally bias preserving.
\item \textbf{$\swapp$ gates:} The way we perform $\swapp$ operations does not change as all tiles have the same dimensions. The buffer region on the periphery of the lattice must be slightly increased to accommodate the elongated surface code tiles during the walking operation.
$\swapp$ operations themselves do not spread errors and are also bias preserving. \footnote{Recall that for two single qubit operators \(A\) and \(B\), \(\mathsf{SWAP} (A \otimes B) = (B \otimes A) \mathsf{SWAP}\).}
\end{enumerate}
Together, this completes the requirements for implementing logical Clifford operations $\cK_1$.

\subsection{Syndrome-extraction circuits for hierarchical codes}
\label{subsec:sec-hierarchical}

The syndrome-extraction circuit $\CFT_N$ for $\cH_N$ is the circuit $C_n^{\cQ}$ where we replace each outer qubit by a tile $\cRS_{\ell}$.
In the circuit $\CFT_N$, each gate of $C_n^{\cQ}$ from the set $\cK$ is replaced by the corresponding element in $\cK_1$ followed by surface code error correction on each outer qubit.
Recall that in Section~\ref{subsec:inner-code-construction}, we discussed how to perform preparation, measurement and logical entangling gates.
In Section~\ref{subsec:swap-bilayer}, we discussed how to perform $\swapp$ gates.

\begin{theorem}
\label{thm:synd-cct-params}
Each element $\cH_{N}$ has an associated $2$-dimensional syndrome-extraction circuit $\CFT_{N}$ with the following properties:
\begin{align*}
  \Space(\CFT_N) = \Theta(N)~, \qquad \Depth(\CFT_N) = O\left( \frac{\sqrt{N}}{R} \right)~.
\end{align*}
Further, each lattice position in $\CFT_N$ only interacts with a fixed set of other lattice positions whose size is independent of $N$.
\end{theorem}
\begin{proof}
    Consider the family of $\dsl n,k,d, \Delta_q, \Delta_g \dsr$ quantum LDPC codes $\{\cQ_n\}$ of constant rate $\rho > 0$.
    
    From Section~\ref{subsec:sec-from-routing}, $\Space(C_n^{\cQ}) = \Theta(n)$.
    To construct $\CFT_N$, each qubit in $C_n^{\cQ}$ is replaced by a surface code $\cRS_{\ell}$.
    It follows that
    \begin{align}
        \Space(\CFT_N) = \Theta(\ell)^2 \cdot \Theta(n) = \Theta(N)~.
    \end{align}
  
  Secondly, each $\cK_1$ operation in $\CFT_N$ requires depth $\Theta(\ell)$ for error correction.
  This is because entangling gates are implemented transversally followed by $d_{\ell}$ rounds of error correction and, per Lemma~\ref{lem:range-R-swap-depth}, the Level-1 logical $\swapp$ operation requires depth $O(L \cdot \ell/R) = O(\sqrt{N}/R)$.
  
  This implies that
  \begin{align}
      \Depth(\CFT_N) = O\left( \frac{\sqrt{N}}{R} \right)~.
  \end{align}
      
  By assumption, logical two-tile operations in $\cK_1$ can be implemented such that the set of lattice positions that interact with each other remain fixed.
  Furthermore, the construction of Section~\ref{subsec:sec-from-routing} guarantees that each of the outer positions in $\CFT_N$ only interact with a fixed and constant-sized set of other outer positions.
  Therefore, all of the physical positions of $\CFT_N$ need only interact with a fixed and constant-sized set of positions in $\CFT_N$.

    This completes the proof.
\end{proof}

%% file: 5-asymptotic.tex
\section{Overhead, threshold and asymptotics}
\label{sec:ft-asymptotics}

In this section, we prove that the hierarchical code has a threshold if we use the syndrome-extraction circuits $\CFT_N$ presented in Section~\ref{sec:surface-code}.
We present a formal version of Theorem~\ref{thm:informal-threshold}.

We recall the bound by Delfosse \emph{et al}.\ in Equation~\eqref{eq:dbt}
\begin{equation}
    \tag{\ref{eq:dbt}}
    \Depth(C_n^{\cQ}) = \Omega\left(\frac{n}{\sqrt{\Space(C_n^{\cQ})}}\right)~.
\end{equation}
According to this bound, the physical circuit $C_n^{\cQ}$ cannot have constant depth and space footprints simultaneously.
This blowup in the volume of the circuit introduces additional failure modes.
Consequently, saturating these bounds by no means ensures the existence of a threshold.
This then seems to have defeated the purpose of simulating a non-local circuit using local operations.

In this section, we present a geometrically-local construction of a circuit $\CFT_N$ that encodes a growing number of encoded qubits \emph{and} guarantees that a threshold exists.
We use code concatenation to define a family $\{\cH_N\}$ that we call \emph{hierarchical} codes.
The $N$\textsuperscript{th} element of this family is obtained by concatenating an LDPC code $\cQ_n$ and a rotated surface code $\cRS_{\ell}$ of size $\Theta(\ell^2)$.
Here $N = \Theta(n \cdot \ell^2)$.

In Section \ref{subsec:evoln-syndrome-ext}, we study the failure rate per round $\pround$ and establish its dependence on $\pphys$ for a syndrome-extraction circuit for any $\dsl n,k,d, \Delta_q, \Delta_g\dsr$ quantum LDPC code.
The circuits are themselves faulty and are described by a locally decaying faults model.
We show in Section \ref{subsec:coarse-grain-concat} that logical gates in the bilayer architecture guarantee that Level-1 logical errors in surface code blocks are suppressed exponentially following logical Clifford operations.
This allows us to deal with the Level-1 syndrome-extraction circuit for the outer code directly without having to keep track of Level-0 failure probabilities.
Permuting tiles can introduce Level-1 correlated errors among the tiles. 
In Section~\ref{subsec:cft-threshold}, we show that there exists a choice of $\ell$ such that the $\pround^{(1)}$ is an arbitrarily small constant.
We can then invoke Gottesman's threshold theorem which we discussed in Section~\ref{subsec:ldpc-review} to prove the existence of a threshold.

We conclude this overview by highlighting some features of this construction.

\begin{enumerate}
\item We show that $\ell = \Theta(\log(n))$ is sufficient to achieve a threshold for the outer code.
If the code $\cQ_n$ has constant rate, then the code $\cH_N$ has rate $O(\log(n)^{-2}) \rightarrow 0$ as $n\rightarrow \infty$.
\item Although outer and inner codes are both LDPC codes, $\cH_N$ is itself no longer an LDPC code---it uses stabilizer measurements that have weight $\log(n)$ as the side length of the inner surface codes is $\ell = \Theta(\log(n))$.
However, the logical operators of the surface code can be measured (destructively) using only single-qubit measurements.
\item There is a cost to locality---the sub-threshold scaling of the logical failure rate is qualitatively different from the typical exponential error suppression as a function of the distance.
Instead, we see a strictly sub-exponential, but still superpolynomial, suppression of the logical failure rate with the distance.
\end{enumerate}

\subsection[Evolution of errors]{Evolution of errors in syndrome-extraction circuit $C_n^{\cQ}$}
\label{subsec:evoln-syndrome-ext}

Consider an $\dsl n,k,d, \Delta_q, \Delta_g\dsr$ LDPC code family $\{\cQ_n\}$ with constant rate, i.e.\ $k = \rho\cdot n$ for some constant $\rho \in (0,1)$ and distance $d = \Theta(n^{\delta})$ for some constant $\delta \in (0,1]$.
In Section \ref{subsec:ldpc-review}, we discussed how Gottesman's proof for the existence of a threshold for LDPC codes depends on the number of errors \emph{per round} of syndrome extraction on both data and ancilla qubits.
In this section, we shall study how the probability of errors per round depends on the circuit $C_n^{\cQ}$.

Recall the circuit $C_n^{\cQ}$ described in Section~\ref{subsec:sec-from-routing}.
Qubits are arranged on two parallel lattices where each lattice has dimensions $L \times L$.
Here, $L$ is the smallest natural number that satisfies $ 2L^2 \geq \Space$.
In total the circuit has $s = 2\Delta + 4$ stages which can be broken down as follows.
The syndrome-measurement circuit $C_n^{\cQ}$ is divided into two phases, one for each of $\ssX$- and $\ssZ$-type syndrome measurements.
Each phase is further divided into $\Delta + 2$ stages, where $\Delta = \max(\Delta_q,\Delta_g)$.
In addition to one stage each to prepare and measure ancillas, there are $\Delta$ stages where we simulate long-range entangling gates.
In each such stage, we permute qubits in the lattice using $\swapp$ gates.
Given access to $\swapp$ operations of range $R$, the permutation has bounded depth $\dperm = O(L/R)$.
This is then followed by nearest-neighbor entangling gates.

Our first result is technical and allows one to compose locally decaying distributions.
\begin{lemma}
\label{lem:composition}
    Let $\Pr_1$ and $\Pr_2$ be two independent and locally decaying distributions on $[n]$ with rates $p_1$ and $p_2$.
    Consider the distribution $\widehat{\Pr}:\Pow(n) \to \bbR$ defined as
    \begin{align}
        \widehat{\Pr}(E) = \sum_{\substack{E_1, E_2 \subseteq{[n]}\\ E \subseteq E_1 \union E_2 }} \widehat{\Pr}_2(E_1) \cdot \widehat{\Pr}_1(E_2)~.
    \end{align}
    Then $\Pr$ is a locally decaying distribution with rate $p = p_1 + p_2$, i.e.\ for all $E \subseteq [n]$,
    \begin{align}
        \Pr(E) \leq (p_1 + p_2)^{|E|}~.
    \end{align}
\end{lemma}
\begin{proof}
    For $E \subseteq [n]$, we can write
    \begin{align}
        \Pr(E) &= \sum_{\substack{E_1, E_2 \subseteq [n] \\ E \subseteq E_1 \cup E_2}} \widehat{\Pr}_1(E_1) \widehat{\Pr}_2(E_2) \\
    &\le \sum_{\substack{E_1, E_2 \subseteq E \\ E = E_1 \sqcup E_2}} {\Pr}_1(E_1) {\Pr}_2(E_2) \\
    &\le \sum_{w=0}^{|E|} \binom{|E|}{w} p_1^w p_2^{|E|-w} \\
    &  = (p_1+p_2)^{|E|}
    \end{align}
    The result follows.
\end{proof}

\begin{lemma}
\label{lem:lse_map}
    Let $\be \in \bbF_2^n$ be a random binary vector such that $E = \supp(\be)$ is distributed according to a locally decaying distribution with rate $p$.
    Let $\bM \in \bbF_2^{m\times n}$ with row and column weight at most $\Delta$ and $\bf = \bM\cdot \be$ be a random variable induced from $\be$.
    Then $F = \supp(\bf)$ is distributed according to a $2^\Delta p^{1/\Delta}$ locally decaying distribution.
\end{lemma}

\begin{proof}
    For a set of bits $A \subseteq [n]$, denote its image under $\bM$ by $\bM(A)$ defined as the union of columns $\bM_i$ of $\bM$ in $A$. I.e. $\bM(A) := \bigcup_{i\in A} \supp \bM_{i}$.
    The probability that an error set $F \subseteq [m]$ on the output occurs
    \begin{align}
        \Pr(F) &= \sum_{\substack{E \subseteq [n] \\ F \subseteq \bM(E)}} \widehat{\Pr}(E)~,
    \end{align}
    i.e. in order for $F$ to have occurred, there must be a set of errors on the input such that $F$ is in the image.
    We can rewrite this sum in terms of the \emph{largest subset} of the powerset $\scrI \subseteq \Pow(n)$ such that for any single set $E\in \scrI$:
    \begin{enumerate}
        \item \label{it:I-assumption-1} We have that $F \subseteq \bM(E)$.
        \item \label{it:I-assumption-2} For all non-empty subsets $G\subset E$, $F\nsubseteq \bM(E\setminus G)$.
    \end{enumerate}
    Each element of $\scrI$ is minimal in the sense that it is a subset of no other element of $\scrI$ (Assumption~\ref{it:I-assumption-2}) while still having $F$ in its image (Assumption~\ref{it:I-assumption-1}).
    The second condition will allow us to replace the $\widehat{\Pr}(\cdot)$ with $\Pr(\cdot)$ in the sum without loosening the upper bound.
    Additionally, the column weight of $\bM$ is at most $\Delta$, so the size of each element $E \in \scrI$ is at least $|F| / \Delta$.
    Using the locally decaying distribution assumption yields
    \begin{align}
        \Pr(F) &\le \sum_{E \in \scrI} \Pr(E)\\
               &\le |\scrI| \cdot p^{|F|/\Delta}~.
    \end{align}
    
    It now remains to count the number of elements of $\scrI$:
    Let $J \subseteq [m]$ be the preimage of $F$ in the sense that for every element $e$ in $J$, the intersection of $\bM(\{e\})$ with $F$ is not empty.
    The row weight of $\bM$ is at most $\Delta$, so $J$ is no larger than $|F|\Delta$.
    Every set $E$ in $\scrI$ must satisfy $E\subseteq J$ or else there would be some element $a \in E$ such that $F \subseteq \bM(E\setminus\{a\})$ (contradicting Assumption~\ref{it:I-assumption-2} on $\scrI$).
    Finally, there are $2^{|J|}$ subsets of $J$, so $|\scrI| \le 2^{|J|} \le 2^{|F|\Delta}$.
    Continuing with the bound, we have that
    \begin{align}
        \Pr(F) &\le |\scrI| \cdot p^{|F|/\Delta}\\
               &\le 2^{|F|\Delta} \cdot p^{|F|/\Delta}\\
               &= \left( 2^{\Delta} \cdot p^{1/\Delta}\right)^{|F|} ~.
    \end{align}
    The result follows.
\end{proof}

The syndrome-extraction circuit $C_n^{\cQ}$ has a special structure---errors do not spread from one data qubit to another or from one ancilla qubit to another.
We show that this implies that $\ssD$ and $\ssA$ are distributed according to a locally decaying distribution.
Before doing so, we review the \emph{symplectic representation} formalism which we use in the following proofs.

\textbf{The symplectic representation:} For $\Space \in \bbN$, consider any Pauli operator $\ssP \in \cP_{\Space}$ and suppose it is expressed as $\ssP = \ssX(\bp_x)\ssZ(\bp_z)$ for $\bp_x, \bp_z \in \bbF_2^{\Space \times 1}$.
Clifford unitary operators $U$ map Pauli operators to Pauli operators under conjugation, i.e.\ $U\ssP U\conj$ is a Pauli operator.
Equivalently, this can be represented as a linear map on $\bp_x$ and $\bp_z$.
Corresponding to $U$, there exists a matrix $\bM \in \bbF_2^{2\Space \times 2\Space}$\footnote{The matrix $\bM$ has additional structure---it is symplectic \cite{gottesman1997stabilizer}, but this is not relevant for this proof.} such that the action of $U$ on $\ssP$ can equivalently can be expressed as
\begin{align}
     \begin{pmatrix}
            \bp_x\\
            \bp_z
        \end{pmatrix} \to
        \bM \cdot \begin{pmatrix}
            \bp_x \\
            \bp_z
        \end{pmatrix} \pmod{2}~.
\end{align}
All arithmetic on symplectic vectors is performed modulo $2$; we drop the `mod $2$' suffix in the equations that follow.

Recall that the $\Space = \Space(C_n^{\cQ})$ qubits in $C_n^{\cQ}$ are partitioned into data qubits and ancilla qubits respectively.
Controlled-$\ssP$ gates for $\ssP \in \cP$ only use the ancilla qubits as control and data qubits as target.
Let $\bd_x$, $\bd_z \in \bbF_2^{n \times 1}$ and $\ba_x$, $\ba_z \in \bbF_2^{m_0 \times 1}$ represent the Pauli operators $\ssD$ and $\ssA$ on data and ancilla qubits respectively.
For the purposes of understanding how errors accumulate over one round of syndrome measurements, we are not interested in the physical locations of the qubits.
As far as their action on $\ssD$ and $\ssA$ are concerned, we treat $\swapp$ gates as (noisy) idle gates \footnote{Colloquially, $\swapp$ gates change the locations of qubits in physical space, not in `math' space.
For instance, suppose each qubit has a label $1,...,n$ and we choose to represent the vector $\bp_x$ as $(\bp_x(1),...,\bp_x(n))$, where $\bp_x(i)$ represents the Pauli operator on the $i$\textsuperscript{th} qubit.
Then moving qubits around in physical space using $\swapp$ gates does not affect the $i$\textsuperscript{th} component $\bp_x(i)$.
For this reason, we ignore the action of $\swapp$ on $\bp_x$ and $\bp_z$.
}.

In any given time step of $C_n^{\cQ}$ where we apply entangling gates, all qubits interact with the same type of gate ($\cnot$ or $\cz$) or remain idle.
The corresponding symplectic matrices have a very special form.
We can write the joint evolution of $\ssD$ and $\ssA$ under the Clifford transformation acting on the $\ssX$- and $\ssZ$-components separately:
\begin{enumerate}
    \item If qubits are only involved in $\cnot$ operations that use ancilla qubits as control qubits and data qubits as target qubits, then there exists a matrix $\bM \in \bbF_2^{m_{\ssX} \times n}$ such that
    \begin{align}
    \label{eq:cnot-symplectic}
        \begin{pmatrix}
            \bd_x\\
            \ba_x
        \end{pmatrix} \to
        \begin{pmatrix}
            (\bM)^T \cdot \ba_x + \bd_x \\
            \ba_x
        \end{pmatrix}~,
        \qquad
        \begin{pmatrix}
            \bd_z\\
            \ba_z
        \end{pmatrix} \to
        \begin{pmatrix}
            \bd_z\\
            \ba_z + \bM \cdot \bd_z
        \end{pmatrix}~.
    \end{align}
    For every pair of qubits indexed by $i \in [m_{0}]$ and $j \in [n]$ that are the control and target of a $\cnot$, the $(i,j)$ entry of $\bM$ is $1$.
    The other entries are $0$.
    In this setting, we note that $\ba_x$ and $\bd_z$ remain invariant.
    \item If qubits are only involved in $\cz$ operations that use ancilla qubits as control qubits and data qubits as target qubits, then there exists a matrix $\bN \in \bbF_2^{m_{\ssZ} \times n}$ such that
    \begin{align}
         \begin{pmatrix}
            \bd_x\\
            \ba_x
        \end{pmatrix} \to
        \begin{pmatrix}
            \bd_x \\
            \ba_x
        \end{pmatrix}~,
        \qquad
        \begin{pmatrix}
            \bd_z\\
            \ba_z
        \end{pmatrix} \to
        \begin{pmatrix}
            \bd_z + \bN^T \cdot \ba_x\\
            \ba_z + \bN \cdot \bd_x
        \end{pmatrix}~.
    \end{align}
    For every pair of qubits indexed by $i \in [m_{\ssZ}]$ and $j \in [n]$ that are the control and target of a $\cz$, the $(i,j)$ entry of $\bN$ is $1$.
    The other entries are $0$.
    In this setting, we note that $\bd_x$ and $\ba_x$ remain invariant.
\end{enumerate}
In the symplectic representation, we can see that the structure of a syndrome-extraction circuit is special because in each phase where we measure either $\ssX$ or $\ssZ$ syndromes, there is always an invariant subspace (for example, $\bd_x$, $\ba_z$ when measuring $\ssX$-type syndromes).

\textbf{The induced error model:} In the symplectic representation, a faulty Clifford operation can be expressed as an affine map---there exists random variables $\bb_x$, $\bb_z \in \bbF_2^{\Space \times 1}$ such that the noisy operation can be expressed as
\begin{align}
    \begin{pmatrix}\bq_x\\ \bq_z\end{pmatrix} = \bM \cdot
    \begin{pmatrix}\bp_x\\ \bp_z\end{pmatrix} +
    \begin{pmatrix}\bb_x\\ \bb_z\end{pmatrix}~.
\end{align}
The errors $\bb_x$, $\bb_z$ are caused by faults.
The faults are themselves are distributed according to the locally decaying distribution $\cF$ with failure rate $\pphys$.
Let $\cX$, $\cZ$ be the induced distributions over $\bb_x$ and $\bb_z$.
For example, $\cX(\bb_x)$ represents the sum of the probabilities over all events where the error is $(\bb_x',\bb_z')$ such that $\supp(\bb_x) \subseteq \supp(\bb_x')$.
In other words, it represents the total probability that the error has a non-trivial $\ssX$ component on $\supp(\bb_x)$.

When the circuit $C$ is composed of elements from $\cK$, we can say more about the induced distributions $\cX$ and $\cZ$.
\begin{lemma}
    \label{lem:induced-dist}
    Consider a Clifford circuit of depth $1$ composed of elements from $\cK$.
    The induced total probabilities $\cX$, $\cZ$ are locally decaying distributions with failure rate $\sqrt{\pphys}$.
\end{lemma}
\begin{proof}
    We shall prove this statement for the distribution $\cX$; the proof for the distribution $\cZ$ is identical.
    Fix an arbitrary vector $\bb_x\in \{0,1\}^{\Space}$.

    Suppose a fault $F$ results in some error $\bb_x'$ such that $\supp(\bb_x') \supseteq \supp(\bb_x)$.
    This implies that $F$ must obey $\supp(F) \supseteq \supp(\bb_x)$.
    Let $F$ be the smallest set of fault locations such that $\supp(F) \supseteq \supp(\bb_x)$.
    Because the circuit $C$ has depth $1$ and is composed entirely of only $1$- and $2$-qubit gates, this implies that $|F| \leq |\bb_x| \leq 2|F|$.
    
    By definition, the total probability of the fault $F$ is $\cF(F)$ a locally decaying distribution with failure rate $\pphys$.
    \begin{align}
        \cX(\bb_x) &\leq \cF(F) \leq (\pphys)^{|F|}\\
                   &\leq (\sqrt{\pphys})^{|\bb_x|}~.
    \end{align}
    The result follows.
\end{proof}

    We are now ready to study $\pround$ and its dependence on $C_n^{\cQ}$.
    To set the stage, we first consider ideal syndrome extraction in the absence of circuit faults.
    We focus our attention on the extraction of $\ssX$-type syndromes and note that the analysis for the $\ssZ$-type syndromes is identical.
    
    Consider a corrupted code state $\ssE \ket{\psi}$ where $\psi$ is a code state and $\ssE = \ssX(\be_x) \ssZ(\be_z)$ is some Pauli operator.
    If the syndrome-extraction circuit $C_n^{\cQ}$ has no faults, the joint state of the data and ancilla qubits after the circuit is described by\footnote{The ancillas are disentangled from the data block by measurement.}
    \begin{align}
    \label{eq:ideal-sec-state}
        \ssE \ket{\psi} \otimes \ssZ(\bsig_{\ssX}) \ket{+}^{\otimes m_{\ssX}}~,
    \end{align}
    where $\bsig_{\ssX}$ represent the ideal syndromes for $\ssX$-type stabilizer generators.

    In this setting, we can use Equation~\eqref{eq:cnot-symplectic} to update $\ssX$- and $\ssZ$-components of Pauli operators under the action of $\cnot$.
    Initially, the $\ssX$ and $\ssZ$ components of the state $\ssE \ket{\psi} \otimes \ket{+}^{\otimes m_{\ssX}}$ can be expressed as $(\be_x | \bzero)$, $(\be_z | \bzero)$, where $\bzero$ is the all zeros vector of length $m_{\ssX}$.
    The vector $\bzero$ means that we assume that the input to the circuit is the state $\ssE\ket{\psi} \otimes \ket{+}^{\otimes m_{\ssX}}$; preparation faults on the ancilla occur in the first time step.
    For $1 < t < \Delta + 2$, we apply $\cnot$ gates specified by a matrix $\bM^{(t)} \in \bbF_2^{m_{\ssX} \times n}$.
    If we do not apply an entangling gate (i.e.\ when we $\swapp$ qubits), then $\bM^{(t)}$ is a matrix of zeros.
    Otherwise, the $(i,j)$ entry of $\bM^{(t)}$ is 1 if and only if the $i$\textsuperscript{th} syndrome qubit and the $j$\textsuperscript{th} data qubit are involved in a $\cnot$ gate in the $t$\textsuperscript{th} time step.
    In the absence of circuit faults, the $\ssX$ components of the error $(\be_x|\bzero)$ are left unaffected during the phase where we measure $\ssX$-type syndromes.
    On the other hand, the $\ssZ$-components transform as
    \begin{align}
    \label{eq:ideal-transform-z}
      \begin{pmatrix} \be_z \\ \bzero \end{pmatrix}\mapsto 
      \begin{pmatrix} \be_z \\ \sum_{t}\bM^{(t)} \cdot \be_z \end{pmatrix}~.
    \end{align}
    The vector $\sum_{t} \bM^{(t)} \cdot \be_z$ is the $\ssX$-type syndrome $\bsig_{\ssX}$.
    In other words, $\sum_t \bM^{(t)} =: \bH_{\ssX}$ is the symplectic representation of the $\ssX$-type stabilizer generators.
    Note that $\bH_{\ssX}$ is a sparse matrix with at most $\Delta_q$ ones per row and $\Delta_g$ ones per column.

    Next, we move on to the setting where circuit components are faulty.
    The final state of the data and ancilla qubits is different from Equation~\eqref{eq:ideal-sec-state} because of circuit faults.
    We express it as
    \begin{align}
    (\ssD \otimes \ssA) \left(\ssE \ket{\psi} \otimes \ssZ(\bsig_{\ssX}) \ket{+}^{\otimes m_{\ssX}} \right)~,
    \end{align}
    where, $\ssD$ and $\ssA$ represent errors due to faults in the circuit $C_n^{\cQ}$ on the data qubits and ancilla qubits respectively.
    The symplectic representation of the final Pauli operator on the data and ancilla qubits is
    \begin{align}
      \begin{pmatrix}\be_x + \bd_x\\ \ba_x\end{pmatrix}~,\qquad
      \begin{pmatrix}\be_z + \bd_z\\ \bsig_{\ssX} + \ba_z\end{pmatrix}~.
    \end{align}
    The probability of errors per round, $\pround$, is the maximum failure rate for the distributions describing $\bd_x, \bd_z$, $\ba_x$, and $\ba_z$.

    \begin{theorem}
        \label{thm:pround-scaling}
        The induced distributions $\cX$ and $\cZ$ that govern the errors $\ssD \otimes \ssA$ are locally decaying distributions with failure rate $\pround$, where
        \begin{align*}
            \pround \leq 2^{\Delta+1} \cdot \Depth(C_n^{\cQ}) \cdot (\pphys)^{1/(2\Delta + 2)}~,
        \end{align*}
        where $\Delta = \max(\Delta_q,\Delta_g)$ is the number of stages in the circuit $C_n^{\cQ}$.
    \end{theorem}
    \begin{proof}
    Recall that the circuit $C_n^{\cQ}$ proceeds in two phases, with the first phase used to measure $\ssX$-type syndromes and the second phase used to measure $\ssZ$-type syndromes.
    For brevity, we allow $\Depth_{\ssX}$ and $\Depth_{\ssZ}$ to be the depth of the circuit $C_n^{\cQ}$ corresponding to each phase; this means $\Depth(C_n^{\cQ}) = \Depth_{\ssX} + \Depth_{\ssZ}$.
    Here, we focus on the first phase of $C_n^{\cQ}$ which is used to measure $\ssX$-type syndromes and study the evolution of $\ssZ$-type errors; the proof of the remaining three cases is identical and for this reason we omit them.
    
    Let $\bb_x^{(t)}$, $\bb_z^{(t)} \in \{0,1\}^n$ and $\bc_x^{(t)}$, $\bc_z^{(t)} \in \{0,1\}^{m_{\ssX}}$ be the errors on data and ancilla qubits induced by faults caused at time $t$.
    In turn, these errors can spread to other qubits and interact with errors at later times.
    Using Equation~\eqref{eq:cnot-symplectic} repeatedly, we can write the final error $\be_z + \bd_z$, $\bsig_{\ssX} + \ba_z$ in terms of the errors at each step as follows: 
    \begin{align}
    \label{eq:faulty-transform}
        \begin{pmatrix}\be_z + \bd_z \\[0.5em] \bsig_{\ssX} + \ba_z\end{pmatrix} &=
        \begin{pmatrix}\be_{z} + \sum_t  \bb_z^{(t)} \\[0.5em] \sum_t \bM^{(t)} \cdot \be_{z} + \sum_t \bM^{(t)} \cdot \sum_{t' < t} \bb_z^{(t')} + \sum_t \bc_z^{(t)} \end{pmatrix}~.
    \end{align}
    As we are only measuring $\ssX$-type syndromes, all sums are over time steps $t$ in the first phase of the circuit.
    For time steps $t$ where we do not apply a $\cnot$, all entries of $\bM^{(t)}$ are $0$.

     We can simplify Equation~\eqref{eq:faulty-transform} by eliminating $\be_z$ and $\bsig_{\ssX} = \sum_{t} \bM^{(t)} \cdot \be_z$:
     \begin{align}
     \label{eq:e-free-evoln}
        \begin{pmatrix}\bd_z \\[0.5em] \ba_z\end{pmatrix} &=
        \begin{pmatrix}\sum_t  \bb_z^{(t)} \\[0.5em] \sum_t \bM^{(t)} \cdot \sum_{t' < t} \bb_z^{(t')} + \sum_t \bc_z^{(t)}\end{pmatrix}~.
    \end{align}
    While it is a straightforward consequence of the linear evolution under symplectic transformations, being able to write $\bd_z$ and $\ba_z$ without $\be_x$ and $\be_z$ means that the $\ssZ$-components of the errors $\bd_z$ and $\ba_z$ do not depend on the input error $\ssE$.

    Furthermore, the special structure of the syndrome-extraction circuit is reflected here --- $\bd_z$ is simply the sum of the errors $\bb_z^{(t)}$ caused by faulty gates at each step.
    In other words, $\ssZ$ errors on data qubits are not affected by $\ssZ$ errors on ancilla qubits.

    We simplify this further using two observations.
    First, we will find it useful to reorder the sums within this equation as follows:
    \begin{align}
        \sum_{t} \bM^{(t)} \cdot  \sum_{t' < t} \bb_z^{(t')} = \sum_{t} \sum_{t'} \mathds{1} [{t' < t}] \; \bM^{(t)} \cdot  \bb_z^{(t')} &= \sum_{t'} \left(\sum_{t > t'} \bM^{(t)}\right) \cdot \bb_z^{(t')}~.
    \end{align}
    where $\mathds{1}[t' < t]$ is the indicator function, i.e.\ it is $1$ when $t' < t$ and $0$ otherwise.
    The terms on the right-hand sides of these equations have a natural interpretation --- for example, the error $\bb_z^{(t')}$ that occurs on data qubits at time $t'$ can propagate to ancilla qubits at times $t > t'$.
    
    Second, it is difficult to directly deal with sums of random vectors modulo 2 that appear in Equation~\eqref{eq:e-free-evoln}.
    Instead, we re-write Equation~\eqref{eq:e-free-evoln} in terms of the \emph{support} of the vectors.
    To this end, we note that
    \begin{align}
        \supp\left(\sum_{t' > t} \bM^{(t')}\right) \subseteq \supp \left(\sum_{t} \bM^{(t)}\right) = \supp(\bH_{\ssX})~.
    \end{align}
    
    Together, these observations mean we can rewrite Equation~\eqref{eq:e-free-evoln} as
    \begin{align}
        \supp \begin{pmatrix} \bd_z \\[0.5em] \ba_z \end{pmatrix} &\subseteq
        \supp \begin{pmatrix} \sum_t  \bb_z^{(t)} \\[0.5em] \bH_{\ssX} \cdot \sum_{t} \bb_z^{(t)} + \sum_t \bc_z^{(t)} \end{pmatrix}\\[2em]
        &\subseteq \bigcup_t
        \supp \begin{pmatrix}  \bb_z^{(t)} \\[0.5em] \bH_{\ssX} \cdot \bb_z^{(t)} + \bc_z^{(t)}\end{pmatrix}
        \\[2em]
        &\subseteq \bigcup_t \supp \left[
        \begin{pmatrix}
            \bI_{n} & \bzero \\
            \bH_{\ssX} & \bI_{m_{\ssX}}
        \end{pmatrix}
        \begin{pmatrix}
            \bb_z^{(t)} \\
            \bc_z^{(t)}
        \end{pmatrix}\right]\label{eq:support-bound}~.
    \end{align}
    where $\bI_n$ and $\bI_{m_{\ssX}}$ are identity matrices of dimensions $n$ and $m_{\ssX}$ respectively.
    We pause to explain the two simplifications in words.
    Substituting $\sum_{t' > t} \bM^{(t')}$ with $\bH_{\ssX}$ corresponds to a worst-case setting---an error on an ancilla qubit can propagate to \emph{all} data qubits in its support \emph{regardless} of when the error on the ancilla qubit occurs.
    Second, by dealing with the union of the supports of the vectors instead of the vectors themselves, we upper bound the maximum size of the final error.
    Evaluating the probability of this event allows us to upper bound the probability of a final error $\ssD \otimes \ssA$.

    We can bound the probabilities of the terms in Equation~\eqref{eq:support-bound}.
    The errors at time $t$,
    \begin{align*}
        \begin{pmatrix}
            \bb_z^{(t)}\\
            \bc_z^{(t)}
        \end{pmatrix}~,
    \end{align*}
    are independent of errors occurring at $t' \neq t$ --- induced errors at different time steps are independent because faults occurring at different time steps are independent.
    As shown in Lemma~\ref{lem:induced-dist}, the induced distributions over errors at each time step are locally decaying distributions with failure rate $\sqrt{\pphys}$.

    Next, Lemma~\ref{lem:lse_map} describes how the distribution is transformed when errors undergo linear transformations.
    Consider the terms in Equation~\eqref{eq:support-bound}:
    \begin{align}
    \label{eq:union}
            \begin{pmatrix}
                \bI_{n} & \bzero \\
                \bH_{\ssX} & \bI_{m_{\ssX}}
            \end{pmatrix}
            \begin{pmatrix}
                \bb_z^{(t)} \\
                \bc_z^{(t)}
            \end{pmatrix}~,
    \end{align}

    $\bH_{\ssX}$ has row and column weight at most $\Delta$, so the block matrix that appears in Equation~\eqref{eq:union} has column weight at most $\Delta+1$.
    By Lemma \ref{lem:lse_map}, each term in the union is distributed according to a locally decaying distribution with failure rate $2^{\Delta+1}(\pphys)^{1/2(\Delta+1)}$.
    
    Finally, Lemma~\ref{lem:composition} allows us to bound the failure rate of the compositions of independent locally decaying distributions.
    This, in turn, is an upper bound on the rate of the locally decaying distribution $\cZ$ over $\bd_z, \ba_z$.
    The union extends over the depth $\Depth_{\ssX}$ of the circuit required to measure $\ssX$-type syndromes terms.
    Applying Lemma~\ref{lem:composition} repeatedly, we find $\cZ$ is a locally decaying distribution with failure rate
    \begin{align}
        2^{\Delta+1} \cdot \Depth_{\ssX} \cdot (\pphys)^{1/2(\Delta+1)}~.
    \end{align}
    By an identical argument, the $\ssX$ errors are distributed according to a $2^{\Delta+1} \cdot \Depth_{\ssX} \cdot \pphys^{1/2(\Delta+1)}$ locally decaying distribution.
    In turn, this means that the induced distributions $\cX$ and $\cZ$ are locally decaying distributions with failure rate $2^{\Delta+1} \cdot \Depth_{\ssX} \cdot (\pphys)^{1/2(\Delta+1)}$.

    Repeating the same analysis for the $\ssZ$-type syndrome measurements, we find that the $\cX$ and $\cZ$ distributions describing induced errors are locally decaying distributions with failure rate
    \begin{align}
        2^{\Delta+1} \cdot \Depth_{\ssZ} \cdot (\pphys)^{1/2(\Delta+1)}~.
    \end{align}
    We can use Lemma~\ref{lem:composition} again to bound the failure rate per for the entire circuit $C_n^{\cQ}$.
    As $\Depth(C_n^{\cQ}) = \Depth_{\ssX} + \Depth_{\ssZ}$, we arrive at the result that $\cX$ and $\cZ$ are locally decaying distributions with failure rate $\pround$ where
    \begin{align}
        \pround = 2^{\Delta+1} \cdot \Depth(C_n^{\cQ}) \cdot (\pphys)^{1/2(\Delta+1)}~.
    \end{align}
    \end{proof}

When qubits are arranged on an $L \times L$ lattice, the circuit depth $\Depth(C_n^{\cQ})$ is $O(\sqrt{n}/R)$.
If gates are constrained by geometric locality, i.e. $R = \omega(L)$, then the depth of the circuit $C_n^{\cQ}$ grows with the code size $n$.
However, for the existence of a threshold, we require $\pround$ to be some fixed constant.
We therefore only achieve a threshold if the physical failure probability vanishes as the size of the code increases:
\begin{align}
     \pphys = O\left[ \left( \frac{1}{\Depth(C_n^{\cQ})} \right)^{2(\Delta+1)} \right]~.
\end{align}
However, if we were to use a concatenated construction, where the outer code is the constant-rate LDPC code $\cQ_n$ and the inner code is a surface code $\cRS_{\ell}$, then we can choose $\pphys$ to decrease \emph{exponentially} with the size of the inner code.
We study this in the next section.

Finally, we comment that the factor $2^{\Delta+1}$ that appears in Theorem~\ref{thm:pround-scaling} can very likely be reduced.
However, this particular version of the theorem is sufficient for our purposes, namely to prove the existence of a threshold for the hierarchical scheme.
For readers interested in applying the hierarchical scheme to the real world, we shall estimate the logical failure rate of the hierarchical scheme numerically in Section~\ref{sec:numerical estimates}.

\subsection{Coarse graining concatenated circuits}
\label{subsec:coarse-grain-concat}
In the next two sections, we will analyze the concatenated code by applying Gottesman's theorem described in Section~\ref{subsec:ldpc-review} to both the inner code and the outer code.
In this section, we apply it to the inner code; for $\Space = \Space(C_n^{\cQ})$, the inner code $\cRS_{\ell}^{\otimes \Space}$ is itself an LDPC code.
In Section~\ref{subsec:cft-threshold}, we will apply Gottesman's theorem to the outer code.

In Section \ref{subsec:concat-review}, we described how we cannot ignore the details of the Level-0 syndrome-extraction circuit in a concatenated code.
In this section, we show that if logical gates on surface codes are performed as described in Section~\ref{sec:surface-code}, then they are fault tolerant.
We show the existence of a threshold $\qphys^{(0)}$ such that if the failure rate per round is below $\qphys^{(0)}$, then we can directly study Level-1 operations and ignore Level-0 operations.

Consider an input state $\rho_{\mathrm{in}} \in (\cRS_{\ell})^{\otimes \Space}$ in the bilayer architecture.
Let Level-0 faults on the syndrome-extraction circuit be distributed according to a locally decaying distribution with failure rate $\pphys^{(0)}$.

The failure rate per round on the data qubits and the syndrome qubits is the same because data and syndrome qubits both interact with $4$ other qubits.
Let $\qin^{(0)}$, $\qround^{(0)}$ be the thresholds for surface code error correction as defined in Section \ref{subsec:surface-review}.
Suppose we are below threshold.
Then after error correction, tiles that have not failed are described by a locally decaying Level-0 error model with failure rate $\pround^{(0)}$.
Theorem~\ref{thm:pround-scaling} guarantees that the failure rate per round grows with the depth of the syndrome-extraction circuit; it also relies on the degree of the qubits and stabilizer generators.
If we measure $\ssX$ and $\ssZ$ syndromes separately, the depth of the syndrome-extraction circuit is at most $12$.
The degree of the qubits and stabilizers is $4$.
Using Theorem~\ref{thm:pround-scaling}, we can bound the failure rate per round of surface code syndrome extraction:\footnote{The constant $2^{\Delta+1} = 32$ and $32 \times 12 = 384$.}
\begin{align}
\label{eq:pround-surface}
    \pround^{(0)} < 384 \left(\pphys^{(0)}\right)^{1/10}~.
\end{align}
This bound can be much better---for example $\ssX$ and $\ssZ$ syndromes can be measured in parallel which, in turn, can reduce the depth of the circuit; we can also likely reduce the constant $384$ in front of $\pround$.
However, we continue to use the bound in Equation~\eqref{eq:pround-surface} for simplicity.

We can use Theorem~\ref{thm:pround-scaling} to show that the logical operations for the bilayer architecture are fault tolerant.
We argue that both the Level-0 and Level-1 failure rates after the operation are constant.

\begin{theorem}
\label{thm:coarse-grain}
    Let $C$ be the circuit on a state $\rho_{\mathrm{in}} \in \cRS_{\ell}^{\otimes \Space}$ such that each tile is involved in at most one logical gate in $\cK_1$.
    Tiles that have not suffered a logical error are described by a locally decaying error with Level-0 input failure rate $\pround^{(0)}$.
    
    There exists a threshold $\qphys^{(0)}$ such that, if $\pphys^{(0)} \leq \qphys^{(0)}$, then 
    \begin{enumerate}
        \item the circuit $C$ is described by a Level-1 locally decaying faults model with Level-1 failure rate $\pphys^{(1)} := \exp(-c_{\mathrm{EC}} \cdot \ell)$.
        \item the output is described by a Level-0 locally decaying errors model with failure rate less than $\pround^{(0)}$.
    \end{enumerate}
\end{theorem}
\begin{proof}
Let $\rho_{\mathrm{in}} \in \cRS_{\ell}^{\otimes \Space}$ be a noisy code state with Level-0 errors described by a locally decaying distribution with failure rate $\pround^{(0)}$.

For sufficiently low logical failure rate, we can use Gottesman's result presented in Section~\ref{subsec:ldpc-review} to bound the failure rate for error correction and to show that after error correction, the Level-0 errors are locally decaying distributions with failure rate $\pround^{(0)}$.

\textbf{State preparation:}
Suppose we wished to prepare the state $\encket{0}^{\otimes m}$ for the code $\cRS_{\ell}^{\otimes m}$.
Each Level-0 qubit is prepared in $\ket{0}$ and we then perform the syndrome-extraction circuit for the surface code on all $m$ copies.
The Level-0 errors are described by a locally decaying error model with failure rate $\pin^{(0)} = \pphys^{(0)}$.
The faults in the syndrome-extraction circuit $C$ are also described by a locally decaying faults model with failure rate $\pphys^{(0)}$.
Error correction is successful if
\begin{align}
\label{eq:prep-bound}
    \pphys^{(0)} < \qin^{(0)}~, \qquad \pround^{(0)} < \qround^{(0)}~.
\end{align}

\textbf{Entangling gates:} Entangling gates between data and ancilla blocks are performed in a transversal manner.
Errors due to faults in the transversal gate are distributed according to a locally decaying distribution with failure rate $\pphys^{(0)}$.
Lemma~\ref{lem:composition} shows that the input to error correction is a state with Level-0 errors described by a locally decaying distribution with failure rate $\pround^{(0)} + \pphys^{(0)}$.

Error correction is successful if
\begin{align}
\label{eq:entangling-bound}
    \pround^{(0)} + \pphys^{(0)} < \qin^{(0)}~, \qquad \pround^{(0)} < \qround^{(0)}~.
\end{align}

\textbf{$\swapp$ gates:} Assume that the Level-0 failure rate is $\pround^{(0)}$.
The logical $\swapp$ operation is decomposed entirely in terms of physical $\swapp$ operations.
As these are non-entangling operations, the error distribution is a locally decaying distribution with failure rate $\sqrt{\pphys^{(0)}}$.
We can use Lemma~\ref{lem:composition} to find the effective failure rate per round.
This is equal to the sum of the failure rate per round of syndrome extraction and the failure rate of the $\swapp$ gate itself.
Note that because the $\swapp$ gate has larger depth, we perform more than $d_{\ell}$ rounds of syndrome extraction.

Therefore, the failure rate per round on both data and ancilla qubits is $\pround^{(0)} + \sqrt{\pphys^{(0)}}$.
Error correction is successful if
\begin{align}
\label{eq:swap-bound}
    \pround^{(0)} < \qin^{(0)}~, \qquad \pround^{(0)} + \sqrt{\pphys^{(0)}} < \qround^{(0)}~.
\end{align}

\textbf{Logical measurement of Pauli operators:}

We will wish to measure logical operators on tiles that represent Level-1 ancilla qubits.
Consider a state with Level-0 errors distributed according to a locally decaying distribution with failure rate $\pround^{(0)}$.
We first study the logical measurement of a single tile.

To destructively measure the logical $\ssX$ ($\ssZ$) operator on a single tile, we can measure each of the physical qubits in the $\ssX$ ($\ssZ$) basis.
This is permitted by our available operations in $\cK_0$.
Faults on measurements are distributed according to a locally decaying distribution with rate $\pphys^{(0)}$.
The resulting distribution on the output bits is a locally decaying distribution with rate $\pround^{(0)} + \pphys^{(0)}$.
We can use each of the individual Level-0 qubit outputs to infer the values of each of the $\ssX$-type ($\ssZ$-type)  stabilizer generators and correct $\ssZ$ ($\ssX$) errors.
This fails with probability $\exp(-c_{\mathrm EC} \cdot \ell)$.

We can now study all tiles that undergo measurement.
As measurements on each tile are performed separately, this induces a Level-1 measurement error with probability $\exp(-c_{\mathrm EC} \cdot \ell)$.

In the mean time, tiles that represent data qubits remain idle for $1$ time step.
As we assume idle errors are distributed according to a locally decaying distribution with failure rate $\pphys^{(0)}$, the Level-0 error rates on these tiles are $\pround^{(0)} + \pphys^{(0)}$.
Error correction is successful if
\begin{align}
\label{eq:measurements}
    \pround^{(0)} + \pphys^{(0)} < \qin^{(0)}~.
\end{align}

\textbf{Combining requirements for all operations:}
We can use Equation~\eqref{eq:pround-surface} to state $\pround^{(0)} < 384 (\pphys^{(0)})^{1/10}$ and note that both $\pphys^{(0)}$ and $\sqrt{\pphys^{(0)}}$ are less than $(\pphys^{(0)})^{1/10}$.
Therefore, we can define the threshold $\qphys^{(0)}$ using the bounds in Equations~\eqref{eq:prep-bound}, Equation~\eqref{eq:entangling-bound}, Equation~\eqref{eq:swap-bound} and Equation~\eqref{eq:measurements}:
\begin{align}
    \qphys^{(0)} = \min\left[\left(\frac{\qin^{(0)}}{385}\right)^{10},
                             \left(\frac{\qround^{(0)}}{385}\right)^{10}
                    \right]~.
\end{align}
Below threshold, we can invoke Gottesman's result to guarantee error suppression; we obtain a logical failure rate $\pphys^{(1)} = \exp(-c_{\mathrm EC} \cdot \ell)$.

\end{proof}

\subsection[]{The syndrome-extraction circuit $\CFT_N$ has a threshold}
\label{subsec:cft-threshold}

In this section, we prove that the hierarchical code $\cH_N$ has a threshold if we measure syndromes using the circuit $\CFT_N$.
We review the construction first and the corresponding assumptions on failure rates.
Thus far, we have simply stated the relationship between $\ell$ and $n$, i.e.\ that $\ell = \Theta(\log(n))$, without justification.
We show in Lemma~\ref{cor:putting-together} that letting the inner code have size $\ell = \Theta(\log(n))$ is indeed sufficient to achieve arbitrarily small, but constant, Level-1 failure rate per round $\pround^{(1)}$.
We bring these elements together in Theorem~\ref{thm:d-scaling-robust} to show that the hierarchical construction has a threshold.

Recall that the hierarchical code $\cH_N$ is constructed by concatenating an outer $\dsl n,k,d, \Delta_q, \Delta_g \dsr$ constant-rate LDPC code $\{\cQ_n\}$ and inner $\dsl d_{\ell}^2, 1, d_{\ell}\dsr$ code $\cRS_{\ell}$.
The family $\cQ_n$ has parameters $k = \rho \cdot n$ for $\rho > 0$ and distance $d = \Theta(n^{\delta})$ for $\delta > 0$.

Suppose we are given an input state of the concatenated code $\cH_N$ subject to the following error model:
\begin{enumerate}
    \item Errors on the state are distributed according to locally decaying distributions:
    \begin{enumerate}
        \item on Level-0 with failure rate $\pin^{(0)}$, and
        \item on Level-1 with failure rate $\pin^{(1)}$.
    \end{enumerate}
    \item Level-0 faults in the circuit are distributed according to a locally decaying distribution with failure rate $\pphys^{(0)}$.
\end{enumerate}
Let $\pround^{(0)}$ denote the failure rate per round of syndrome extraction for the rotated surface code.
As the depth of the syndrome-extraction circuit for the rotated surface code is constant, for fixed values of $\pphys^{(0)}$, $\pround^{(0)}$ is also a constant (See Equation~\eqref{eq:pround-surface}).
We assume $\pin^{(0)} \leq \pround^{(0)}$ because it will make the following statements easier.

Qubits are laid out on a bilayer architecture as described in Section~\ref{sec:surface-code}.
Physical qubits are aggregated to form $\Space(C_n^{\cQ})$ rotated surface codes $\cRS_{\ell}$; these form $2L^2$ tiles where $L$ is the smallest integer satisfying $2L^2 \geq \Space(C_n^{\cQ})$.

The product code $\cRS_{\ell}^{\otimes \Space}$ is itself an LDPC code.
The tiles will be used to simulate long-range entangling gates required to perform the syndrome-extraction circuit $C_n^{\cQ}$ for the outer code.
Single-tile preparation and measurement, and two-tile entangling gates are described in Section~\ref{subsec:inner-code-construction}; Level-1 $\swapp$ gates and permutations of tiles were described in Section~\ref{subsec:swap-bilayer}.
Recall that $\qphys^{(0)} \in (0,1]$ was defined in Section~\ref{subsec:coarse-grain-concat}.
Per Theorem~\ref{thm:coarse-grain}, if the input state has Level-0 errors described by a locally decaying distribution with failure rate $\pround^{(0)}$ and $\pphys^{(0)} < \qphys^{(0)}$, Level-1 circuit faults are distributed according to a locally decaying distribution with failure rate $\pphys^{(1)}$ and Level-0 residual errors are described by a locally decaying distribution with failure rate $\pround^{(0)}$ on tiles that have not failed.
This result allows us to coarse grain the Level-0 circuit and study Level-1 errors and faults directly.

For the outer code to have a threshold, we require that the $\dsl n,k,d,\Delta_q, \Delta_g \dsr$ LDPC code family $\{\cQ_n\}$ has a syndrome-extraction circuit such that $\pround^{(1)}$ remains a sufficiently small constant as discussed in Section~\ref{subsec:ldpc-review}.
In the following lemma, we show that $\ell = \Theta(\log(n))$ is sufficient to achieve this.

\begin{lemma}
\label{cor:putting-together}
    Suppose Level-1 faults on the syndrome-extraction circuit $\CFT_N$ are distributed according to a locally decaying distribution with failure rate $\pphys^{(1)}$.
    Then, for arbitarily small constant $\epsilon > 0$, $\pround^{(1)} < \epsilon$ can be achieved using $\ell = \Theta(\log(n))$.
\end{lemma}
  \begin{proof}
    From Theorem \ref{thm:pround-scaling}, the failure rate per round scales as 
    \begin{align}
        \pround^{(1)} = 2^{\Delta+1} \cdot \Depth(C_n^{\cQ}) \cdot \left(\pphys^{(1)}\right)^{1/2(\Delta+1)}~,
    \end{align}
    where $\Delta = \max(\Delta_q, \Delta_g)$ is some constant for a fixed family $\cQ_n$.
    We want $\pround^{(1)}$ to be an arbitrarily small constant $\epsilon > 0$.
    Per Theorem~\ref{thm:coarse-grain}, Level-1 faults are distributed according to a locally decaying distribution with failure rate which implies that $\pphys^{(1)} = \exp(-c_{\mathrm{EC}} \cdot \ell)$.
    From Section~\ref{subsec:sec-from-routing}, $\Depth(C_n^{\cQ}) = O(L/R) = O(\sqrt{n}/R)$.
    Therefore, the upper bound on $\pphys^{(1)}$ can be satisfied by choosing $\ell = \Theta(\log(n))$.
\end{proof}

In particular, there exists a threshold $\qround^{(1)}$ for the outer code $\cQ_n$ for the syndrome-extraction circuit $C_n^{\cQ}$.
This can be achieved for some $\ell$ such that $\ell = \Theta(\log(n))$.

\begin{theorem}
  \label{thm:d-scaling-robust} 
    There exists a choice of $\ell$ such that $\ell = \Theta(\log(n))$ and thresholds $\qin^{(0)}$, $\qphys^{(0)}$ and $\qin^{(1)}$ such that if
    \begin{align*}
        \max\left(\pin^{(0)}, \pround^{(0)}\right) < \qin^{(0)}~, \qquad \pphys^{(0)} < \qphys^{(0)}~,\qquad \pin^{(1)} < \qin^{(1)}~,
    \end{align*}
    then the following is true.
    With probability at least $1-p_{\cH}(N)$, the state after error correction is correctable by an ideal decoder where, for some positive number $c_{\cH}$ that is independent of $N$,
      \begin{align*}
          p_{\cH}(N) < \exp\left(-c_{\cH} \cdot \frac{N^{\delta}}{\log^{2\delta}(N)}\right) ~.
      \end{align*}
      Furthermore, the residual errors are distributed according to a locally decaying distribution with failure rates $\pround^{(1)}$ and $\pround^{(0)}$ on Level-1 and Level-0 respectively.
\end{theorem}
\begin{proof}
  Suppose $\pin^{(0)} < \qin^{(0)}$ and $\pphys^{(0)} < \qphys^{(0)}$.
  By definition, this is sufficient to perform Level-1 logical gates and surface code error correction as per Theorem~\ref{thm:coarse-grain}.

  Next, the LDPC code has thresholds $\qin^{(1)}$ and $\qround^{(1)}$ (See Section \ref{subsec:ldpc-review}).
  The input state has Level-1 errors described by a locally decaying distribution with failure rate $\pin^{(1)}$.
  For the syndrome-extraction circuit on the outer LDPC code to be successful, we require $\pin^{(1)} < \qin^{(1)}$.
  
  Finally, we require that the Level-1 failure rate per round is below the corresponding threshold
  \begin{align}
      \pround^{(1)} < \qround^{(1)}~.
  \end{align}
  From Lemma~\ref{cor:putting-together}, this can be achieved using $\ell = \Theta(\log(n))$.

  By definition, syndrome-extraction is successful if the ideal decoder $\cR_{\cH}$ is able to recover the final state.
  The code $\cH_N$ fails if the outer LDPC code $\cQ_n$ fails, i.e.\ the probability of failure is $p_{\cH}(N) = p_{\cQ}(n) = \exp(-c_{\cQ} \cdot d(n)) = \exp(-c_{\cQ} \cdot \Theta(n^{\delta}))$.
  
  Using Equation~\eqref{eq:bounds-n-from-N}, we can express the probability of failure $p_{\cH}(N)$ in terms of $N$:
  \begin{equation}
    p_{\cH}(N) < \exp\left(-c_{\cH} \cdot \frac{N^{\delta}}{\log^{2\delta}(N)}\right)~.
  \end{equation}
  for some positive number $c_{\cH}$ that is independent of $N$.

  Residual errors are distributed according to locally decaying distribution:
  \begin{enumerate}
      \item on Level-1 with failure rate $\pround^{(1)}$; this is guaranteed by Gottesman's result applied to the outer code.
      \item on Level-0 errors failure rate $\pround^{(0)}$;
      This is guaranteed by Theorem~\ref{thm:coarse-grain}.
  \end{enumerate}
  The result follows.
\end{proof}

We reiterate that $p_{\cH}(N)$ is an upper bound on the failure rate for the Level-2 error probability distribution.

This analysis depends crucially on the failure rate of the $\swapp$ gates; $\pround^{(1)}$, and therefore the size of the inner code, scales with the depth of the circuit because of noisy $\swapp$ operations. In proving Theorem \ref{thm:d-scaling-robust}, we were agnostic to the failure modes in the circuit and assumed that all Level-1 two-qubit gates fail with probability $\pphys^{(1)}$.
However, if the fidelity of physical $\swapp$ gates can be improved over the fidelity of entangling gates, this can reduce the overhead for the hierarchical scheme significantly.
We provide evidence for this in Section~\ref{subsec:numerical_estimates_results} where we estimate the logical failure rate for the hierarchical scheme.
In certain architectures such as trapped neutral atoms, $\swapp$ gates can be performed by physically moving the trap \cite{bluvstein2022quantum}.
In this case, the failure rate for the $\swapp$ gates may have no direct connection to the failure rate for $\cnot$ operations.

%% file: 6-quant-estim.tex
\section{Comparisons with the basic encoding}
\label{sec:numerical estimates}

We have shown that the hierarchical code $\{\cH_N\}$ has a syndrome-extraction circuit that can be constructed using gates restricted by geometric locality such that it has a threshold.
Below threshold, the WER is suppressed superpolynomially, but subexponentially in $N$.
It is natural to ask whether the resources spent in performing $\swapp$ gates can be better spent simply building a more robust surface code.
In this section, we consider the basic encoding $\cB_M$ which encodes $K$ logical qubits in surface codes $\cRS_{\ell_M}$.
We compare the hierarchical scheme and the basic enconding in different ways.

We show in Section \ref{subsec:asymptotic-comparison} that for a target WER, the syndrome-extraction circuit for the hierarchical memory is more efficient than the syndrome-extraction circuit for the basic encoding.
This is measured by the depth and width of the corresponding circuits.
We will state and prove a formal version of Theorem~\ref{thm:asymptotic-comparison}.
Depending on the value of the threshold for the outer LDPC codes however, it is not immediately obvious that this scaling manifests for practical block lengths.
In the rest of this section, we present numerical estimates for the WER $p_{\cH}(N)$ of the hierarchical memory and contrast it with the WER $p_{\cB}(M)$ for the basic encoding.
We do this by demanding a fixed total number of qubits for both schemes and compare $p_{\cH}(N)$ and $p_{\cB}(M)$.
We demonstrate that there is a \emph{crossover point}, i.e.\ a value of the physical error rate where, for fixed total number of qubits, the hierarchical memory outperforms the basic encoding, i.e. $p_{\cH}(N) < p_{\cB}(M)$.
In our estimates, this happens at gate error rates roughly between $10^{-3}$ and $10^{-4}$.
While these are preliminary estimates, they are promising nonetheless as they are in the realm of possibility.

In Section~\ref{subsec:setup-num-est}, we briefly discuss the codes we use as outer and inner codes.
To estimate the crossover point, we make some assumptions about the noise model, gates, and decoder.
Owing to these assumptions, our estimates should only be interpreted as a proof-of-principle that the overhead of the hierarchical scheme pays off in a reasonable parameter regime.
In Section~\ref{subsec:numerical_estimates_results}, we present the results of our simulations.
All together we believe these assumptions, especially those related to the decoder, code, and noise model, are conservative.
We return to these assumptions in Section~\ref{sec:future} and for each assumption, we outline how one might expect it to change (1) in the future, and (2) in a more realistic setting.
In general, we expect that with careful engineering (e.g. high-rate linear-distance codes, architecture-specific considerations, improved decoding algorithms) and more realistic noise modeling (e.g. including significant error correlations), the cross-over to when hierarchical memories outperform surface codes will occur at smaller numbers of logical qubits, higher physical error rates, and higher target WERs than in our estimates.

\subsection{Asymptotic comparison with surface code}
\label{subsec:asymptotic-comparison}
We have proved the existence of a threshold when we simulate an LDPC code using local gates.
However, the existence of a threshold alone might not warrant switching over to a different scheme when there already exists an excellent local scheme --- the surface code.
We recall that we are only constructing a quantum memory, and not a scheme for universal, fault-tolerant quantum computation.
In this section, we ask how the surface code would perform if we used the same total number of qubits used in the concatenated scheme above to plainly encode all logical qubits.
We find that there is a space-time tradeoff to implementing a hierarchical scheme.

The hierarchical scheme $\{\cH_N\}$ with corresponding fault tolerant syndrome-extraction circuits $\{\CFT_N\}$ achieves the following costs:
\begin{align}
\label{eq:repeat-params-HN}
  \Space(\CFT_N) = \Theta(N) \qquad \Depth(\CFT_N) = O\left(\frac{\sqrt{N}}{R}\right)~.
\end{align}
This family encodes $k(n) = \rho \cdot n$ qubits.
Note that the depth is for a \emph{single} round of syndrome extraction.
We will return to this point shortly.

Consider the \emph{basic encoding} defined by the family $\{\cB_{M}\}$ where $M = \Theta(k \cdot \ell_M^2)$, and $\cB_{M} = \bigotimes_{i=1}^{k} \cRS_{\ell_M}$ is a $k$-fold product of rotated surface codes $\cRS_{\ell_M}$.
Each code $\cRS_{\ell_M}$ has distance $d_{M} = \Theta(\ell_M)$.
Let the corresponding circuits be denoted $\{\Cplain_{M}\}$.

To compare with $\cH_N$, we probe the parameters of $\Cplain_M$ required to guarantee the same logical error suppression.
Let $p_{\cB}(M)$ denote the failure rate for the Level-1 logical probability of failure for $\cB_M$---we declare failure if any of the $k$ logical qubits fail.

We assume that the Level-0 physical failure rates are sufficiently below threshold to perform surface code error correction, i.e.\ $\pphys^{(0)} < \qphys^{(0)}$.
Furthermore, we also assume that the code state $\cB_M$ is prepared such that there are no input errors, i.e.\ $\pin^{(1)} = \pin^{(0)} = 0$.
This allows us to isolate the rate of error suppression because of error correction.

\begin{lemma}
\label{lem:surface-suppression}
  Let $\{\cB_M\}$ be the basic encoding such that $p_{\cB}(M) < \exp(-c_{\cH} \cdot N^{\delta}/ \log(N)^{2\delta})$ where $c_{\cH}$ is a positive constant.
  Then
  \begin{align*}
    \Space(\Cplain_{M}) = \Omega\left[\left(\frac{N}{\log(N)} \right)^{1+2\delta}\right]~, \qquad \Depth(\Cplain_M) = \Omega\left[\left(\frac{N}{\log^2(N)}\right)^{\delta}\right]~.
  \end{align*}
\end{lemma}
\begin{proof}
By assumption, $\pin^{(1)} = \pin^{(0)} = 0$ and $\pphys^{(0)} < \qphys^{(0)}$ and therefore we can perform error correction.
The Level-1 logical failure probability for the code $\cRS_{\ell_M}^{\otimes k}$ is described by a locally-decaying error model with failure rate $p_{\cRS}(\ell_M)$ (See Section \ref{subsec:surface-review}), where
\begin{align}
  p_{\cRS}(\ell_M) &= \exp(-c_{\mathrm{EC}} \cdot \ell_M)~.
\end{align}
We declare failure if any of the $k$ tiles of $\cB_M$ fails, which implies that
\begin{align}
\label{eq:obs1}
\begin{matrix}
    p_{\cRS}(\ell)    &\leq& &p_{\cB}(M) &\leq& 1-(1-p_{\cRS}(\ell))^{k}\\
\exp(-c_{\mathrm{EC}} \cdot \ell_M) &\leq& &p_{\cB}(M) &\leq& n \cdot \exp(-c_{\mathrm{EC}} \cdot \ell_M)
\end{matrix}~.
\end{align}
To guarantee that the error rate $p_{\cB}(M)$ is lower than $p_{\cH}(N)$, we at least require that
\begin{align}
    \exp(-c_{\mathrm{EC}}\cdot \ell_M) \leq \exp\left(-c_{\cH} \cdot \frac{N^{\delta}}{\log(N)^{2\delta}}\right)~.
\end{align}
This implies that $\ell_M = \Omega(N^{\delta}/\log(N)^{2\delta})$.

We can now compute the space and depth requirements for $\Cplain_M$.
The space cost $\Space(\Cplain_M)$ is $\Theta(k \cdot \ell_M^2)$.
The hierarchical memory $\cH_N$ uses a constant-rate $\dsl n,k,d, \Delta_q, \Delta_g \dsr$ quantum LDPC code $\cQ_n$ where $k(n) = \rho \cdot n$ and $d(n) = \Theta(n^{\delta})$.
This implies that
\begin{align}
    \Space(\Cplain_M) = \Omega\left[\left(\frac{N}{\log(N)}\right)^{1+2\delta}\right]~.
\end{align}

Furthermore, each tile requires $\ell_M$ rounds of error correction; syndrome-extraction circuits on separate tiles can be run in parallel.
Therefore
\begin{align}
    \Depth(\Cplain_M) = \Theta(\ell_M) = \Omega(N^{\delta}/\log^{2\delta}(N))~.
\end{align}
This completes the proof.
\end{proof}

Comparing with Equation~\eqref{eq:repeat-params-HN}, the basic encoding requires a larger space overhead for all $\delta > 0$:
\begin{align}
    \frac{\Space(\Cplain_M)}{\Space(\CFT_N)} = \Omega\left[ \frac{N^{2\delta}}{\log^{1+2\delta}(N)} \right]~.
\end{align}

As stated, however, the time overhead is worse.
Although the depth of the syndrome-extraction circuit $\CFT_N$ is $O(\sqrt{N}/R)$, we will need to perform $d(n)$ rounds of syndrome extraction to be fault tolerant.
However, this is not a fundamental requirement; it is due to the nature of Gottesman's proposal in \cite{gottesman2014fault} which uses an inefficient minimum-weight decoder.
There exist constant-rate LDPC codes that possess efficient, single-shot decoding algorithms, i.e.\ syndrome extraction only needs to be performed a constant number of times for the decoding algorithm to work \cite{leverrier2015quantum,fawzi2018constant,fawzi2018efficient}; furthermore the algorithm requires $O(N)$ time.
For such codes, we can compare the depth of the syndrome-extraction circuit $\Depth(\CFT_N)$ and that of the basic encoding $\Cplain_M$.
In addition to the width blowup, the basic encoding also requires a larger time overhead when $\delta > 1/2$.
\begin{align}
    \frac{\Depth(\Cplain_M)}{\Depth(\CFT_N)} = \Omega\left[ \frac{N^{\delta-1/2}}{R \cdot \log^{2\delta}(N)} \right]~.
\end{align}
Using LDPC codes with single-shot decoding algorithms, the hierarchical memory is a more efficient way to achieve a low logical error rate in terms of both circuit depth and width.

Having said this, it is not clear if this advantage manifests for practical block lengths.

For small codes and high error rates, it may well be that it is still optimal to use the basic encoding.
We expect to see a \emph{crossover point}---a value of physical error rate where the hierarchical scheme $\{\cH_N\}$ has a lower logical failure rate than the basic encoding $\{\cB_M\}$ with the same overhead.
Where exactly this crossover happens will depend on a number of parameters that are specific to the implementation, including the choice of the outer code, its threshold and our choice of decoders.
In the rest of this section, we attempt to estimate where this happens.

\subsection{Setup for numerical estimates}
\label{subsec:setup-num-est}

\textbf{Outer code:}
We choose a quantum expander code as our outer code \cite{tillich2014quantum, leverrier2015quantum}.
We do not utilize any of the structure of the code, so any LDPC code with constant rate and polynomially scaling distance will suffice.
For these reasons, we only briefly discuss the code construction.
We pick a classical code by sampling the check matrix from the ensemble of $m \times n$ matrices with $5$ ones in each column and $8$ ones in each row.

In particular, we pick a classical code with length $896$ and $336$ encoded bits.
Work by Litsyn and Shevelev \cite{litsyn2002ensembles} computes the asymptotic weight distribution of codewords --- with high probability, this code has distance $119$.
The resulting quantum code has parameters
\begin{align}
    N = 1\,116\,416, \quad K = 112\,896,\quad D = 119,\quad \Delta_q = 16,\quad \Delta_g = 13~.
\end{align}
We choose this code as it has a high rate which is necessary to reduce the amount of overhead in the scheme. 
The large block length we consider here is a consequence of using code families with sub-linear distance scaling.
However, the full trade-off for rate, distance, check weight, etc for linear-distance codes has not yet been explored.
We note that even at a relative distance $d/n$ of $10^{-3}$, a linear-distance code of such large block length would achieve a distance of roughly $10^3$.

While the hierarchical memory construction has good asymptotic performance guarantees, if the overhead is too high then the hierarchical memory wins only at an extremely low WER.\footnote{
  1 year is $\approx 10^{16}$ nanoseconds, 1 Hubble time is $\approx 10^{10}$ years or $\approx 10^{26}$ nanoseconds.
  For any practical purpose %
  a WER of $10^{-25}$ per gate time should suffice. 
}

\textbf{Inner code:} As in the earlier sections, we consider the square rotated surface codes $\cRS_{\ell}$ for the inner code.
In our estimates, we allow for $d_\ell = 3,5, 9, 15, 21, 27$.

We make some assumptions about errors on both the logical and physical levels.
We present these assumptions together below and discuss justifications for some assumptions in what follows.

\textbf{Level-0 noise model:}
We assume circuit-level Pauli noise on each physical qubit.
We treat $\swapp$ gates and other Clifford operations separately.
\begin{enumerate}
\item \label{it:circuit-level-noise} 
Each $t$-qubit gate (except $\swapp$ gates) at the physical level fails with a probability $p$ and leaves behind one of the $4^t-1$ non-trivial $t$-qubit Pauli operators picked uniformly at random.
We assume that qubit reset completely removes all traces of the original state.
However, it may reset to the wrong computational basis state with probability $p$.
\item The failure probability of the physical $\swapp$-gate is $r_{\swapp} \cdot p$, where $r_{\swapp} = 1, 10^{-1}, 10^{-2}$.
In this setting, the surface code syndrome-extraction circuit is performed every $1/r_\swapp$ $\swapp$-gates, so that at the physical level the circuit-level noise model remains relatively unchanged for different values of $r_\swapp$.
This assumption will be discussed in detail in Section~\ref{subsec:gates}.
\end{enumerate}

\textbf{Level-1 noise model:}
We assume that the surface code fails at a probability $p_{\cRS}^{(1)}(d_\ell)$, and that the effective noise witnessed by the outer code because of all the $\swapp$ gates is $p^{(1)}$.

\begin{enumerate}
\item The logical error rate $\psurf(d_\ell)$ of each $\ell \times \ell$ rotated surface code tile is \cite{wang2011surface, fowler2012surface}
\begin{equation}\label{eq:fmmc}
  \psurf(d_\ell)\approx 0.1 \left(\frac{p}{10^{-2}}\right)^{\ceil{d_\ell/2}}
\end{equation}
per $d_\ell$ physical level timesteps where one physical level timestep is one round of syndrome extraction plus one (optional) transversal gate which totals roughly 6 gates.
This assumption is discussed in Section~\ref{subsec:decoder_perf_surf}.
\item The effective error rate per long-range $\cnot$ gate is $p^{(1)} = 1-(1-\psurf(d_\ell))^{t_\mathrm{route} + 1}$.
It is analogous to the two-qubit gate error rate in the model with long-range gates.
$t_{\mathrm{route}}$ is the time required for permutation routing presented in Equation~\eqref{eq:t_route_opt}.
\end{enumerate}

\textbf{Level-2 noise model:} Finally, we assume that the logical failure rate for the LDPC code, $p_{\cQ}(n)$, is consistent with a minimum-weight decoder.
\begin{enumerate}
\item \label{it:ldpc-wer} 
For our LDPC code, we assume that the WER under circuit-level Pauli noise using long-range gates is
\begin{equation}
    p_{\cQ} = \left(\frac{p^{(1)}}{10^{-3}}\right)^{10}
\end{equation}
per cycle of syndrome extraction.
The threshold of the code is assumed to be about $10^{-3}$ under circuit noise.
The exponent is $10$ rather than half the distance which is $\sim 55$ because of hook errors.
This assumption is discussed in Section~\ref{subsec:decoder_perf_ldpc}.
\end{enumerate}
If desired, readers can skip ahead to the numerical estimates in Section \ref{subsec:numerical_estimates_results} and return to the justification of the noise model later.

\subsubsection{Decoder performance for the inner code}
\label{subsec:decoder_perf_surf}
Consider Equation~\eqref{eq:fmmc} for the scaling of the logical failure rate for a surface code of distance $d_{\ell}$. 
We assumed a surface code threshold of $10^{-2}$.

This equation neglects:
\begin{enumerate}
\item finite-size effects present at very small code distances.
\item the slight reduction in threshold from inserting a layer of gates failing with rate $p$ between syndrome extraction cycles\footnote{In general, the precise value of the circuit-level threshold already requires some assumptions about what gates are native in the device: The optimal syndrome extraction circuit with our physical layer layout requires 5 to 8 gates depending on these assumptions, so the insertion of an additional gate is relatively unimportant.}.
Recall that this is necessary to implement a logical $\swapp$ operation in the bilayer architecture as discussed in Section \ref{subsec:swap-bilayer}.
\item the distinction between rotated and standard surface codes.
Owing to this, the expression for the logical error rate is an order of magnitude estimate.
We expect our conclusions should be somewhat insensitive to the precise form of the logical error rate and also apply to more general locally decaying error models. %
For calculational convenience, we assume that the failure rate $q$ after $T$ syndrome extraction rounds is given by $1-q = \left(1-\psurf(d_\ell)\right)^{T/d}$.
\end{enumerate}

\subsubsection{Physical $\swapp$ fidelity}
\label{subsec:gates}
Recall that simulating a long-range $\cnot$ via $\swapp$ gates results in a $\cnot$ failure rate of $p^{(1)}$.
We assume that the effective failure rate witnessed by the outer code is
\begin{align}\label{eq:rswap_failure_rate}
    p^{(1)} = 1-(1-\psurf(d_\ell))^{\frac{r_{\swapp}t_\mathrm{route}}{d_\ell} + 1}~.
\end{align}
The parameter $t_{\mathrm{route}}$ is the time required to perform a permutation routing in the bilayer architecture as specified in Equation~\eqref{eq:t_route_opt}.
For convenience, we restate it here
\begin{equation}
    \tag{\ref{eq:t_route_opt}}
    t_{\mathrm{route}} = (2d_{\ell}+1)(3L-3) + 8~.
\end{equation}
The parameter $r_{\swapp}$ bounds the (in)fidelity of the $\swapp$ operation in terms of the $\cnot$ gate (in)fidelity as we now explain.

In the previous sections, we assumed that all gates failed at the same rate.
As noted in Section~\ref{subsec:cft-threshold}, the main source of noise in the hierarchical model stems from the $\swapp$ gates.
This worst-case model was convenient for a proof of the existence of a threshold.
Furthermore, in many devices the $\swapp$ gate is implemented using the same mechanism as the two-qubit entangling gates and so the noise rates are comparable.
However, this is not the only way to implement $\swapp$ gates.

In platforms where the qubits can be physically moved, we can effectively ``rewire'' the connectivity of the device at runtime.
Physically swapping qubits does not require the qubit degree of freedom to be coupled to, and so one might expect that it is an easier task to perform with higher fidelity or speed.
Such techniques have been demonstrated in some experimental platforms: Rearrangable tweezers in Rydberg platforms~\cite{endres2016atom, barredo2016atom, bluvstein2022quantum} and ion shuttling in trapped ion platforms~\cite{hensinger2006t, kaufmann2017fast}.
In this setting, it is possible that the $\swapp$ gate has much higher fidelity than $\cnot$ gates. 

Accordingly, in our model, we assign a constant $r_\swapp$ which specifies the ratio of $\swapp$-gate and idle noise to $\cnot$-gate noise.
With a less noisy $r_\swapp$, we perform $1/r_\swapp$ level-0 $\swapp$ operations per round of surface code syndrome extraction such that the physical error rate in the surface code syndrome extraction circuit remains constant with respect to $r_\swapp$.
Utilizing this optimization, an entire permutation takes $r_{\swapp}t_\mathrm{route}$ rounds of syndrome extraction.
Equation~\eqref{eq:rswap_failure_rate} is in terms of the surface code cycle ($d_\ell$ rounds of syndrome extraction), and the total number of surface code cycles is $\frac{r_{\swapp}t_\mathrm{route}}{d_\ell}$ for a permutation and $1$ for an entangling gate.
We have omitted floor and ceiling functions in this discussion for simplicity.

For example, in a neutral atom system \cite{bluvstein2022quantum}, an array of qubits with a coherence time of seconds was rearranged with an average rearrangement speed of several microseconds per lattice site moved.
If the dominant source of errors in rearrangement is due to idle errors, then we should assign an infidelity to the $\swapp$ gate of roughly $10^{-5}$ whereas the two-qubit gate possessed an infidelity of about $10^{-2}$ i.e. $r_\swapp \approx 10^{-3}-10^{-2}$.
Routing does not require generating entanglement, so the qubit can remain encoded in well-isolated degrees of freedom.
Owing to this, we consider three scenarios: where $r_{\swapp}$ is $10^0$, $10^{-1}$, or $10^{-2}$.

We note that it is a simplification to consider the rearrangement primitive in each platform (tweezers, ion shuttling, etc.) as simply $\swapp$-gates: Frequently there are effects like accumulated motional heating, recooling, acceleration speed limits, etc, but we expect the basic routing ideas and qualitative conclusions remain the same even in the more complicated setting.

\subsubsection{Hook errors}
\label{subsec:hook_errors}

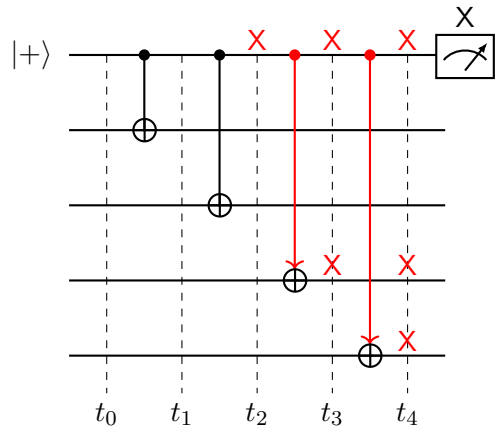
\begin{figure}[h]
  \centering
  \begin{tikzpicture}
      \node at (-0.5,5) {$\ket{+}$};
      \draw[dashed] (0.5,5)--(0.5,0.5) node[below] {$t_0$};
      \foreach \y in {1,...,5} {
         \draw[thick] (0,\y)--(5,\y);
      }
      \foreach \x in {1,2} {
          \draw[dashed] ({\x+0.5},5)--({\x+0.5},0.5) node[below] {$t_{\x}$};
         \fill[black] (\x,5) circle (0.075);
         \draw[thick] (\x,5)--(\x,{5-\x});
         \node at ({\x-0.06},{5-\x}) {\usebox\cnottarget};
      }
      \foreach \x in {3,4} {
        \draw[dashed] ({\x+0.5},5)--({\x+0.5},0.5) node[below] {$t_{\x}$};
         \fill[red] (\x,5) circle (0.075);
         \draw[color=red,thick,->] (\x,5)--(\x,{5.15-\x});
         \node at ({\x-0.06},{5-\x}) {\usebox\cnottarget};
      }
      \node at (5.2,5.1) {\usebox\mmt};
      \node at (2.5,5.2) {\color{red} $\ssX$};
      \node at (3.5,5.2) {\color{red} $\ssX$};
      \node at (3.5,2.2) {\color{red} $\ssX$};
      \node at (4.5,5.2) {\color{red} $\ssX$};
      \node at (4.5,2.2) {\color{red} $\ssX$};
      \node at (4.5,1.2) {\color{red} $\ssX$};
  \end{tikzpicture}
  \caption{Hook errors: errors flowing onto data qubits.
  In this case, an $\ssX$ error appears in the midst of a syndrome-extraction circuit.
  This then propagates to the data qubits.}
  \label{fig:err-flow}
\end{figure}

Current decoder technology for LDPC codes is relatively immature, so we assume a WER scaling consistent with a minimum-weight decoder.
At physical failure rate $p$, we assume that the logical failure rate of an $\dsl n,k,d, \Delta_q, \Delta_g\dsr$ LDPC code $\cQ_n$ below threshold is dominated by a term proportional to $p^t$ where $t$ is the smallest number of fault locations that is uncorrectable.
If our intuition is informed by an i.i.d.\ error model on qubits, we may expect $t \approx d/2$.
However, this is not true in the context of syndrome-extraction circuits as corrupted syndrome qubits can spread errors to many data qubits.

These errors, called hook errors \cite{dennis2002topological}, are harmful errors that can dominate the lower error rate performance of the quantum code.
By a rough estimate, they can reduce the distance of a $\dsl n,k,d, \Delta_q, \Delta_g\dsr$ LDPC code by a factor of $\Delta_g/2$.
To explain how they work, consider the example measurement circuit for an $\ssX$-type stabilizer generator as shown in Figure~\ref{fig:syndrome-extraction}.
An $\ssX$ error on the ancilla can propagate to a much larger data error. 

In theory, the hook error is $O(1)$-sized and the syndrome extraction circuit is fault-tolerant.
However, addressing these errors can significantly reduce the size of the LDPC code required to achieve a target logical failure rate.

Suppose ancilla qubits fail with probability $p$.
A measurement circuit for a weight $w$ operator can create hook errors with weight ranging from $1$ up to $w$.
If the circuit is measuring the checks of a code, the weight of the hook error can be reduced by using the measured stabilizer generator giving a maximum reduced weight of $\floor{w/2}$.
An error is uncorrectable if it has weight at least $\ceil{d/2}$.
If we assume that each ancilla failure results in an error of weight $\lceil \Delta_g/2\rceil$, then we only need $t$ failures to cause an uncorrectable error, where $t$ satisfies
\begin{align}
    t \cdot \floor{\Delta_g/2} \ge \ceil{d/2}~.
\end{align}
Then the probability of logical failure is $p^{t}$ where $t \approx d/\Delta_g$. 

This assumption is conservative --- hook errors depend on the choice of syndrome-extraction circuit and may be minimized by particular choices of gate scheduling.
For example, in the rotated surface code, there is a two-qubit gate schedule for the syndrome-extraction circuit such that the hook error has intersection-1 with a logical operator \cite{tomita2014low, chamberland2018flag}.
Using such a schedule, the below-threshold scaling is $\propto p^{\ceil{d/2}}$ as one would expect from a depolarizing noise model.
For general LDPC codes, the existence of measurement schedules that reduce the effects of hook errors is not yet clear.

While there exist many methods for %
suppressing hook errors such as Shor, Steane or Knill error correction \cite{nielsen2002quantum}, nearly all require more ancilla qubits.
This presents a trade-off where, for a given number of qubits, either a larger block length code with correspondingly better parameters or a more resource-intensive syndrome-extraction circuit could be used.
In the setting of a constant-rate LDPC code, larger distances come with more logical qubits, so the lowest overhead solution is to use the naive syndrome-extraction circuit with as large a code as possible\footnote{For very resource-constrained settings, it may still be worthwhile to use more sophisticated syndrome extraction circuits for a better effective relative distance.}.
Later, in Section~\ref{subsec:biased-noise}, we will propose a method to mitigate the effects of hook errors outside of the asymptotic regime.

\subsubsection{Decoder performance for the outer code}
\label{subsec:decoder_perf_ldpc}
Hook errors discussed in the previous section need to be considered in the context of the concatenated scheme --- the probability that an ancilla qubit fails is $p^{(1)}$.

Following the discussion in Section~\ref{subsec:hook_errors} on hook errors, $\ceil{\ceil{d/2}/\floor{\Delta_g/2}} = 10$ for the code selected in Section \ref{subsec:setup-num-est}.
Assuming a threshold of about $10^{-3}$ under circuit noise, the WER under circuit-level Pauli noise using long-range gates goes as
\begin{equation}\label{eq:ldpc_noise_model}
    \left(\frac{p^{(1)}}{10^{-3}}\right)^{10}~.
\end{equation}
per cycle of syndrome extraction.

For context, a slightly better threshold of about $3\times 10^{-3}$ has been observed for $(3,4)$ hypergraph product codes using efficient decoders~\cite{tremblay2021constant} under circuit-level noise for syndrome-extraction circuits with long-range gates.

In practice, more information is available to the decoder owing to the concatenated structure: A decoder using this extra information about individual qubit reliability is likely to have a better threshold.
We return to this subject in Section~\ref{subsec:concat-decoder}.

\subsection{Results}
\label{subsec:numerical_estimates_results}

\begin{figure}[h]
\includegraphics[width=\textwidth]{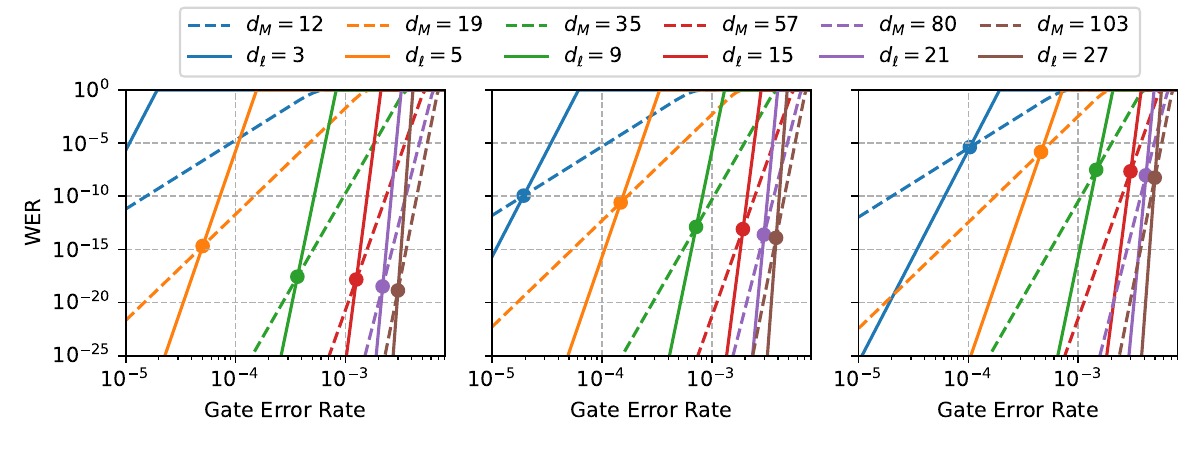}
\caption{\label{fig:comparison_swap_speed} Comparison of a hierarchical memory (solid line) using a $(5,8)$ quantum expander code with parameters 
    $\dsl 1\,116\,416, 112\,896, 119\dsr$ and inner code distance $d_{\ell}$, and a surface code (dashed line) with distance $d_M$.
    Lines of the same color use roughly the same number of physical qubits including all necessary ancilla qubits.
    All memories store $112\,896$ logical qubits.
    The $3$ plots correspond to different values of $r_{\swapp}$ equal to $10^0$, $10^{-1}$, and $10^{-2}$ (left to right) under the decoder performance assumptions made in Section \ref{subsec:setup-num-est}.
    The surface code distance is rounded up, so it always uses slightly more qubits.
    The WER is with respect to the hierarchical memory syndrome extraction cycle.
   }
\end{figure}

Using the duration of the hierarchical code syndrome extraction cycle as the unit of time that defines the WER, the results of the estimates are shown in Figure~\ref{fig:comparison_swap_speed} for $r_{\swapp} =$ $10^0$, $10^{-1}$, or $10^{-2}$ and several sizes of inner rotated surface code.
We can see the better scaling with gate error rate that the hierarchical memory achieves.
While the LDPC code distance is fairly large, the ``effective'' distance has been reduced immensely by the weight-6 hook errors (potentially arising from the measurement of weight-13 check operators); because the outer code has distance $d=119$, under our pessimistic assumptions just 10 fault locations are sufficient to cause an uncorrectable error. 
We expect future LDPC codes with better distance and better understanding of hook errors in syndrome extraction circuit gate scheduling will improve the WER scaling.

\begin{figure}[h]
\begin{center}
    \includegraphics{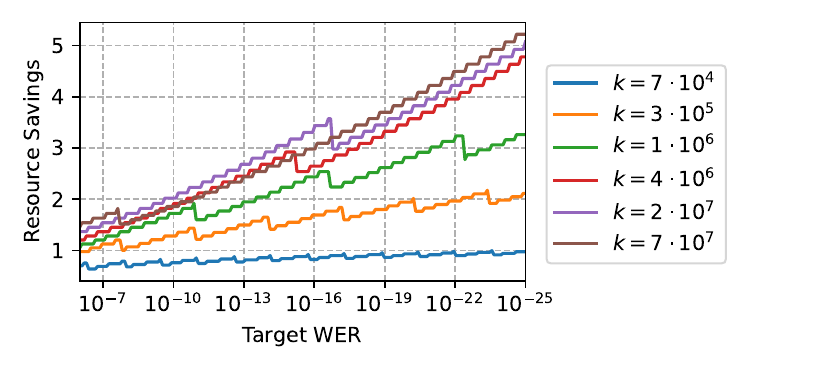}
\end{center}
\caption{\label{fig:wer_resource_comparison} 
    Estimated resource savings over surface codes for a hierarchical memory for $r_{\swapp}=10^{-1}$ and a gate error rate of $3\cdot 10^{-3}$ under the performance assumption of Section \ref{subsec:setup-num-est}.
    The resource here refers to the total space footprint of the circuits; the $y$-axis represents the ratio $\Space(\Cplain_M)/\Space(\CFT_N)$.
    We plot the resource savings for the $(4,8)$ family of quantum expander codes with input code block lengths $512\cdot 2^m$ for $m\in \{0,1,2,3,4,5\}$.
    The number of logical qubits is indicated in the legend.
    Discontinuities in the plot are due to discretization of the surface code distance.
    Rare noise sources that create high weight errors may provide further resources savings over surface codes.
   }
\end{figure}

Under a standard circuit-level noise model, below a gate error rate of around $10^{-3}$, and for a target WER of $10^{-20}$ to $10^{-10}$, the hierarchical scheme may realize significant resource savings.
Especially so if $\swapp$ gates have much lower gate error rates than $\cnot$ gates. 
We plot such a comparison in Figure~\ref{fig:wer_resource_comparison} with a gate error rate of $3\times 10^{-3}$ ($99.7\%$ gate fidelity) and $r_{\swapp} = 0.1$
With further engineering and more careful modeling, we believe the overhead of the hierarchical scheme can be reduced much more, so that the crossover point occurs at a practically relevant gate error rate and target WER.
In the next section we will outline two ideas that will improve the performance of the hierarchical scheme: The decoder for the outer code is given far more information about the level-1 qubit reliabilities than in a circuit level noise model, and in the presence of noise bias, the syndrome extraction circuit can be tailored to reduce the effects of hook errors.

Another reason we expect this estimate is conservative is that we have assumed that the noise model is independent circuit noise which creates only 2-body correlated errors in the underling surface codes.
In the setting of large, long-lived quantum memories, we expect it will become necessary to address noise sources that affect large patches of the system.
Sources of such noise could include cosmic rays (superconducting qubits), large deviations in global control devices such as lasers (AMO systems), lightning strikes, power supply ripples, etc. %
For large memories, different parts of the memory may rely on systems operating independently (ex. lasers, fridges, power supplies, etc) which would make such ``global'' noise large on the scale of any reasonable surface code patch, but small on the scale of the full hierarchical memory.
Concatenation of surface codes with constant length outer codes \cite{xu2022distributed} has previously been considered in order to address such issues.
It may be practical\footnote{Physics is local, so a very large surface code is likely sufficient, but it may be impractically large.} to protect against such noise sources with a hierarchical scheme without additional overhead.

\subsection{Future Performance Improvements}
\label{sec:future}
Having concluded a rough estimate of what the performance of the hierarchical memory might look like, we outline some ideas that could further improve the performance of the hierarchical memory relative to surface codes.
In this section, we re-examine the WER for LDPC codes using biased-noise qubits and message-passing decoders.

\subsubsection{Noise-Bias Tailored Syndrome Extraction}
\label{subsec:biased-noise}
As discussed in Section~\ref{subsec:hook_errors}, hook errors can be very damaging for general LDPC codes.
In this section, we estimate the failure rate for hierarchical codes by making further assumptions on the dependence of logical failure given $\eta$-biased qubits.
In particular, Equation~\eqref{eq:ldpc_noise_model_biased} presented below is an ansatz for the logical failure probability $p_{\cQ}$ of the outer code.
However, we expect that this estimate can be considerably improved in the future by investigating in more detail how $p_{\cQ}$ and depends on the bias $\eta$.

$\ssX$ errors on the ancilla qubit will propagate to an $\ssX$ or $\ssZ$ error on the data while $\ssZ$ errors on the ancilla qubit will simply flip the measurement outcome without propagating to a higher weight data error.
If $\ssX$ errors can be suppressed on the ancilla qubits, then hook errors become much less likely.
In many platforms, such noise is common or can be engineered into the experiment~\cite{grimm2020stabilization,lescanne2020exponential,cong2022hardware}.
Noise bias has been exploited in the past by tailoring the quantum error correction scheme \cite{aliferis2008fault,webster2015reducing,puri2020bias,bonilla2021xzzx,roffe2022bias} to the noise.

In Section~\ref{subsec:bilayer-biased}, we introduced a technique to modify the bilayer architecture such that Level-1 qubits are noise biased.
We can use this noise bias to suppress errors on the ancilla ($\ssX$) that propagate to higher weight data errors.
We modify the assumptions of Subsection~\ref{subsec:decoder_perf_ldpc} and Equation~\eqref{eq:ldpc_noise_model} in a way that attempts to capture this behavior.
Further study will be needed to make more precise estimates of logical error rates in this modified architecture. 

The modified bilayer architecture uses elongated Level-1 qubits. If we choose the $\ssX$ distance to be larger than the $\ssZ$ distance according to $d_{\ssX}= d_{\ssZ} + \lceil  2\log(\eta)/\log(1/p)\rceil$, then the logical $\ssX$ error rate of the inner code is suppressed relative to the logical $\ssZ$ error rate by the bias factor $1/\eta$. If the accuracy threshold of the outer code is still $10^{-3}$ as we assumed for the case without noise bias, then for the modified architecture our estimate 
for the Level-2 WER becomes 
\begin{equation}\label{eq:ldpc_noise_model_biased}
    p_{\cQ} =
    \left(\frac{p^{(1)}}{10^{-3}}\right)^{\ceil{d/2}} +
    \left(\frac{p^{(1)}/\eta}{10^{-3}}\right)^{\ceil{\ceil{d/2}/\floor{\Delta_g/2}}}~.
\end{equation}
The first term is the contribution from Level-1 logical $\ssZ$ errors; these do not propagate from ancilla to data, so that  $\ceil{d/2}$ Level-1 errors are needed to cause a logical error at level 2. The second term arises from Level-1 logical $\ssX$ errors. These can propagate from ancilla to data, but they occur at a rate suppressed by the bias factor $1/\eta$.

Since surface codes are CSS codes, the $\ssX$ and $\ssZ$ noise can be corrected independently, so the $\ssX$ and $\ssZ$ logical failure rates can be examined independently up to small correlations introduced by $\ssY$-errors.
Ignoring theses correlations and assuming that Equation~\eqref{eq:fmmc} still holds with $d$ replaced by $d_\ssX$ or $d_\ssZ$, 

\begin{figure}[h]
\includegraphics[width=\textwidth]{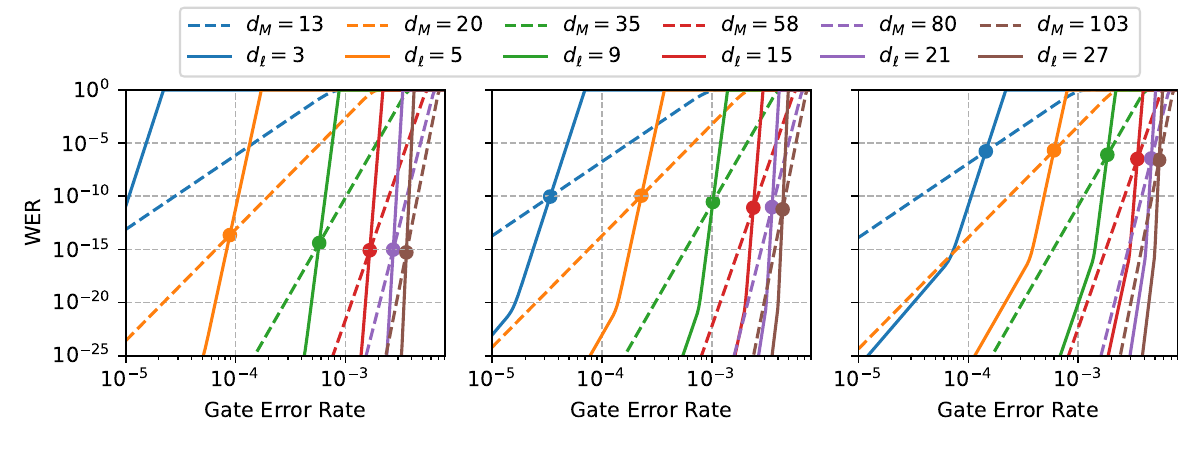}
\caption{\label{fig:comparison_bias}
    Comparison of a hierarchical memory (solid line) using a $(4,8)$ quantum expander code with parameters 
    $\dsl 327\, 680, 65\, 536, 32\dsr$ and inner code distance $d_{\ssZ} = d_{\ell}$, and a surface code (dashed line) with distance $\ell_M$. 
    Lines of the same color use roughly the same number of physical qubits including all necessary ancilla qubits.
    The noise bias permits a smaller block length, so all memories store $65\, 536$ logical qubits.
    The $3$ plots correspond to different values of $r_{\swapp}$ equal to $10^0$, $10^{-1}$, and $10^{-2}$ (left to right) under the modified decoder performance assumptions made in subsection \ref{subsec:biased-noise}. 
    The surface codes underlying the hierarchical memory are rectangular with $d_\ssX = 2d_{\ell}+1$
    The WER is for one round of the hierarchical memory's syndrome-extraction cycle.
   }
\end{figure}

We plot a similar comparison to Figure~\ref{fig:comparison_swap_speed} with $d_\ssX = 2d_\ssZ + 1$, so that $\ceil{d_\ssX/2} = 2 \ceil{d_\ssZ/2}$\footnote{This choice is somewhat arbitrary. For a given WER target and gate error rate, the optimal aspect ratio is likely to be such that the target is at the ``kink'' of the WER in Figure \ref{fig:comparison_bias}. } (Figure \ref{fig:comparison_bias}).
Using the greater resilience to hook errors, we also pick a smaller code with higher rate %
with parameters $\dsl 327\, 680, 65\, 536, 32\dsr$.
Notice the effect of the bias is to increase the rate at which the WER falls with the level-0 gate error rate.
The increased slope only persists until the two terms in Equation~\eqref{eq:ldpc_noise_model_biased} become equal.
One can see that using the bias, the hook errors are greatly suppressed leading to a better logical failure rate scaling in practically relevant regimes and a crossover point at a larger physical gate error rate.

\subsubsection{Decoders that use the concatenated structure}
\label{subsec:concat-decoder}

Our asymptotic analysis used the underlying surface codes in a black-box manner---when decoding the outer LDPC code $\{\cQ_n\}$, the tiles had either failed or succeeded.
In contrast to this ``hard information'', much more information is available to the decoder for the outer code in the hierarchical setting.
We may have access to ``soft information'', i.e.\ information about how reliable individual surface code patches are, which can then be passed to the outer code decoder.
It is known that maximum-likelihood decoding on each level of a concatenated code, together with message passing between levels is an optimal decoding algorithm \cite{poulin2006optimal}.

\textbf{Choice of outer code decoder:} Soft information can be used in the quantum setting using Belief Propagation (BP), a class of iterative algorithms.
Broadly, in each iteration, BP makes a series of graph-local decisions---qubits that are in the support of a stabilizer generator exchange information and update their beliefs about whether they have been corrupted.
As there are only a constant number of qubits in the support of each stabilizer generator, the decision requires a constant-sized computation.
Although it is very successful in the classical setting, BP faces difficulties when applied to quantum codes.
In the classical setting, BP converges to an distribution over bits that corresponds to the most likely error.
In the quantum setting, it was pointed out early on \cite{poulin2008iterative} that degeneracy is a major issue for BP---there are many errors that are equivalent as they differ only by a stabilizer generator.
However, BP is unable to tell the difference and gets stuck in a local minimum.
One simple way to get around this issue would be if more information were available about the qubits.
If each qubit were known to fail with a different probability -- even if that difference is small --- it can help BP avoid local minima.

Since then, many ideas have been developed to use soft information in the quantum setting that overcome the shortcomings of BP \cite{panteleev2021degenerate,roffe2020decoding,quintavalle2021single,grospellier2021combining,kuo2022exploiting,liu2019neural,du2022stabilizer}.
We now discuss ways to obtain soft information from the surface code.

\textbf{Choice of inner code decoder:} The tensor network decoders \cite{bravyi2014efficient,chubb2021general,tuckett2018ultrahigh,bonilla2021xzzx,tuckett2019tailoring} are one class of surface code decoders that yield such soft information.
The decoder outputs the probability of different (coset) logical failures for each tile.
Unfortunately, it is unclear how to implement these algorithms in the fault-tolerant setting where syndrome information is unreliable.
This setting requires growing bond dimension which makes implementing the decoder quite challenging.

More recently, BP decoders have been implemented for surface codes \cite{roffe2020decoding,old2022generalized,acharya2022suppressing}.
It is conceivable that such an algorithm could serve as a soft decoder for the surface codes as well.

A natural question is whether standard decoders such as Min-Weight Perfect Matching (MWPM) \cite{dennis2002topological} or the Union-Find Decoder (UFD) \cite{delfosse2017almost} could be modified to yield soft information.
For simplicity, consider bit flip noise at a rate of $p$.
We define the \emph{decoding graph} given by associating a vertex with each measured stabilizer generator.
We add a special \emph{boundary vertex} to which we associate the total parity of all measured stabilizer generators.
Including the boundary vertex, each single qubit error is detected in exactly two places.
For each error, an edge is added between the vertices where it is detected.
To each edge, assign the weight $-\log \left( \frac{p}{1-p} \right)$ which is the log-likelihood of an error.
The most likely error given the syndrome is then a subset of edges with minimal weight that produces the syndrome and can be computed efficiently by mapping onto the minimum-weight perfect matching problem.

On average, the expected weight of an error (and correction) will be linear in the block length.
This is asymptotically larger than the distance of a surface code, so the most important feature of the correction is its \emph{shape}.
The Union-Find Decoder operates in two steps: First, it identifies clusters such that a valid correction is contained within the support of the clusters.
Then, it treats the identified clusters as an erasure and runs an erasure correction decoder which produces a valid correction contained within the erasure. %

One such way to obtain soft information from this process is to compute the log-likelihood of the minimum weight error that would lead to a logical fault when combined with the erasure.
This can be computed efficiently by setting the edge weights within the erasure to $0$ and computing the minimal weight path between inequivalent boundaries.
Call this quantity $\phi$.
We note that when no errors are detected, $\phi = - d \log \left( \frac{p}{1-p}\right)$, and when the cluster spans the system, $\phi = 0$.
In the first case, it is extremely unlikely $(\propto p^d)$ for a logical fault to have occurred while in the latter case, there is a 50\% probability for a logical fault to have occurred.
When passed to an outer-level decoder, $\phi$ or a monotonic function of $\phi$ may yield sufficient information to improve the logical failure rate dramatically.

%% file: appendix.tex
\appendix 

\section{Glossary}
\label{app:glossary}

\begin{enumerate}
    \item Set notation: for natural numbers $n \in \bbN$, $[n] = \{1,...,n\}$.
    \item Sums over sets: For a set $S$ and a subset $A \subseteq S$, the sum $\sum_{B \supseteq A} f(B)$ is taken over all subsets $B\subseteq S$ such that $A \subseteq B$.
    \item Asymptotics: for functions $f, g: \bbN \to \bbR$, we say
    \begin{enumerate}
        \item $f(n) = O(g(n))$ if there exists an $n_0 \in \bbN$ and a positive number $c$ independent of $n$ such that for all $n > n_0$, $f(n) \leq g(n)$.
        \item $f(n) = \Omega(g(n))$ if $g(n) = O(f(g))$.
        \item $f(n) = \Theta(g(n))$ if there exists an $n_0 \in \bbN$ and positive numbers $a$, $b$ independent of $n$ such that $a \cdot g(n) \leq f(n) \leq b \cdot g(n)$.
    \end{enumerate}
    We may use $O_p(\cdot)$, $\Omega_p(\cdot)$ and $\Theta_p(\cdot)$ to indicate that the numbers $a$, $b$ and $c$ may depend on some parameter $p$ pertinent to the problem at hand.
    \item (Circuit) Step: A single timestep in which each qubit may participate in only one gate.
    \item (Circuit) Stage: This refers to the time interval in the circuit $C_n^{\cQ}$ required to simulate one entangling gate.
    One stage has at most $\dperm$ steps.
    \item (Measurement) Round: A complete measurement of all the stabilizer generators of the code producing one outcome for each stabilizer generator.
    \item $\cK$ is the set of Clifford operations we use to construct syndrome-extraction circuits in $2$ dimensions.
    It includes the following elements.
    \begin{enumerate}
      \item Initialization of new qubits in state $\ket{0}$ or $\ket{+}$,
      \item Single-qubit Pauli gates,
      \item Two-qubit Clifford gates $\cnot$ and $\cz$ between nearest-neighbor qubits,
      \item Single-qubit Pauli $\ssX$ and $\ssZ$ measurements,
      \item Physical $\swapp$ operation with range $R$.
    \end{enumerate}
    \item In the context of concatenated codes, $\cK$ carries subscripts $\cK_0$, $\cK_1$ to refer to Level-0 (physical) and Level-1 (logical) Clifford operations.
\end{enumerate}

\section[]{Constructing the ideal syndrome-extraction circuit $(C_n^{\cQ})^{\mathrm{ideal}}$}
\label{app:proof-ideal-sec}

In this section, we return to the claim in Section~\ref{subsec:sec-from-routing}.
We prove that the syndrome-extraction circuit $(C_n^{\cQ})^{\mathrm{ideal}}$ for a $\dsl n,k,d, \Delta_q,\Delta_g \dsr$ code can be constructed such that its depth is at most $s := 2\Delta + 4$, where $\Delta = 2\max(\Delta_q,\Delta_g)$.
\begin{proof}
  By definition, each qubit participates in at most $\Delta_{q}$ stabilizer generators and each stabilizer generator contains at most $\Delta_{g}$ qubits in its support.
  We use the \emph{Tanner graph} $\cT(\cQ_n) = (V \union C^{\ssX} \union C^{\ssZ}, E)$, a tripartite graph corresponding to the code $\cQ_n$ where:
    \begin{enumerate}
      \item There is a vertex $v \in V$ for each qubit in the code. $|V| = n$.
      \item There is a vertex $u^{\ssX}_i \in $ for each $\ssX$-type generator $\ssS^{\ssX}_i$. $|C^{\ssX}| = m_{\ssX}$.
      \item There is a vertex $w^{\ssZ}_j$ for each $\ssZ$-type generator $\ssS^{\ssZ}_j$. $|C^{\ssZ}| = m_{\ssZ}$.
    \end{enumerate}

  Consider the bipartite Tanner graph $\cT^{\ssX} = (V \union C^{\ssX}, E)$ that corresponds to the $\ssX$-type generators of the code $\cQ$.
  
  In each step, each qubit can be involved in at most one gate.
  This can be phrased as a graph coloring problem: we color the edges of $\cT^{\ssX}$ such that no two edges incident to a vertex have the same color.
  Since $\cT^{\ssX}$ is bipartite, such an edge coloring can be computed efficiently using $\max(\Delta_q,\Delta_g)$ colors \cite{schrijver2003combinatorial}.

  To measure the $\ssX$-type syndromes, the first phase of the circuit $(C_n^{\cQ})^{\mathrm{ideal}}$ is partitioned into $\max(\Delta_q,\Delta_g)$ steps.
  In the $t$\textsuperscript{th} step, we perform the two-qubit gates corresponding to the edge color $t$.

  Once completed, the same process is repeated for the $\ssZ$-type syndromes.
  Following a similar line of reasoning, this requires $\Delta = \max(\Delta_q,\Delta_g)$ applications of two-qubit gates.

  The circuit thus has two phases: first the $\ssX$-type syndromes are measured followed by the $\ssZ$-type syndromes~\footnote{This is unlike the surface code where both types of syndromes are measured at once.} which completes a measurement of all stabilizer generators.

  The total number of entangling stages is therefore $2\Delta$, where $\Delta = \max(\Delta_q,\Delta_g)$.
  Accounting for one stage for preparing and measuring ancilla qubits in each phase, we have a total of $s = 2\Delta + 4$ stages to measure syndromes.
\end{proof}